\documentclass[12pt]{article}
\usepackage{graphicx}
\usepackage{caption}
\usepackage{subcaption}
\usepackage{amsfonts}
\usepackage{amsmath}
\usepackage{amstext}
\usepackage{tabularx}
\usepackage{mathrsfs}
\usepackage{natbib}
\usepackage{setspace}
\usepackage{amssymb}
\usepackage{color}
\usepackage{rotating}
\usepackage{multirow}
\usepackage{amsmath,amssymb,amsthm}
\usepackage{booktabs}
\usepackage{bm}
\usepackage{thrmappendix}
\usepackage{threeparttable}
\usepackage{longtable}

\theoremstyle{plain} 
\newtheorem{theorem}{Theorem}

\newtheorem{lemma}{Lemma}

\theoremstyle{definition}

\theoremstyle{plain} 
\newtheorem{remark}{Remark}

\newcommand{\tr}{^{\top}}

\def\X{{\mathbf{X}}}
\def\Z{{\mathbf{Z}}}
\def\E{{\mathbb{E}}}
\def\R{{\mathbb{R}}}
\def\P{{\mathbb{P}}}
\def\T{{\mathcal{T}}}
\def\1{{\mathbf{1}}}
\def \C{{\mathcal{C}}}

\def\zb{{\bm{\beta}}}

\def\F{\mathbb{F}}
\def\H{\mathbb{H}}

\DeclareMathOperator*{\argmin}{arg\,min}

\begin{document}

\title{On the unbiased asymptotic normality of quantile regression with fixed effects%
\thanks{%
The authors are grateful to the editor, Jianqing Fan, the associate editor, and two anonymous referees for their constructive comments and suggestions. In addition we thank Roger Koenker for useful comments and discussions regarding this paper. Computer programs to replicate the numerical analyses are available
from the authors. All the remaining errors are ours.} }
\author{Antonio F. Galvao\thanks{Department of Economics, University of Arizona. E-mail: \texttt{agalvao@email.arizona.edu}}
\and
Jiaying Gu\thanks{Department of Economics, University of Toronto. E-mail: \texttt{jiaying.gu@utoronto.ca}}
\and
Stanislav Volgushev\thanks{Department of Statistical Sciences, University of Toronto. E-mail: \texttt{volgushe@utstat.toronto.edu} }
}
\maketitle

\begin{abstract}
\begin{spacing}{1}
Nonlinear panel data models with fixed individual effects provide an important set of tools for describing microeconometric data. In a large class of such models (including probit, proportional hazard and quantile regression to name just a few) it is impossible to difference out the individual effects, and inference is usually justified in a `large $n$ large $T$' asymptotic framework. However, there is a considerable gap in the type of assumptions that are currently imposed in models with smooth score functions (such as probit, and proportional hazard) and quantile regression. In the present paper we show that this gap can be bridged and
establish unbiased asymptotic normality for fixed effects quantile regression panels under conditions on $n,T$ that are very close to what is typically assumed in standard nonlinear panels. Our results considerably improve upon existing theory and show that quantile regression is applicable to the same type of panel data (in terms of $n,T$) as other commonly used nonlinear panel data models. Thorough numerical experiments confirm our theoretical findings.    
\end{spacing}
\end{abstract}

\thispagestyle{empty}

\newline


\newpage

\baselineskip18.9pt \pagenumbering{arabic}

\section{Introduction}

Nonlinear panel data models are important in microeconometric applications. They provide flexible modeling tools and allow researchers to account for unobserved heterogeneity through individual specific effects. Examples of such models include, among others, probit, logit, Poisson, negative binomial, proportional hazard, tobit, and quantile regression. Nevertheless, incorporating individual-specific fixed effects (FE) into these frameworks results in models whose dimension depends on the number of individuals, and thus grows with sample size; this raises some important issues for the asymptotic analysis of the FE estimators. As first noted by \citet{NeymanScott48}, leaving the individual heterogeneity unrestricted in a nonlinear or dynamic panel model can result in inconsistent estimators of the common parameters due to the incidental parameters problem. Consistency can be recovered in settings where it is possible to remove individual-specific effects by a transformation or when the number of individuals, $n$, and the number of time periods, $T$, jointly go to infinity.  


The nonlinear panel data literature has shown that $n/T \to 0$ is a sufficient condition to obtain asymptotic (unbiased) normality of nonlinear panel data FE estimators under smoothness conditions on the objective function.\footnote{Under an asymptotic framework where both the numbers of individuals and time periods grow at the same rate, it is possible to show that the fixed effects estimator for smooth objective functions has a limiting normal distribution with non-zero mean.} We refer the reader to \cite{ArellanoHahn07}, \cite{ArellanoBonhomme11}, and \cite{Fernandez-ValWeidner17} for reviews of the literature.\footnote{Important work in this literature include, among others, \cite{Kiviet95}, \cite{PhillipsMoon99}, \cite{Lancaster00}, \citet{HahnKuersteiner02}, \cite{Lancaster02}, \cite{Arellano03a}, \cite{AlvarezArellano03}, \cite{LiLindsayWaterman03}, \cite{HahnNewey04}, \cite{Woutersen04}, \cite{Carro07}, \citet{BesterHansen09}, \cite{Fernandez-Val09}, \citet{HahnKuersteiner11}, \cite{Fernandez-ValLee13}, \citet{DhaeneJochmans15}, \cite{Fernandez-ValWeidner16}.}

An important class of nonlinear panel data models is quantile regression (QR). Panel QR models have provided a valuable method of statistical analysis and recent examples of their application include instrumental variables models [e.g., \citet{ChetverikovLarsenPalmer16}, \cite{Galvao11}, and \cite{HardingLamarche09}], nonseparable models with time homogeneity [\cite{ChernozhukovFernandez-ValHoderleinHolzmannNewey15}, and \cite{ChernozhukovFernandez-ValHahnNewey13}], censored regression models [\cite{GalvaoLamarcheLima2013}, and \cite{WangFygenson09}], nonlinear models [\cite{ArellanoBonhomme16}],  interactive effects [\cite{HardingLamarche14}], and growth charts [\cite{WeiHe06}], amongst many others.\footnote{For other recent developments in panel data QR, see e.g., \cite{Canay11},  \cite{GrahamHahnPoirierPowell2015}, \cite{SuHoshino16}, \cite{GalvaoKato16}, and \cite{HardingLamarche17}. \cite{GalvaoKato17}  review the QR methods for panel (longitudinal) data.} 
Unfortunately, as in most of the nonlinear panel data literature, transformations to remove the unobserved individual effects are not available for QR models. Thus, FE panel data QR suffers from the incidental parameters problem and large $n,T$ asymptotics must be employed in the analysis. \cite{KatoGalvaoMontes-Rojas12} analyze the standard fixed effects (FE-QR) estimator in which individual effects are introduced through dummy variables, and \cite{GalvaoWang15} consider the minimum distance (MD-QR) estimator. These studies establish the asymptotic properties of the corresponding estimators and derive sufficient conditions for their consistency and asymptotic normality. Both the FE-QR and MD-QR estimators are shown to be asymptotically normal under the stringent condition that $n^{2}(\log n)(\log T)^2/T\to0$, and $T \to \infty$ as $n \to \infty$. This requirement is much more restrictive than what is usually required in the standard literature on nonlinear panel data models under smoothness conditions. 

The substantial discrepancy between existing conditions that guarantee asymptotic normality of panel data QR and other nonlinear panel models gives rise to the following question: are the restrictive conditions imposed in the QR literature really necessary, or are these conditions rather an artifact of the proof techniques employed so far?\footnote{The source of this stronger requirement is the lack of smoothness of the QR objective function, which implies that classical higher-order expansions are not valid.} Answering this question is of central importance to the panel data QR literature and will have profound implications for the recommendation that econometricians can give to practitioners when it comes to the set of tools that should be used in the analysis of panel data with moderate length. If indeed the conditions required for QR turn out to be substantially more restrictive compared to other nonlinear models, it would substantially limit the application of QR to panels that have only modest time dimension compared to the cross sectional dimension.







The main contribution of this paper is to provide an answer to the question posed above. We prove that unbiased asymptotic normality of both the FE-QR and MD-QR estimators continues to hold provided that $n(\log T)^{2}/T \to 0$, as $(n,T) \to \infty$. This significantly improves upon the previous condition available in the literature, $n^{2}(\log n)(\log T)^2/T\rightarrow0$, and shows that QR is applicable to the same type of panels as other nonlinear models. In addition, we provide a complete set of corresponding results under conditions that allow for temporal dependence of observations within individuals, thus encompassing a large class of possible empirical applications.

While the results that we obtain are close to what is established in most of the nonlinear panel data literature under smooth conditions, we would like to emphasize that the method of proof is rather different. The main difficulty stems from the non-smooth objective function which renders most of the usual techniques (which crucially rely on Taylor expansions of the objective function) inapplicable. A key insight in the proofs is a detailed analysis of the \textit{expected values} of remainder terms in the classical Bahadur representation for QR, while previous approaches (including \cite{KatoGalvaoMontes-Rojas12} and \cite{GalvaoWang15}) focused on the \textit{stochastic order} of those remainder terms. A similar analysis was previously performed in \cite{VCC} for general QR models with growing dimension under the assumption of independent and identically distributed observations. The proofs involve subtle empirical process arguments, and extending those results to settings with dependent data requires a substantial amount of work. We further note that the results in \cite{VCC} are not directly applicable to either of the estimators we consider here. The FE-QR estimator can not be expressed as an average at all, and several intermediate steps in the Bahadur representation are needed before the ideas from \cite{VCC} become applicable. The MD-QR estimator is itself an average, but a complication in the corresponding analysis is due to the estimated weights, which also depend on estimates of the quantile coefficients in a non-smooth way. 

Finally, we conduct Monte Carlo simulations to study the finite sample properties of the FE-QR and MD-QR estimators and verify our theory. The numerical experiments confirm the theoretical findings. 

The rest of the paper is organized as follows. Section~\ref{sec:mod} describes the QR model and the estimators. Section~\ref{sec:theo} contains the asymptotic results, theory in the independent case is presented in Section~\ref{sec:indep} while Section~\ref{sec:dep} considers extensions to temporal dependence within individuals. In Section~\ref{sec:sim} we present the Monte Carlo experiments and conclude in Section~\ref{sec:conc}. All proofs are collected in the Appendix.

\section{The model and estimators}\label{sec:mod}
\subsection{The model}
Assume that we have observations $(Y_{it},\X_{it})_{i=1,..,n, t=1,...,T}$ where $Y_{it}$ denotes the response variable for individual $i$ at observation period $t$, and $\X_{it}$ denotes the corresponding $p$-dimensional vector of explanatory variables. In this paper, we use a quantile regression (QR) panel data model with fixed effects (FE) to describe the influence of $\X_{it}$ on the response $Y_{it}$. This model takes the form
\begin{equation}
\label{eq:model}
q_{i,\tau}(\X_{it})=\alpha_{i0}(\tau)+\X_{it}\tr\bm{\beta}_0(\tau), \; t=1,...T, \; i=1,...,n,
\end{equation}
where $q_{i,\tau}(\X_{it})=\inf\{Y:\text{Pr}(Y_{it}<Y|\X_{it})\geq\tau\}$ is the conditional $\tau$-quantile of $Y_{it}$ given $\X_{it}$. Here $\bm{\beta}_0(\tau)$ denotes the vector of common slope coefficients while $\alpha_{i0}(\tau)$ is a quantile-specific individual effect, which is intended to capture individual specific sources of variability, or unobserved heterogeneity that was not adequately controlled by other covariates in the model. In what follows, we will also use the notation $\bm{\gamma}_{i0}(\tau)=(\alpha_{i0}(\tau),\bm{\beta}_{0}(\tau)\tr)\tr$ and $\Z_{it}\tr=(1,\X_{it}\tr)$. This model is semiparametric in the sense that the functional form of the conditional distribution of $Y_{it}$ given $\X_{it}$ is left unspecified and no parametric assumption is made on the relation between $\X_{it}$ and $\alpha_{i0}$. 

Next we describe two estimators for the parameters of interest, $\bm{\beta}_{0}(\tau)$, namely the standard fixed effects QR and the minimum distance QR.

\subsection{The estimators}
To estimate the panel QR model, \cite{Koenker04} considers the fixed effects quantile regression (FE-QR) estimator. This is an individual dummy variables estimator, which is a natural analog of the dummy variables estimator for the standard FE mean regression model. However, in contrast to mean regression, there is no general transformation that can suitably eliminate the individual specific effects in the FE-QR estimator, and hence one is required to deal directly with the full problem. The FE-QR is defined as follows\footnote{We follow \cite{Koenker04} in treating the $\alpha_i$ as fixed parameters. An alternative approach which leads to equivalent results is to treat the $\alpha_i$ as random (with no restrictions placed on the dependence with $X_{it}$). In this case the model can be written as $Q_{Y_{it}|X_{it},\alpha_i(\tau)}(\tau) = \beta_0(\tau)^\top x + \alpha_{0i}(\tau)$; here $Q_{Y_{it}|X_{it},\alpha_i(\tau)}$ denotes the conditional quantile function of $Y_{it}$ given $(X_{it},\alpha_i(\tau))$ (see for instance \cite{KatoGalvaoMontes-Rojas12}, \cite{GalvaoWang15} and \cite{GalvaoKato16} for this interpretation). Both interpretations lead to the same asymptotic results.}
\begin{equation}
\label{eq:feqr}
(\widehat{\bm{\alpha}}(\tau),\widehat{\bm{\beta}}(\tau)) := \argmin_{(\bm{\alpha},\bm{\beta}) \in \mathbb{R}^n \times \mathbb{R}^{p}} \frac{1}{nT} \sum_{i=1}^{n} \sum_{t=1}^{T} \rho_{\tau}( Y_{it}-\alpha_{i} - \X_{it}\tr \bm{\beta} ),
\end{equation}
where $\bm{\alpha} := (\alpha_{1},\dots,\alpha_{n})\tr$ and  $\rho_{\tau}(u) := \{ \tau - \bm{1}(u \leq 0) \} u$ is the check function as in \cite{KoenkerBassett78}. We refer to the FE-QR estimator defined in equation \eqref{eq:feqr} as $\widehat{\boldsymbol{\beta}}_{FE}(\tau)$.

In typical applications, the number of individuals can be large and the FE-QR estimator involves optimization with substantial number of parameters to be estimated, which makes the problem computationally cumbersome. Hence, motivated by the practical implementation challenges of the FE-QR, we also consider a simple to implement and efficient QR minimum distance (MD) estimator for panels with fixed effects to estimate the common parameter of interest $\bm{\beta}_{0}(\tau)$.\footnote{The MD estimation is a flexible methodology and has been largely applied to  panel data problems, examples, among others, include \cite{Swamy70}, \cite{Chamberlain82, Chamberlain84}, \cite{AhnSchmidt95}, \cite{HsiaoPesaranTahmiscioglu02}, \cite{Hsiao03}, \cite{Pesaran06}, \cite{LeeMoonWeidner12}, and \cite{MoonShumWeidner17}.}  As in \cite{GalvaoWang15}, we use a minimum distance quantile regression (MD-QR) estimator, $\widehat{\bm{\beta}}_{MD}^{\rm{inf}}(\tau)$, defined as follows
\begin{equation}
\label{eq:femd}
\widehat{\bm{\beta}}_{MD}^{\rm{inf}}(\tau)=\left(\sum_{i=1}^nW_{i}^{-1}\right)^{-1}\sum_{i=1}^nW_i^{-1}\widehat{\bm{\beta}}_{i}(\tau),
\end{equation}
where $\widehat{\bm{\beta}}_{i}(\tau)$ is the slope coefficient estimator from each individual QR problem using the time-series data, i.e. 
\begin{equation*}
\widehat{\bm{\gamma}}_{i}(\tau)=(\widehat{\alpha}_{i}(\tau),\widehat{\bm{\beta}}_{i}(\tau)) := \argmin_{(\alpha_{i},\bm{\beta}_{i}) \in \mathbb{R} \times \mathbb{R}^{p}} \frac{1}{T} \sum_{t=1}^{T} \rho_{\tau}( Y_{it}-\alpha_{i} - \X_{it}\tr \bm{\beta}_{i} ),
\end{equation*}
where $W_i$ denotes the associated variance-covariance matrix of $\widehat{\bm{\beta}}_{i}(\tau)$ for each individual.

However, in applications, the estimator $\widehat{\boldsymbol{\beta}}^{\rm{inf}}_{MD}$, defined in equation \eqref{eq:femd}, is infeasible unless $W_i$ is known for every individual. The feasible estimator is defined with each $W_i$ replaced by its corresponding consistent estimators $\widehat{W}_i$, such that the MD-QR is given by
\begin{equation}
\label{eq:fefmd}
\widehat{\bm{\beta}}_{MD}(\tau)=\left(\sum_{i=1}^n\widehat{W}_{iT}^{-1}\right)^{-1}\sum_{i=1}^n\widehat{W}_{iT}^{-1}\widehat{\bm{\beta}}_{i}(\tau).
\end{equation}
This feasible two-step estimator can be implemented by obtaining, in the first step, consistent estimates for the slope coefficients and their associated variance-covariance matrices, $\widehat{W}_{iT}$. One can obtain such estimates from the standard QR algorithm for each individual separately. In the second step the estimated $W_{i}$'s are substituted into \eqref{eq:femd} which results in \eqref{eq:fefmd}. The specific form of $\widehat{W}_{iT}$ depends on the assumption on the dependence across individuals, detailed expressions for such estimators will be provided later.

\section{Asymptotic theory}\label{sec:theo}

{In this section, we investigate the asymptotic properties of the FE-QR estimator in~\eqref{eq:feqr} and also the MD-QR estimator in~\eqref{eq:fefmd}. We begin by stating and discussing a set of basic assumptions that will be used throughout the remaining parts of this paper. In Section~\ref{sec:indep} we discuss the asymptotic distribution of the FE-QR and the feasible MD-QR estimators given in \eqref{eq:feqr} and \eqref{eq:fefmd}, respectively, under the additional assumption that data within individuals are independent. An extension of those results to the case where temporal dependence of observations within individuals is allowed is provided in Section~\ref{sec:dep}.   } 

\subsection{Basic assumptions}\label{sec:ass}

Throughout the paper, we use the following notations: for a square matrix $A$, let $\|A\|$ denote the maximal absolute eigenvalue of $A$, while for a vector $v$, $\|v\|$ denotes the usual Euclidean norm. Recall that $\Z_{it}^\top = (1, \X_{it}^\top)$ is a vector of dimension $p+1$ and let $\mathcal{Z}$ denote the support of $\Z_{it}$. We make the following assumptions. 

\begin{itemize}
\item[(A0)]\label{A0} The processes $\{ (Y_{it},\X_{it}) : t \in \mathbb{Z} \}$ are strictly stationary for each $i$ and independent across $i$. 
\item[(A1)]\label{A1} Assume that $\|\Z_{it}\| \leq M < \infty$ almost surely, and that $c_{\lambda}\leq\lambda_{\min}(\E [\Z_{it} \Z_{it}^\top])\leq\lambda_{\max}(\E [\Z_{it} \Z_{it}^\top])\leq C_{\lambda}$ holds uniformly in  $i$ for some fixed constants $c_{\lambda}>0$ and $C_{\lambda} <\infty$.
\item[(A2)]\label{A2} The conditional distribution $F_{Y_{i1}|\Z_{i1}}(y|\mathbf{z})$ is twice differentiable w.r.t. $y$, with the corresponding derivatives $f_{Y_{i1}|\Z_{i1}}(y|\mathbf{z})$ and $f'_{Y_{i1}|\Z_{i1}}(y|\mathbf{z})$. Assume that 
\begin{equation*}
f_{max} :=\sup_i \sup_{y \in \mathbb{R},\mathbf{z}\in \mathcal{Z}} |f_{Y_{i1}|\Z_{i1}}(y|\mathbf{z})| < \infty,
\end{equation*}
and 
\begin{equation*}
\overline{f'} := \sup_i \sup_{y \in \mathbb{R},\mathbf{z}\in \mathcal{Z}} |f'_{Y_{i1}|\Z_{i1}}(y|\mathbf{z})| < \infty.
\end{equation*}
\item[(A3)]\label{A3} Denote by $\mathcal{T}$ an open neighborhood of $\tau$. Assume that uniformly across $i$, there exists a constant $f_{\min} < f_{\max}$ such that
\begin{equation*}
0 < f_{\min} \leq \inf_i \inf_{\eta \in \mathcal{T}} \inf_{\mathbf{z} \in \mathcal{Z}} f_{Y_{i1}|\Z_{i1}}(q_{i,\eta}(\mathbf{z})|\mathbf{z}),
\end{equation*}
where $q_{i,\eta}(\mathbf{z})$ is the conditional $\eta$ quantile of $Y_{i1}$ given $\Z_{i1} = \mathbf{z}$. 
\end{itemize}

Condition (A0) assumes that the data are independent across individuals, and strictly stationary within each individual. Additional assumptions on the temporal dependence within individuals will be added later when we state the corresponding results. Condition (A1) poses a boundedness condition on the norm of the regressors, which is also standard in the literature, see for instance \cite{Koenker04, KatoGalvaoMontes-Rojas12, GalvaoWang15}. Condition (A1) also assures that the eigenvalues of $\E [\Z_{it} \Z_{it}^\top]$ are bounded away from zero and infinity uniformly across $i$, similar assumptions were made in \cite{CVC, BCCF}. Conditions (A2) and (A3) impose smoothness and boundedness of the conditional distribution, the density and its derivatives. The same type of assumption has been imposed in \cite{GalvaoWang15}. \cite{CVC} and \cite{BCCF} also make similar assumptions when deriving Bahadur representations for QR estimators in a setting without panel data.


\subsection{The independent case}\label{sec:indep}

The main contribution of this section is to provide new insights on the asymptotic properties of the FE-QR estimator in~\eqref{eq:feqr} and the feasible version of the MD-QR estimator in \eqref{eq:fefmd} under the assumption that data are independent and identically distributed (i.i.d.) within individuals. We are able to considerably improve the existing conditions on $n,T$, thus reconciling asymptotic results in the panel data QR and nonlinear panel literature. Throughout this section, we will use the following assumption.
\begin{enumerate}
\item[(I)]\label{A4n} Assume that for each $i$, the observations are $(\X_{it},Y_{it})_{t=1,...,T}$ are i.i.d. across $t$. Moreover, assume that $n \to \infty, T \to \infty$ and\footnote{The assumption $n \to \infty$ could be dropped at the cost of more complicated notation. Nevertheless, since this case can be handled by existing results we prefer to concentrate on the more complicated setting $n \to \infty$. In addition, we conjecture that the factor $(\log T)^2$ could be improved, but that seems impossible with our current method of proof. We leave such an improvement for future research.} 
\begin{equation*}
\frac{n (\log T)^2}{T} = o(1).
\end{equation*}
\end{enumerate}

Condition (I) excludes temporal dependence. In Section~\ref{sec:dep}, we extend the results to the dependent case under suitable mixing conditions as in \cite{HahnKuersteiner11}. We note that the independence assumption is used mainly to apply some standard stochastic inequalities; our results are extended in Section~\ref{sec:dep} to the dependent case by replacing these stochastic inequalities by those that hold under suitable dependence conditions.

We begin by discussing the FE-QR estimator. To this end let $f_i(y) := \E[f_{Y_{i1}|\X_{i1}}(y+q_{i,\tau}(\X_{i1})|\X_{i1})]$ denote the marginal density of $Y_{i1} - q_{i,\tau}(\X_{i1})$ and define
\begin{equation} \label{eq:def-gi}
g_i := \E[f_{Y_{i1}|\X_{i1}}(q_{i,\tau}(\X_{i1})|\X_{i1})\X_{i1}]/f_i(0).
\end{equation}
Let 
\begin{equation} \label{eq:def-gamman}
\Gamma_n := \frac{1}{n}\sum_{i=1}^n \E[f_{Y_{i1}|\X_{i1}}(q_{i,\tau}(\X_{i1})|\X_{i1})\X_{i1}(\X_{i1}-g_i)^\top].
\end{equation}
\begin{enumerate}
\item[(FI)] Assume that $\Gamma_n$ is non-singular for each $n$, that $\Gamma := \lim_{n \to \infty} \Gamma_n$ exists and is non-singular and that $L := \lim_{n \to \infty} n^{-1}\sum_{i=1}^n \E[(\X_{i1}-g_i)(\X_{i1}-g_i)^\top]$ exists and is non-singular. 
\end{enumerate}

\begin{theorem}\label{th1FE}
Assume that (A0)--(A3), (I), (FI) hold. Then, for fixed $\tau\in(0,1)$, we have that
\begin{equation}\label{eq:mainFE}
\sqrt{nT}(\widehat{\bm{\beta}}_{FE}(\tau) - \bm{\beta}_{0}(\tau)) \stackrel{\mathcal{D}}{\longrightarrow} \mathcal{N}\big(0,\tau(1-\tau)\Gamma^{-1}L\Gamma^{-1}\big).
\end{equation}
\end{theorem}

The limiting distribution above is the same as in Theorem~3.2 of \cite{KatoGalvaoMontes-Rojas12}. However, in contrast to the assumption $n^2(\log n)^3/T = o(1)$ in the latter reference we only require $n (\log T)^2/T = o(1)$. Some intuition on why such an improvement is possible is provided in Remark~\ref{rem:proofideasFE} below. Finally, note that consistent estimation of the asymptotic variance-covariance matrix $\tau(1-\tau)\Gamma^{-1}L\Gamma^{-1}$ was discussed in \cite{KatoGalvaoMontes-Rojas12}, see their Proposition~3.1. Since that result does not require a stringent growth condition on $n$ it continues to be applicable under our weaker assumptions.

We next discuss the MD-QR estimator. To complete the computation of the MD-QR estimator in \eqref{eq:fefmd} we need an estimate $\widehat{W}_{iT}$. Assuming that data within individuals are i.i.d., the asymptotic covariance matrix of $\widehat{\bm{\gamma}}_{i}(\tau)=(\widehat{\alpha}_{i}(\tau),\widehat{\bm{\beta}}_{i}(\tau))$ takes the form $V_i = B_i^{-1}A_iB_i^{-1}$ where
\begin{equation}
\label{eq:AiBi}
A_i := \tau (1-\tau) \E[\Z_{i1}\Z_{i1}^\top], \quad B_i = \E[f_{Y\mid \Z}(q_{i,\tau}(\Z_{i1})\mid \Z_{i1})\Z_{i1}\Z_{i1}^\top].
\end{equation}
{Note that $W_{i}$ is the lower $p\times p$ sub-matrix of $V_{i}$. Hence by estimating $V_{i}$ consistently we recover a consistent estimator of $W_{i}$.} A common way to estimate $V_{i}$ uses the Hendricks-Koenker sandwich covariance matrix estimator (\cite{HendricksKoenker91}) which takes the following form
\begin{equation}
\label{eq:HKvar}
\widehat{V}_{iT} := \widehat{B}_{iT}^{-1} \widehat{A}_{iT} \widehat{B}_{iT}^{-1},
\end{equation}
with 
\begin{align*}
\label{eq3}
\widehat{B}_{iT} & :=  \frac{1}{T} \sum_{t=1}^{T} \widehat{f}_{it} \Z_{it}\Z_{it}^\top, \quad \widehat{A}_{iT} := \tau (1-\tau)\frac{1}{T} \sum_{t=1}^{T}  \Z_{it}\Z_{it}^\top, \quad
\widehat{f}_{it}  := \frac{2d_T}{\Z_{it}^\top (\widehat{\bm{\gamma}}_{i}(\tau + d_T) - \widehat{\bm{\gamma}}_{i}(\tau - d_{T}))}.
\end{align*}
Here $d_T$ denotes a smoothing parameter that should converge to zero at an appropriate rate in order to guarantee consistency of $\widehat{B}_{iT}$, which is imposed in assumption (MI) below. 

Given the definitions above, we consider the following MD-QR estimator  
\begin{equation}
\label{eq4}
\widehat{\bm{\beta}}_{MD}(\tau)=\Big(\sum_{i=1}^n\widehat{W}_{iT}^{-1}\Big)^{-1}\sum_{i=1}^n\widehat{W}_{iT}^{-1}\widehat{\bm{\beta}}_{i}(\tau),
\end{equation}
where $\widehat{W}_{iT}$ is the lower $p\times p$ sub-matrix of $\widehat{V}_{iT}$ in \eqref{eq:HKvar}. Since consistency of a closely related estimator (the only difference is the form of the variance estimators $\widehat{W}_{iT}$) was established in \cite{GalvaoWang15}, we focus on conditions which ensure asymptotic normality. {In addition to (A0)-(A3), (I) we make the following assumption  
\begin{enumerate}
\item[(MI)]\label{A4} Assume that $d_T = o(1)$, as $T \to \infty$ and
\begin{equation*}
\frac{\log T}{T d_T^2} = o(1).
\end{equation*}
Let $W_i$ denote the lower $p\times p$ sub-matrix of $V_i$ and assume that $\left( \frac{1}{n}\sum_{i=1}^{n} W_i^{-1} \right)^{-1}$
is a well-defined positive definite matrix if $n$ is fixed and that $\Sigma_{1} := \lim_{n\to \infty} (\frac{1}{n}\sum_{i=1}^{n} W_i^{-1})^{-1}$ exists and is a well-defined positive definite matrix if $n \to \infty$.
\end{enumerate}
}

The asymptotic normality of the MD-QR estimator is formalized in the next theorem.

\begin{theorem}\label{th1}
Assume that (A0)--(A3), (I) and (MI) hold. In this case, for fixed $\tau$, we have that
\begin{equation}\label{eq:main}
\sqrt{nT}(\widehat{\bm{\beta}}_{MD}(\tau) - \bm{\beta}_{0}(\tau)) \stackrel{\mathcal{D}}{\longrightarrow} \mathcal{N}\big(0,\Sigma_{1}\big).
\end{equation}
\end{theorem}

\bigskip


An important by-product of the proof of Theorem \ref{th1} is that we are able to estimate the asymptotic variance-covariance matrix, $\Sigma_{1}$, consistently, uniformly over a growing number of individuals. More precisely, we obtain the following result. 

\begin{lemma}
\label{lemma:CovHK}
Under the conditions of Theorem \ref{th1}, we have $\max_{i=1,...,n}\|\widehat{B}_{iT} - B_{i} \| = o_P(1)$ and $\max_{i=1,...,n}\|\widehat{A}_{iT} - A_{i} \| = o_P(1)$. Moreover, $\widehat{\Sigma}_1 := (\frac{1}{n}\sum_{i=1}^n \widehat{W}_{iT}^{-1} )^{-1}$ is a consistent estimator for $\Sigma_1$.
\end{lemma}

The result in Lemma \ref{lemma:CovHK} allows for the construction of simple to implement inference procedures based on the Wald statistic.

\begin{remark}{\rm In the proof of Theorem \ref{th1} in the Appendix (see Lemma~\ref{lem:aitbit}), we also provide uniform convergence rates and Bahadur representations of the estimators $\widehat{A}_{iT},\widehat{B}_{iT}$. Those results directly yield a corresponding uniform rate and Bahadur representation for the Hendricks-Koenker covariance matrix estimator $\widehat{V}_{it}$. This result may be of independent interest. }
\end{remark}

The following remarks describe the intuition for the proofs of Theorems~\ref{th1FE} and \ref{th1}. Since the proof of Theorem \ref{th1FE} builds on intermediate results as well as intuition from the proof of Theorem~\ref{th1}, we begin by discussing some of the key ideas in the  proof of the latter result, which, to the best of our knowledge, are not present in the existing literature.


\begin{remark}[Intuition for proofs of the MD estimator] \label{rem:proofideas} {\rm 
Theorem~\ref{th1} shows asymptotic normality of $\widehat{\bm{\beta}}_{MD}(\tau)$ under the assumption $n (\log T)^2/T = o(1)$, which considerably relaxes previous conditions of the form $n^2 (\log n)(\log T)^2/T = o(1)$ (see \cite{GalvaoWang15}). The crucial insight for providing the improved condition comes from a closer analysis of the remainder term in the Bahadur representation for $\widehat{\zb}_{MD}$ 
\begin{equation}\label{eq:mdex}
\widehat{\bm{\beta}}_{MD}(\tau) - \bm{\beta}_{0}(\tau) = \frac{1}{n} \sum_{i=1}^{n} W_{i}^{-1} \frac{1}{T}\sum_{t=1}^{T} \phi_{i,\tau}(\Z_{it}, Y_{it}) + \frac{1}{n} \sum_{i=1}^{n} r_{n,i}(\tau) + o_p\Big(\frac{1}{\sqrt{nT}}\Big),
\end{equation}
where $\phi_{i,\tau}(\Z_{it}, Y_{it})$ denotes the last $p$ entries of $B_i^{-1} \Z_{it} (\1(Y_{it} \leq q_{i,\tau}(\Z_{it}))-\tau)$. Previous approaches bound the sum of remainder terms $\frac{1}{n} \sum_{i=1}^{n} r_{n,i}(\tau)$ by $\sup_i|r_{n,i}(\tau)|$, which is of order $O_P((\log T)^b/T^{3/4})$ for some constant $b$. The power of $T$ in this bound cannot be improved, which seems to suggest that in order to obtain unbiased asymptotic normality the condition $\sup_i|r_{n,i}(\tau)| = o_P((nT)^{-1/2})$ needs to be imposed. This gave rise to the condition $n^2 (\log n)(\log T)^2/T = o(1)$ in \cite{GalvaoWang15}. 

The intuitive reason for why this approach results in conditions that are too strong is that the remainder terms $r_{n,i}$ can be viewed as independent (across $i$) random vectors, and thus the order of their mean is governed by their expected values, while the variance is of the order $n^{-1}\sup_i Var(r_{n,i}(\tau)) = o((nT)^{-1/2})$. In order to make this intuition precise, we need to derive a more detailed expansion for $\widehat{\bm{\beta}}_{MD}(\tau) - \bm{\beta}_{0}(\tau)$ which allows to bound the expected values of the remainder terms. To achieve this, we start by deriving a Bahadur type expansion for the estimated weights which takes the form (see equation \eqref{Vit1} in the Appendix)
\begin{equation*}
\widehat{W}_{iT}^{-1} - W_{iT}^{-1} = \frac{1}{T} \sum_{t} \tilde{\eta}_{iT}(\Z_{it}, Y_{it}) + \tilde{R}_{3i}(\tau),
\end{equation*}
where $\tilde{\eta}_{iT}(\Z_{it}, Y_{it})$ are centered random variables with $(\tilde{\eta}_{iT}(\Z_{it}, Y_{it}))_{t=1,...,T}$ i.i.d. for all $i$ and the remainder terms $\tilde{R}_{3i}(\tau)$ are small in a suitable sense. A similar expansion is derived for the estimators $\widehat{\bm{\beta}}_i$ which is given by (see Lemma~\ref{VCClemma} in the Appendix)
\begin{equation*}
\widehat{\bm{\beta}}_{i}(\tau) - \bm{\beta}_{0}(\tau) = \frac{1}{T}\sum_{t=1}^{T} \phi_{iT}(\Z_{it}, Y_{it}) + R_{iT}^{(1)}(\tau) + R_{iT}^{(2)}(\tau),
\end{equation*}
where again $\phi_{iT}(\Z_{it}, Y_{it})$ are i.i.d. and centered, $R_{iT}^{(2)}(\tau)$ are `small' uniformly, and $\E[R_{iT}^{(1)}(\tau)]$ is `small'. Combining the two representations above yields the representation in \eqref{eq:mdex}, and a detailed analysis of all remainder terms shows that, in fact, $\sup_i \Big\| \E[r_{n,i}(\tau)] \Big\| = O(\log T/T)$, which gives rise to the less restrictive condition $\log T/T = o((nT)^{-1/2})$.

We conclude this remark by noting that the idea of analyzing expected values of remainder terms when aggregating QR from subsets was also used in \cite{VCC}. However, the setting we consider here is different from the latter paper since we also need to take into account the $\widehat{W}_{iT}^{-1}$ which also depend on $\widehat{\bm{\beta}}_{i}(\tau \pm d_T)$ in a non-smooth way. This considerably complicates the asymptotic analysis.   
}
\end{remark}

{
\begin{remark}[Intuition for proofs of the FE estimator] \label{rem:proofideasFE} {\rm Similarly to the discussion of the MD estimator given above, the main improvement in the conditions for the FE estimator comes from a closer analysis of the remainder terms in the Bahadur representation of $\widehat\zb_{FE}(\tau)$. However, the structure of the estimator is quite different and so a different analysis is required. More precisely, a closer look at the proof of Theorem~3.2 in \cite{KatoGalvaoMontes-Rojas12} shows that the remainder term which is responsible for the condition $n(\log n)^2/T = o(1)$ in the latter paper is of the form
\[
\frac{1}{n} \sum_{i=1}^n \H_{ni}^{(3)}(\widehat\alpha_i,\widehat\zb) - \H_{ni}^{(3)}(\alpha_{i0},\zb_0) - \Big\{H_{ni}^{(3)}(\widehat\alpha_{i},\widehat\zb) - H_{ni}^{(3)}(\alpha_{i0},\zb_0)\Big\},
\]
where $\H_{ni}^{(3)} - H_{ni}^{(3)}$ denotes certain empirical processes with each process using only observations from individual $i$ (see the beginning of section A.2 for a formal definition). A direct analysis of this remainder in \cite{KatoGalvaoMontes-Rojas12} leads to a stochastic order $O_P((\log(n)/T)^{3/4})$ which can not be improved (in terms of the power of $T$) and gives rise to the condition $n^2(\log n)^3/T = o(1)$ in the latter paper. Given the discussion in the MD case it would seem natural to again look at the expected value of this remainder term. However, this is not directly helpful since the terms in the sum are dependent across $i$ through the common $\widehat \zb$ and the $\widehat{\alpha}_i$ so that a direct control of the variance of this term is infeasible. The key idea to solve those difficulties is to first show that this remainder term is sufficiently close to 
\[
\frac{1}{n} \sum_{i=1}^n \H_{ni}^{(3)}(\widetilde\alpha_i,\zb_0) - \H_{ni}^{(3)}(\alpha_{i0},\zb_0) - \Big\{H_{ni}^{(3)}(\widetilde\alpha_{i},\zb_0) - H_{ni}^{(3)}(\alpha_{i0},\zb_0)\Big\}
\]
where $\widetilde \alpha_i := \argmin_{a \in R} \sum_{t=1}^T \rho_\tau(Y_{it} - \X_{it}^\top \zb_0 - a)$. In contrast to the $\widehat\alpha_i$ which are dependent across $i$ because they result from running one single regression the $\widetilde\alpha_i$ only use data from individual $i$. Hence the terms in the above sum are now independent across $i$ and the ideas from~\cite{VCC} can be applied to show that their expected values are of order $T^{-1}\log T$ while the variance of the sum is negligible due to averaging. This eventually leads to the condition $\log T/T = o((nT)^{-1/2})$ which explains the source of our less restrictive growth condition.
}
\end{remark}

}

\subsection{The dependent case}\label{sec:dep}

In this section we extend the asymptotic results provided in Section~\ref{sec:indep} to settings where we allow for dependence across time while still maintaining independence across individuals. We continue to use notations introduced in the previous section. 

{ Throughout this section we will make the following assumptions on the temporal dependence structure.

\begin{itemize}
\item[(D1)]\label{B1} For each $i \geq 1$, the process $\{ (Y_{it},\X_{it}) : t \in \mathbb{Z} \}$ is strictly stationary and $\beta$-mixing. Let $\beta_{i}(j)$ denote the $\beta$-mixing coefficient of the process $\{ (Y_{it},\X_{it}) :t \in \mathbb{Z}  \}$. Assume that there exist constants $b_\beta \in (0,1), C_\beta > 0$ independent of $i$ such that
\begin{equation*}
\sup_i \beta_i(j) \leq C_{\beta} b_\beta^j =: \beta(j) \quad \forall j \geq 1.
\end{equation*}
\item[(D2)]\label{B2} For each $i=1,...,n$ and $j>1$, the random vector $(Y_{i1},Y_{i 1+j})$ has a density conditional on $(\Z_{i1},\Z_{i 1+j})$ and this density is bounded uniformly across $i,j$.
\end{itemize}
}
Condition (D1) relaxes the assumption of i.i.d. within each individual to that of stationary $\beta$-mixing which is used in \cite{KatoGalvaoMontes-Rojas12, GalvaoWang15} and is similar to \cite{HahnKuersteiner11}. The $\beta$-mixing condition is stronger than $\alpha$-mixing. Nevertheless, $\beta$-mixing is still satisfied in a reasonably large class of time series models.\footnote{For example, consider the MA $(\infty)$ process $X_{t}= \sum_{j=0}^{ \infty}a_{j} \varepsilon_{t-j}$, where $a_{j} \to 0$ exponentially fast (note that the ARMA($p,q$) process, subject to standard assumptions, fulfils this condition), and $\{ \varepsilon_{t} \}$ is i.i.d. If the density function of $\varepsilon_{t}$ exists, then $\{ X_{t}\}$ is $\beta$-mixing with $\beta(n) \to 0$ exponentially fast.} Condition (D2) is needed because the data are not i.i.d. and we need to impose a condition on the joint distributions, {similar conditions were imposed in \cite{KatoGalvaoMontes-Rojas12, GalvaoWang15}. }

{To state the asymptotic properties of $\widehat\zb_{FE}(\tau)$ we make the following additional assumption.

\begin{enumerate}
\item[(FD)] Assume that $\Gamma_n$ is non-singular for each $n$ and that $\lim_{n \to \infty} \Gamma_n$ exists and is non-singular. Further assume that 
\[
\widetilde L := \lim_{n \to \infty} \frac{1}{n}\sum_{i=1}^n Var\Big(T^{-1/2}\sum_{t=1}^T (\X_{it}-g_i)(\tau - \1\{Y_{it} \leq \X_{it}^\top \zb_0- \alpha_{i0}\})\Big)
\]
exists and is non-singular.
\end{enumerate}

This assumption was also made by \cite{KatoGalvaoMontes-Rojas12} (see their condition (D3)), it involves the long run covariance matrix of the leading piece in the Bahadur representation for $\widehat \zb_{FE}(\tau)$.

\begin{theorem}\label{th2FE}
Assume that (A0)--(A3), (D1)--(D2) and (FD) hold. If also $n (\log T)^4/T = o(1)$ then, for fixed $\tau$, 
\begin{equation}\label{eq:mainFE_dep}
\sqrt{nT}(\widehat{\bm{\beta}}_{FE}(\tau) - \bm{\beta}_{0}(\tau)) \stackrel{\mathcal{D}}{\longrightarrow} \mathcal{N}\big(0,\tau(1-\tau)\Gamma^{-1}\widetilde L\Gamma^{-1}\big).
\end{equation}
\end{theorem}

Similarly as in the i.i.d. case we obtain the same limiting behaviour as \cite{KatoGalvaoMontes-Rojas12} but under a weaker condition on the growth rate of $n$ relative to $T$. 
}

We next discuss the MD-QR estimator. In order to construct the MD-QR estimator with dependent data, we need to account for the dependence when estimating the covariance matrix. To this end we extend the Hendricks-Koenker estimator in the previous section. Note that the limiting variance of $\widehat{\bm{ \gamma}}_{i}$ in the dependent case takes the form $\widetilde V_{i} = B_{i}^{-1}\widetilde{A}_{i} B_{i}^{-1}$ where $B_i$ is the same as in the independent case (see~\eqref{eq:AiBi}) and 
\begin{equation*}
\widetilde{A}_i := \tau (1-\tau) \E[\Z_{i1}\Z_{i1}^\top] + \sum_{j=1}^{\infty} \E(w_{i1}w_{i1+j}^\top + w_{i1+j} w_{i1}^\top),
\end{equation*}
with $w_{it} := \Z_{it} (\tau - \1(Y_{it} \leq q_{i,\tau}(\bm{Z}_{it}))$. This motivates the following estimator of the limiting variance matrix of $\widehat{\bm{\gamma}}_{i}(\tau)$
\begin{equation}
\label{eq:Dvar}
\widetilde{V}_{iT} := \widehat{B}_{iT}^{-1} \widetilde{A}_{iT}\widehat{B}_{iT}^{-1},
\end{equation}
where $\widehat{B}_{iT} :=  \frac{1}{T} \sum_{t=1}^{T} \widehat{f}_{it} \Z_{it}\Z_{it}^\top$ is defined in the same way as in the independent case and the estimator $\widetilde{A}_{iT}$ now takes the form 
\begin{equation*}
\widetilde{A}_{iT} := \frac{\tau(1-\tau)}{T}\sum_{t=1}^{T} \Z_{it} \Z_{it}^\top + \sum_{1 \leq j \leq m_{T}}\Big(1- \frac{j}{T}\Big) \Big(\frac{1}{T}\sum_{t\in T_{j}} (\widehat{w}_{it} \widehat{w}_{it+j}^\top+\widehat{w}_{it+j}\widehat{w}_{it}^\top)\Big),
\end{equation*}
with $T_j := \{1\leq t \leq T - j\}$. Here $m_T$ is a `bandwidth' parameter (which needs to increase to infinity in order to obtain consistent estimation), and
\begin{equation*}
\widehat{w}_{it} := \Z_{it} (\tau - \1(Y_{it} \leq \widehat{\bm{\gamma}}_{i}(\tau)^\top \Z_{it})).
\end{equation*}
The estimator $\widehat{\bm{\beta}}_{MD}(\tau)$ is defined as in \eqref{eq4} with the new estimator $\widetilde{V}_{iT}$ in \eqref{eq:Dvar} instead of $\widehat{V}_{iT}$. 

To derive the asymptotic distribution of $\widehat{\bm{\beta}}_{MD}(\tau)$ we consider the following set of assumptions.

\begin{itemize}
\item[(MD)]\label{B3} Assume that $d_T \to 0$, $m_T \to \infty$ and
\begin{equation*}
\frac{n (\log T)^4 m_T ^3}{T} = o(1), \quad \frac{\log T}{T d_T^2} = o(1).
\end{equation*}
\end{itemize}

Condition (MD) is similar to (MI) and imposes a restriction on the relative growth of the time dimension compared to number of individuals.

The following theorem provides the asymptotic normality result under stationary $\beta$-mixing dependence, and is an extension of Theorem~\ref{th1} in the previous section.

\begin{theorem} \label{th2}
Assume that (A0)--(A3) and (D1)--(D2) and (MD) hold and let $\widetilde W_{iT}$ denote the lower $p\times p$ sub-matrix of $\widetilde V_{iT}$. Assume that $\Sigma_{2}:= \lim_{n\to \infty} \Big(\frac{1}{n}\sum_{i=1}^{n} \widetilde  W_{iT}^{-1}\Big)^{-1}$ exists and is a well-defined positive definite matrix. In this case, for fixed $\tau$, we have that
\begin{equation}\label{eq:main_dep}
\sqrt{nT}(\widehat{\bm{\beta}}_{MD}(\tau) - \bm{\beta}_{0}(\tau)) \stackrel{\mathcal{D}}{\longrightarrow} N\big(0,\Sigma_{2} \big).
\end{equation}
\end{theorem}

Theorem \ref{th2} shows that the asymptotic normality of the estimator for the dependent data still holds under similar conditions on $n,T$ as Theorem~\ref{th1}. The following Lemma provides a result for the consistent estimation of the covariance matrix for the dependent case.

\begin{lemma}
\label{lemma:CovD}
Under the conditions of Theorem \ref{th2}, we have we have $\max_{i=1,...,n}\|\widehat{B}_{iT} - B_{i} \| = o_P(1)$ and $\max_{i=1,...,n}\|\widetilde{A}_{iT} - \widetilde  A_{i} \| = o_P(1)$. Moreover, $\widehat \Sigma_2 := (\frac{1}{n}\sum_{i=1}^n \widetilde{W}_{iT}^{-1} )^{-1}$ is a consistent estimator for $\Sigma_2$.
\end{lemma}

As in the previous section, we can see that the asymptotic covariance matrix of $\widehat{\bm{\beta}}_{MD}(\tau)$ for the dependent case, $\Sigma_{2}$ is consistently estimated by the inverse of $\frac{1}{n} \sum_{i=1}^{n} \widetilde{W}_{iT}^{-1} $. The proof strategy is similar to the i.i.d. case (see Remark~\ref{rem:proofideas} for an outline of the key steps) but somewhat more complicated due to the dependence within individuals. Note also that the estimators $\widetilde{A}_{iT}$ now contain an additional sum with $m_T \to \infty$ pieces, this poses an additional challenge for its theoretical analysis. As in the previous section, the result in Lemma \ref{lemma:CovD} allows for the construction of inference procedures based on the Wald statistic.

\section{Monte Carlo simulations} \label{sec:sim}

\subsection{Design}

In this section, {we use simulation experiments to assess the finite
sample performance of the FE-QR and MD-QR estimators discussed previously.} 
We consider a simple model as in equation \eqref{eq:model}, where 
the response $y_{it}$ is generated by 
\begin{equation*}
y_{it}=\alpha_{i}+\beta x_{it}+(1+\lambda x_{it})u_{it}.
\end{equation*}
This is a general location-scale-shift model. When $\lambda=0$, 
the exogenous covariate $x_{it}$ exerts a pure location
shift effect. When $\lambda \neq 0$, $x_{it}$ exerts both location and scale
effects. 

We set $\beta=1$ and consider $\lambda\in\{0, 1\}$. We employ three different schemes to generate the disturbances $u_{it}$. Under scheme 1, we generate $u_{it}$ as a $N(0,1)$, under scheme 2 we generate $u_{it}$ as a $t$-distribution with 3 degrees of freedom, and under scheme 3 the disturbances follow $\chi_{3}^{2}$. In this model the true coefficients take the form $\alpha_{i0}(\tau)=\alpha_{i}+F^{-1}(\tau)$ and $\beta_{0}(\tau)=\beta+\lambda F^{-1}(\tau)$ where $F$ denotes the CDF of $u_{it}$.

In all cases, the fixed effects, $\alpha_{i}$, are generated as $\alpha_i = i/n, i=1,...,n$.
The independent variable, $x_{it}$, is generated according to $x_{it}=0.3\alpha_{i}+v_{it}$, where $v_{it}\sim U[0,10]$. This method of generating $\alpha_{i}$ ensures that the classical
random effects estimators are biased because the individual effects and the explanatory variables are correlated. 

In the numerical simulations we study the FE-QR defined in \eqref{eq:feqr}, which includes the dummy variables to account for the fixed effects, and the MD-QR estimator in \eqref{eq:fefmd} using the Hendricks-Koenker sandwich covariance matrix estimator in \eqref{eq:HKvar} as weights.\footnote{For estimation of the asymptotic covariance matrix,  we use the default  bandwidth option in the {\tt quantreg} package in \texttt{R}{; for the Hendricks-Koenker estiamtor this corresponds to the Hall-Sheather method.}} Moreover, for comparison, we report results for the infeasible MD-QR estimator generated with true weights instead of the estimated weights, which are computed by using the true sparsity function in the corresponding variance-covariance matrix. It is important to present results for the MD-QR with the true sparsity to investigate to what extent the finite sample performance of the estimator is affected by estimation of the sparsity function in the weights. The MD-QR using the true sparsity is denoted by MDT-QR.

We report results for bias and standard errors (SE) for the FE-QR, MD-QR, and MDT-QR estimators for the quantiles $\tau=\{0.25,0.5,0.75\}$. All reported results are based on $2,000$ Monte Carlo replications.

\subsection{Results}

Before discussing the Monte Carlo results, it is worthwhile to review the theoretical predictions for the estimators. Remarks~\ref{rem:proofideas} and~\ref{rem:proofideasFE} in Section \ref{sec:indep} above discussed properties of the MD-QR and FE-QR estimator. The results there provide an upper bound for the order of bias of both estimators which are of the form $(\log T)/T$ and thus depend only on the time series dimension $T$. This indicates that $T\times bias$ should be at most constant or increase very slowly. In addition, the standard error should be proportional to $\sqrt{nT}$.

{To illustrate those theoretical predictions we report $T\times bias$ for the three estimators described above, FE-QR, MD-QR, and MDT-QR. Table~\ref{table_N_0_bias}, Table~\ref{table_t_0_bias}, and Table~\ref{table_chi_0_bias} present results for $\lambda = 0$ for all three estimators and normal, $t_3$ and $\chi_3^2$ errors, respectively. Table~\ref{table_N_1_bias}, Table~\ref{table_t_1_bias}, and Table~\ref{table_chi_1_bias} present the corresponding results for $\lambda = 1$. Moreover, $\sqrt{nT}\times SE$ is reported in Table~\ref{table_N_0_se}, Table~\ref{table_t_0_se}, and Table~\ref{table_chi_0_se} for the location case ($\lambda=0$) for all three estimators and the same three error distributions as above, respectively. Finally, Table~\ref{table_N_1_se}, Table~\ref{table_t_1_se}, and Table~\ref{table_chi_1_se} display the corresponding results to the location-scale ($\lambda=1$) case.  }

The first important observation we find when examining Tables~\ref{table_N_1_bias},~\ref{table_t_1_bias},~\ref{table_chi_1_bias} is that $T\times bias$ is indeed roughly constant across $n$ for all three estimators considered, indicating that our upper bound is close to being sharp. Tables~\ref{table_N_0_bias},~\ref{table_t_0_bias},~\ref{table_chi_0_bias} further indicate that the bias in the pure location-shift model (i.e. $\lambda = 0$) seems to be extremely small even when multiplied by $T$ for all estimators, hinting that in some models the bias might be of even higher order than $T^{-1}$. This does not contradict our theory since we only provide an upper bound for the bias. Another interesting finding is that the exact bias of the three estimators can be quite different (see in particular Table~\ref{table_t_1_bias} and Table~\ref{table_chi_1_bias}). Again, this does not contradict our theoretical predictions since we provide an upper bound with rate while the difference seems to be in the constants. It also seems that estimating the weight in the MD estimator can have a significant impact on the constant in the bias. Finally, there is no estimator that has uniformly smallest bias.


{The scaled (by $\sqrt{nT}$) standard errors are reported in Table~\ref{table_N_0_se}, Table~\ref{table_t_0_se}, and Table~\ref{table_chi_0_se} for the location case and in Table~\ref{table_N_1_se}, Table~\ref{table_t_1_se}, and Table~\ref{table_chi_1_se} for the location-shift case. The first main result is that $\sqrt{nT}\times SE$ is stable across a range of $n,T$ for both $\lambda=0$ and $\lambda=1$, with a slight decreasing trend as $T$ grows, especially for the MD-QR estimator. Both FE-QR and MDT-QR suffer less variation. This slight larger trend for MD-QR is probably due to higher-order terms that result from estimated weights in the MD-QR estimator. In addition, when comparing $\sqrt{nT}\times SE$ for MD-QR and MDT-QR we see that the latter starts with smaller results (for small $T$), however, this difference tends to disappear as $T$ increases.}

\begin{figure}[!h]
\caption{Bias and Standard Error for $\lambda=1$ and $\tau=0.75$ for $\chi_3^2$ errors.}
\label{Fig2}\centering
\includegraphics[scale=0.55]{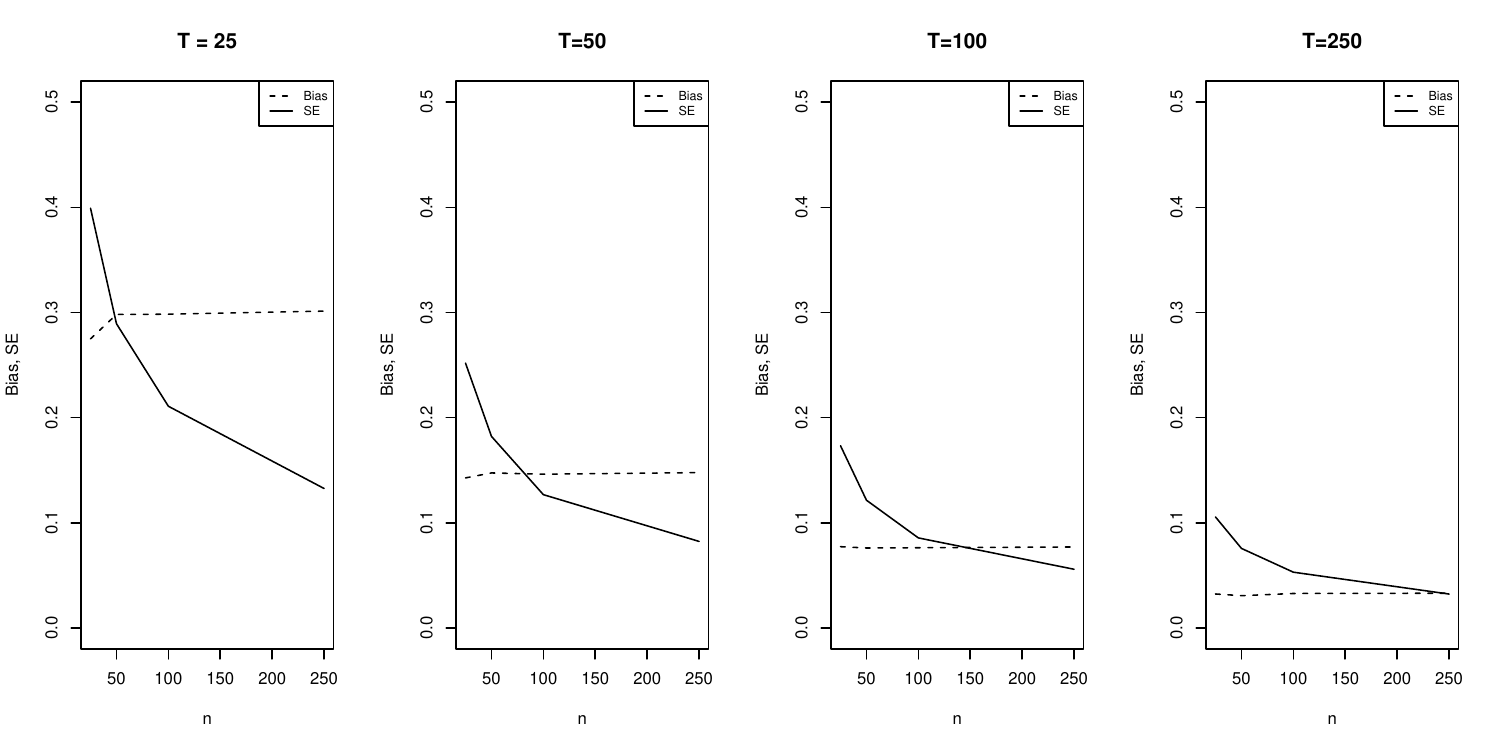}
\end{figure}

{Finally, to further illustrate the discussion in Remark~\ref{rem:proofideas}, Figure~\ref{Fig2} shows different plots of the non-standardized bias (dotted lines) and SE (solid lines) as functions of $n$ for different values of $T$ for $\lambda=1$, $\tau=0.75$ and $\chi_3^2$ errors, for the MD-QR estimator. Plots for the FE-QR estimator are similar and we omit them for brevity. As expected from the theory, a comparison for each particular plot reveals that bias is stable across $n$, and a comparison across plots shows that the bias only depends on $T$ while the standard error decreases with $n$ and $T$.\footnote{For robustness purposes, we have conducted experiments with $n=5,000$ and $T=25$ and $T=100$. The results are qualitatively similar, and we omit them for brevity.} Approximately unbiased normality is expected to hold when the standard error is larger than the bias, which is the case when $T \gg n$ in the plots. For a fixed $T$, as $n$ increases, the standard error decreases and at some point is of the same order as the bias; this is when unbiased normality fails. Nevertheless, as predicted this happens for larger and larger values of $n$ as $T$ increases.} 

\begin{scriptsize} 
	\begin{longtable} {ccc|ccc||ccc||ccc}
		\caption{$T$ $\times$ bias for different estimators with normal errors. Location-shift ($\lambda=0$)}\\
		\label{table_N_0_bias}\\
		\hline
		&          &  &\multicolumn{3}{c}{$\tau = 0.25$} & \multicolumn{3}{c}{$\tau = 0.5$} & \multicolumn{3}{c}{$\tau = 0.75$}\\
		\hline
		$n$	&	$T$	&	$n/T$	&	$MD$	&	$FE$	&	$MDT$	&	$MD$	&	$FE$	&	$MDT$	&	$MD$	&	$FE$	&	$MDT$	\\
		\hline
25 & 25 & 1.00 & -0.014 & -0.011 & -0.002 & -0.012 & -0.013 & -0.008 & 0.003 & -0.003 & -0.004 \\ 
25 & 50 & 0.50 & -0.005 & -0.000 & -0.008 & -0.011 & -0.012 & -0.013 & -0.025 & -0.013 & -0.019 \\ 
25 & 100 & 0.25 & -0.007 & -0.006 & -0.004 & -0.004 & -0.018 & -0.006 & 0.022 & 0.031 & 0.016 \\ 
25 & 250 & 0.10 & 0.007 & 0.001 & 0.003 & -0.015 & -0.036 & -0.011 & -0.015 & -0.021 & -0.014 \\ 
50 & 25 & 2.00 & 0.003 & -0.007 & -0.002 & -0.008 & -0.009 & -0.005 & -0.016 & -0.009 & -0.009 \\ 
50 & 50 & 1.00 & -0.011 & -0.006 & -0.005 & -0.007 & -0.010 & -0.003 & -0.008 & 0.009 & -0.005 \\ 
50 & 100 & 0.50 & -0.030 & -0.036 & -0.032 & -0.019 & -0.027 & -0.024 & -0.020 & -0.023 & -0.026 \\ 
50 & 250 & 0.20 & 0.010 & 0.009 & 0.013 & 0.017 & 0.017 & 0.012 & 0.009 & 0.013 & 0.010 \\ 
100 & 25 & 4.00 & -0.004 & -0.002 & -0.002 & -0.003 & -0.006 & -0.002 & -0.002 & 0.001 & -0.002 \\ 
100 & 50 & 2.00 & -0.020 & -0.019 & -0.018 & -0.016 & -0.016 & -0.015 & -0.011 & -0.016 & -0.015 \\ 
100 & 100 & 1.00 & -0.010 & -0.016 & -0.009 & -0.019 & -0.022 & -0.019 & 0.004 & 0.002 & -0.002 \\ 
100 & 250 & 0.40 & -0.012 & -0.008 & -0.009 & -0.009 & -0.009 & -0.013 & -0.020 & -0.020 & -0.025 \\ 
250 & 25 & 10.00 & 0.001 & 0.002 & 0.000 & -0.001 & -0.002 & -0.002 & -0.002 & -0.002 & -0.002 \\ 
250 & 50 & 5.00 & 0.004 & 0.003 & 0.005 & 0.003 & 0.005 & 0.005 & 0.003 & 0.001 & 0.002 \\ 
250 & 100 & 2.50 & -0.005 & -0.005 & -0.002 & -0.006 & -0.008 & -0.007 & -0.010 & -0.013 & -0.008 \\ 
250 & 250 & 1.00 & 0.005 & 0.003 & 0.005 & 0.017 & 0.015 & 0.014 & 0.004 & 0.001 & 0.003 \\ 

		\hline\hline
\end{longtable}
\end{scriptsize} 
	
	\newpage

\begin{scriptsize} 
	\begin{longtable} {ccc|ccc||ccc||ccc}
		\caption{$\sqrt{nT}$ $\times$ standard error for different estimators with normal errors. Location-shift ($\lambda=0$)}\\
		\label{table_N_0_se}\\
		\hline
		&          &  &\multicolumn{3}{c}{$\tau = 0.25$} & \multicolumn{3}{c}{$\tau = 0.5$} & \multicolumn{3}{c}{$\tau = 0.75$}\\
		\hline
		$n$	&	$T$	&	$n/T$	&	$MD$	&	$FE$	&	$MDT$	&	$MD$	&	$FE$	&	$MDT$	&	$MD$	&	$FE$	&	$MDT$	\\
		\hline
25 & 25 & 1.00 & 0.541 & 0.485 & 0.480 & 0.467 & 0.445 & 0.438 & 0.547 & 0.495 & 0.489 \\ 
25 & 50 & 0.50 & 0.514 & 0.482 & 0.484 & 0.456 & 0.449 & 0.444 & 0.527 & 0.479 & 0.481 \\ 
25 & 100 & 0.25 & 0.490 & 0.481 & 0.474 & 0.452 & 0.451 & 0.442 & 0.497 & 0.479 & 0.471 \\ 
25 & 250 & 0.10 & 0.488 & 0.477 & 0.480 & 0.452 & 0.448 & 0.445 & 0.493 & 0.474 & 0.476 \\ 
50 & 25 & 2.00 & 0.554 & 0.487 & 0.490 & 0.468 & 0.441 & 0.440 & 0.570 & 0.489 & 0.491 \\ 
50 & 50 & 1.00 & 0.511 & 0.477 & 0.477 & 0.464 & 0.456 & 0.449 & 0.497 & 0.477 & 0.469 \\ 
50 & 100 & 0.50 & 0.499 & 0.476 & 0.483 & 0.449 & 0.439 & 0.438 & 0.503 & 0.477 & 0.480 \\ 
50 & 250 & 0.20 & 0.478 & 0.467 & 0.467 & 0.440 & 0.443 & 0.437 & 0.480 & 0.475 & 0.470 \\ 
100 & 25 & 4.00 & 0.554 & 0.481 & 0.479 & 0.476 & 0.449 & 0.444 & 0.552 & 0.479 & 0.484 \\ 
100 & 50 & 2.00 & 0.504 & 0.473 & 0.467 & 0.461 & 0.444 & 0.446 & 0.523 & 0.473 & 0.480 \\ 
100 & 100 & 1.00 & 0.492 & 0.475 & 0.471 & 0.446 & 0.437 & 0.438 & 0.502 & 0.478 & 0.478 \\ 
100 & 250 & 0.40 & 0.477 & 0.463 & 0.468 & 0.427 & 0.422 & 0.423 & 0.480 & 0.465 & 0.471 \\ 
250 & 25 & 10.00 & 0.567 & 0.495 & 0.489 & 0.472 & 0.450 & 0.445 & 0.568 & 0.487 & 0.484 \\ 
250 & 50 & 5.00 & 0.507 & 0.474 & 0.466 & 0.468 & 0.455 & 0.450 & 0.513 & 0.476 & 0.480 \\ 
250 & 100 & 2.50 & 0.489 & 0.475 & 0.476 & 0.437 & 0.424 & 0.426 & 0.491 & 0.477 & 0.471 \\ 
250 & 250 & 1.00 & 0.478 & 0.470 & 0.469 & 0.443 & 0.438 & 0.437 & 0.476 & 0.469 & 0.462 \\ 
			\hline\hline
	\end{longtable}
\end{scriptsize}

\newpage

\begin{scriptsize} 
	\begin{longtable} {ccc|ccc||ccc||ccc}
		\caption{$T$ $\times$ bias for different estimators with $t_3$ errors. Location-shift ($\lambda=0$)}\\
		\label{table_t_0_bias}\\
		\hline
		&          &  &\multicolumn{3}{c}{$\tau = 0.25$} & \multicolumn{3}{c}{$\tau = 0.5$} & \multicolumn{3}{c}{$\tau = 0.75$}\\
		\hline
		$n$	&	$T$	&	$n/T$	&	$MD$	&	$FE$	&	$MDT$	&	$MD$	&	$FE$	&	$MDT$	&	$MD$	&	$FE$	&	$MDT$	\\
		\hline
	25 & 25 & 1.00 & -0.004 & 0.003 & -0.008 & -0.010 & -0.005 & -0.001 & -0.023 & -0.015 & -0.011 \\ 
	25 & 50 & 0.50 & 0.006 & -0.006 & -0.009 & 0.008 & -0.010 & 0.011 & 0.002 & 0.003 & 0.004 \\ 
	25 & 100 & 0.25 & 0.020 & 0.008 & 0.013 & 0.009 & -0.006 & 0.003 & -0.011 & -0.026 & -0.025 \\ 
	25 & 250 & 0.10 & 0.065 & 0.019 & 0.037 & -0.044 & -0.024 & -0.026 & -0.016 & -0.036 & -0.012 \\ 
	50 & 25 & 2.00 & -0.000 & -0.001 & -0.009 & 0.003 & 0.000 & 0.003 & -0.020 & -0.003 & 0.003 \\ 
	50 & 50 & 1.00 & -0.005 & 0.001 & 0.008 & 0.001 & 0.002 & 0.004 & -0.005 & -0.016 & -0.006 \\ 
	50 & 100 & 0.50 & -0.005 & 0.000 & -0.010 & -0.008 & -0.000 & -0.007 & -0.021 & -0.009 & -0.009 \\ 
	50 & 250 & 0.20 & -0.025 & -0.019 & -0.031 & 0.018 & 0.029 & 0.023 & 0.009 & -0.008 & 0.005 \\ 
	100 & 25 & 4.00 & -0.003 & -0.001 & -0.007 & 0.001 & -0.002 & 0.001 & -0.004 & -0.008 & -0.006 \\ 
	100 & 50 & 2.00 & -0.003 & -0.004 & -0.003 & 0.001 & -0.000 & -0.001 & 0.001 & -0.002 & 0.000 \\ 
	100 & 100 & 1.00 & 0.001 & -0.004 & -0.009 & -0.015 & -0.012 & -0.019 & 0.008 & 0.000 & -0.003 \\ 
	100 & 250 & 0.40 & 0.011 & 0.004 & 0.006 & 0.015 & 0.003 & 0.011 & -0.009 & -0.008 & -0.009 \\ 
	250 & 25 & 10.00 & 0.003 & 0.000 & 0.002 & -0.003 & -0.002 & -0.001 & -0.005 & -0.004 & -0.002 \\ 
	250 & 50 & 5.00 & -0.003 & -0.007 & -0.004 & -0.003 & 0.002 & -0.001 & 0.005 & 0.003 & 0.002 \\ 
	250 & 100 & 2.50 & 0.003 & -0.003 & 0.004 & 0.008 & 0.004 & 0.004 & -0.010 & -0.007 & -0.008 \\ 
	250 & 250 & 1.00 & -0.002 & -0.001 & 0.002 & 0.004 & 0.005 & 0.004 & 0.018 & 0.018 & 0.012 \\ 
	
		\hline\hline
	\end{longtable}
\end{scriptsize}

\begin{scriptsize} 
	\begin{longtable} {ccc|ccc||ccc||ccc}
		\caption{$\sqrt{nT}$ $\times$ standard error for different estimators with $t_3$ errors. Location-shift ($\lambda=0$)}\\
		\label{table_t_0_se}\\
		\hline
		&          &  &\multicolumn{3}{c}{$\tau = 0.25$} & \multicolumn{3}{c}{$\tau = 0.5$} & \multicolumn{3}{c}{$\tau = 0.75$}\\
		\hline
		$n$	&	$T$	&	$n/T$	&	$MD$	&	$FE$	&	$MDT$	&	$MD$	&	$FE$	&	$MDT$	&	$MD$	&	$FE$	&	$MDT$	\\
		\hline
	25 & 25 & 1.00 & 0.738 & 0.611 & 0.642 & 0.546 & 0.499 & 0.519 & 0.750 & 0.613 & 0.651 \\ 
	25 & 50 & 0.50 & 0.659 & 0.604 & 0.600 & 0.506 & 0.489 & 0.495 & 0.645 & 0.583 & 0.604 \\ 
	25 & 100 & 0.25 & 0.614 & 0.586 & 0.592 & 0.480 & 0.471 & 0.477 & 0.604 & 0.582 & 0.577 \\ 
	25 & 250 & 0.10 & 0.604 & 0.595 & 0.590 & 0.473 & 0.472 & 0.466 & 0.592 & 0.575 & 0.580 \\ 
	50 & 25 & 2.00 & 0.756 & 0.603 & 0.632 & 0.524 & 0.487 & 0.498 & 0.724 & 0.593 & 0.617 \\ 
	50 & 50 & 1.00 & 0.656 & 0.592 & 0.600 & 0.493 & 0.485 & 0.488 & 0.646 & 0.585 & 0.592 \\ 
	50 & 100 & 0.50 & 0.607 & 0.578 & 0.582 & 0.471 & 0.464 & 0.464 & 0.609 & 0.590 & 0.594 \\ 
	50 & 250 & 0.20 & 0.603 & 0.589 & 0.594 & 0.487 & 0.476 & 0.480 & 0.606 & 0.596 & 0.594 \\ 
	100 & 25 & 4.00 & 0.738 & 0.593 & 0.611 & 0.522 & 0.481 & 0.502 & 0.737 & 0.584 & 0.619 \\ 
	100 & 50 & 2.00 & 0.653 & 0.584 & 0.602 & 0.489 & 0.472 & 0.477 & 0.649 & 0.581 & 0.602 \\ 
	100 & 100 & 1.00 & 0.630 & 0.588 & 0.605 & 0.498 & 0.483 & 0.486 & 0.625 & 0.592 & 0.604 \\ 
	100 & 250 & 0.40 & 0.602 & 0.590 & 0.593 & 0.482 & 0.478 & 0.478 & 0.603 & 0.585 & 0.588 \\ 
	250 & 25 & 10.00 & 0.749 & 0.596 & 0.640 & 0.539 & 0.479 & 0.500 & 0.746 & 0.582 & 0.614 \\ 
	250 & 50 & 5.00 & 0.676 & 0.597 & 0.601 & 0.505 & 0.484 & 0.491 & 0.660 & 0.602 & 0.610 \\ 
	250 & 100 & 2.50 & 0.624 & 0.592 & 0.603 & 0.493 & 0.482 & 0.487 & 0.620 & 0.593 & 0.598 \\ 
	250 & 250 & 1.00 & 0.593 & 0.572 & 0.581 & 0.466 & 0.458 & 0.460 & 0.604 & 0.589 & 0.590 \\ 
		\hline\hline
	\end{longtable}
\end{scriptsize}

	\newpage
	\begin{scriptsize} 
		\begin{longtable} {ccc|ccc||ccc||ccc}
			\caption{$T$ $\times$ bias for different estimators with $\chi_3^2$ errors. Location-shift ($\lambda=0$)}\\
			\label{table_chi_0_bias}\\
			\hline
			&          &  &\multicolumn{3}{c}{$\tau = 0.25$} & \multicolumn{3}{c}{$\tau = 0.5$} & \multicolumn{3}{c}{$\tau = 0.75$}\\
			\hline
			$n$	&	$T$	&	$n/T$	&	$MD$	&	$FE$	&	$MDT$	&	$MD$	&	$FE$	&	$MDT$	&	$MD$	&	$FE$	&	$MDT$	\\
			\hline
25 & 25 & 1.00 & 0.004 & 0.017 & 0.005 & 0.014 & 0.001 & 0.019 & 0.009 & 0.014 & -0.009 \\ 
25 & 50 & 0.50 & -0.017 & -0.019 & -0.016 & -0.019 & -0.025 & -0.014 & -0.048 & -0.017 & -0.057 \\ 
25 & 100 & 0.25 & -0.015 & -0.030 & -0.007 & -0.010 & -0.018 & -0.021 & -0.069 & -0.077 & -0.084 \\ 
25 & 250 & 0.10 & -0.030 & -0.031 & -0.025 & -0.023 & -0.024 & -0.026 & 0.024 & 0.013 & 0.015 \\ 
50 & 25 & 2.00 & -0.005 & -0.011 & -0.012 & -0.006 & -0.015 & -0.006 & -0.035 & -0.008 & -0.024 \\ 
50 & 50 & 1.00 & -0.010 & -0.005 & -0.012 & -0.007 & -0.002 & -0.014 & -0.010 & -0.036 & -0.038 \\ 
50 & 100 & 0.50 & -0.008 & -0.008 & 0.004 & 0.037 & 0.013 & 0.019 & 0.058 & 0.042 & 0.040 \\ 
50 & 250 & 0.20 & -0.011 & 0.006 & -0.014 & -0.020 & -0.003 & -0.022 & 0.129 & 0.131 & 0.124 \\ 
100 & 25 & 4.00 & -0.006 & -0.010 & -0.000 & -0.001 & -0.001 & -0.010 & -0.018 & -0.018 & -0.024 \\ 
100 & 50 & 2.00 & -0.006 & -0.012 & -0.003 & -0.003 & 0.002 & -0.000 & 0.022 & 0.029 & 0.022 \\ 
100 & 100 & 1.00 & -0.001 & 0.008 & 0.007 & 0.006 & 0.011 & 0.010 & 0.022 & 0.009 & 0.009 \\ 
100 & 250 & 0.40 & 0.003 & -0.000 & 0.006 & 0.023 & 0.030 & 0.030 & 0.047 & 0.068 & 0.058 \\ 
250 & 25 & 10.00 & -0.002 & 0.000 & 0.002 & -0.007 & -0.007 & 0.002 & -0.017 & -0.009 & -0.012 \\ 
250 & 50 & 5.00 & 0.003 & 0.004 & 0.006 & -0.003 & -0.005 & 0.002 & 0.013 & 0.017 & 0.009 \\ 
250 & 100 & 2.50 & 0.003 & -0.002 & 0.002 & 0.021 & 0.011 & 0.018 & 0.040 & 0.036 & 0.034 \\ 
250 & 250 & 1.00 & -0.009 & -0.007 & -0.009 & -0.040 & -0.032 & -0.033 & 0.010 & 0.018 & 0.014 \\ 
			\hline\hline
		\end{longtable}
	\end{scriptsize} 

\begin{scriptsize} 
	\begin{longtable} {ccc|ccc||ccc||ccc}
		\caption{$\sqrt{nT}$ $\times$ standard error for different estimators with $\chi_3^2$ errors. Location-shift  ($\lambda=0$)}\\
		\label{table_chi_0_se}\\
		\hline
		&          &  &\multicolumn{3}{c}{$\tau = 0.25$} & \multicolumn{3}{c}{$\tau = 0.5$} & \multicolumn{3}{c}{$\tau = 0.75$}\\
		\hline
		$n$	&	$T$	&	$n/T$	&	$MD$	&	$FE$	&	$MDT$	&	$MD$	&	$FE$	&	$MDT$	&	$MD$	&	$FE$	&	$MDT$	\\
		\hline
	25 & 25 & 1.00 & 0.664 & 0.636 & 0.650 & 0.972 & 0.938 & 0.943 & 1.668 & 1.446 & 1.473 \\ 
	25 & 50 & 0.50 & 0.646 & 0.618 & 0.623 & 0.954 & 0.935 & 0.939 & 1.617 & 1.459 & 1.471 \\ 
	25 & 100 & 0.25 & 0.638 & 0.627 & 0.633 & 0.963 & 0.943 & 0.947 & 1.553 & 1.484 & 1.498 \\ 
	25 & 250 & 0.10 & 0.630 & 0.625 & 0.622 & 0.924 & 0.918 & 0.916 & 1.496 & 1.456 & 1.466 \\ 
	50 & 25 & 2.00 & 0.658 & 0.624 & 0.628 & 0.987 & 0.919 & 0.933 & 1.726 & 1.455 & 1.496 \\ 
	50 & 50 & 1.00 & 0.634 & 0.617 & 0.625 & 0.977 & 0.953 & 0.962 & 1.619 & 1.462 & 1.472 \\ 
	50 & 100 & 0.50 & 0.630 & 0.613 & 0.617 & 0.939 & 0.913 & 0.918 & 1.515 & 1.450 & 1.456 \\ 
	50 & 250 & 0.20 & 0.655 & 0.633 & 0.641 & 0.946 & 0.946 & 0.943 & 1.504 & 1.462 & 1.470 \\ 
	100 & 25 & 4.00 & 0.663 & 0.617 & 0.637 & 1.023 & 0.943 & 0.958 & 1.739 & 1.485 & 1.493 \\ 
	100 & 50 & 2.00 & 0.647 & 0.615 & 0.625 & 0.949 & 0.915 & 0.927 & 1.603 & 1.438 & 1.440 \\ 
	100 & 100 & 1.00 & 0.643 & 0.626 & 0.632 & 0.945 & 0.929 & 0.926 & 1.549 & 1.455 & 1.477 \\ 
	100 & 250 & 0.40 & 0.642 & 0.627 & 0.636 & 0.966 & 0.940 & 0.954 & 1.508 & 1.453 & 1.460 \\ 
	250 & 25 & 10.00 & 0.671 & 0.619 & 0.631 & 0.968 & 0.924 & 0.920 & 1.745 & 1.456 & 1.462 \\ 
	250 & 50 & 5.00 & 0.661 & 0.631 & 0.646 & 0.974 & 0.929 & 0.954 & 1.586 & 1.439 & 1.462 \\ 
	250 & 100 & 2.50 & 0.649 & 0.633 & 0.640 & 0.972 & 0.932 & 0.946 & 1.579 & 1.462 & 1.475 \\ 
	250 & 250 & 1.00 & 0.644 & 0.630 & 0.632 & 0.923 & 0.911 & 0.911 & 1.450 & 1.434 & 1.419 \\ 
		\hline\hline
	\end{longtable}
\end{scriptsize}

\newpage

\begin{scriptsize} 
	\begin{longtable} {ccc|ccc||ccc||ccc}
		\caption{$T$ $\times$ bias for different estimators with normal errors. Location-scale shift ($\lambda=1$)}\\
		\label{table_N_1_bias}\\
		\hline
		&          &  &\multicolumn{3}{c}{$\tau = 0.25$} & \multicolumn{3}{c}{$\tau = 0.5$} & \multicolumn{3}{c}{$\tau = 0.75$}\\
		\hline
		$n$	&	$T$	&	$n/T$	&	$MD$	&	$FE$	&	$MDT$	&	$MD$	&	$FE$	&	$MDT$	&	$MD$	&	$FE$	&	$MDT$	\\
		\hline
	25 & 25 & 1.00 & 0.992 & 0.867 & 0.948 & -0.038 & -0.054 & -0.006 & -0.993 & -0.830 & -0.890 \\ 
	25 & 50 & 0.50 & 1.151 & 0.870 & 0.876 & -0.077 & -0.095 & -0.095 & -1.167 & -0.886 & -0.943 \\ 
	25 & 100 & 0.25 & 1.171 & 0.799 & 0.843 & -0.093 & -0.107 & -0.140 & -1.051 & -0.750 & -0.788 \\ 
	25 & 250 & 0.10 & 1.514 & 0.969 & 1.041 & -0.069 & -0.152 & -0.055 & -1.354 & -0.894 & -0.877 \\ 
	50 & 25 & 2.00 & 1.089 & 0.879 & 0.914 & -0.043 & -0.034 & -0.060 & -1.158 & -0.870 & -1.005 \\ 
	50 & 50 & 1.00 & 1.073 & 0.816 & 0.895 & -0.048 & -0.026 & -0.039 & -1.094 & -0.810 & -0.926 \\ 
	50 & 100 & 0.50 & 1.091 & 0.648 & 0.767 & -0.112 & -0.137 & -0.140 & -1.284 & -0.935 & -0.974 \\ 
	50 & 250 & 0.20 & 1.411 & 0.947 & 0.949 & 0.029 & -0.013 & 0.012 & -1.269 & -0.859 & -0.805 \\ 
	100 & 25 & 4.00 & 1.114 & 0.857 & 0.916 & -0.002 & -0.018 & -0.004 & -1.075 & -0.826 & -0.925 \\ 
	100 & 50 & 2.00 & 1.015 & 0.777 & 0.821 & -0.082 & -0.063 & -0.076 & -1.147 & -0.930 & -0.960 \\ 
	100 & 100 & 1.00 & 1.169 & 0.754 & 0.834 & -0.141 & -0.134 & -0.143 & -1.222 & -0.868 & -0.909 \\ 
	100 & 250 & 0.40 & 1.314 & 0.783 & 0.849 & -0.075 & -0.071 & -0.069 & -1.396 & -0.930 & -0.926 \\ 
	250 & 25 & 10.00 & 1.106 & 0.890 & 0.923 & -0.001 & -0.017 & -0.014 & -1.109 & -0.879 & -0.936 \\ 
	250 & 50 & 5.00 & 1.127 & 0.886 & 0.934 & 0.009 & 0.011 & 0.013 & -1.071 & -0.854 & -0.887 \\ 
	250 & 100 & 2.50 & 1.161 & 0.805 & 0.871 & -0.038 & -0.031 & -0.038 & -1.232 & -0.884 & -0.892 \\ 
	250 & 250 & 1.00 & 1.351 & 0.802 & 0.870 & 0.027 & 0.000 & 0.031 & -1.335 & -0.838 & -0.861 \\ 
		\hline\hline
	\end{longtable}
\end{scriptsize}

\begin{scriptsize} 
	\begin{longtable} {ccc|ccc||ccc||ccc}
		\caption{$\sqrt{nT}$ $\times$ standard error for different estimators with normal errors. Location-scale shift  ($\lambda=1$)}\\
		\label{table_N_1_se}\\
		\hline
		&          &  &\multicolumn{3}{c}{$\tau = 0.25$} & \multicolumn{3}{c}{$\tau = 0.5$} & \multicolumn{3}{c}{$\tau = 0.75$}\\
		\hline
		$n$	&	$T$	&	$n/T$	&	$MD$	&	$FE$	&	$MDT$	&	$MD$	&	$FE$	&	$MDT$	&	$MD$	&	$FE$	&	$MDT$	\\
		\hline
	25 & 25 & 1.00 & 3.222 & 2.791 & 2.841 & 2.661 & 2.572 & 2.526 & 3.186 & 2.857 & 2.870 \\ 
	25 & 50 & 0.50 & 2.909 & 2.744 & 2.739 & 2.574 & 2.528 & 2.487 & 2.963 & 2.750 & 2.757 \\ 
	25 & 100 & 0.25 & 2.771 & 2.688 & 2.701 & 2.511 & 2.462 & 2.469 & 2.736 & 2.658 & 2.634 \\ 
	25 & 250 & 0.10 & 2.707 & 2.646 & 2.669 & 2.462 & 2.427 & 2.438 & 2.725 & 2.628 & 2.628 \\ 
	50 & 25 & 2.00 & 3.195 & 2.806 & 2.853 & 2.655 & 2.533 & 2.507 & 3.203 & 2.837 & 2.810 \\ 
	50 & 50 & 1.00 & 2.934 & 2.733 & 2.699 & 2.608 & 2.547 & 2.520 & 2.832 & 2.702 & 2.686 \\ 
	50 & 100 & 0.50 & 2.788 & 2.667 & 2.680 & 2.493 & 2.448 & 2.430 & 2.782 & 2.664 & 2.701 \\ 
	50 & 250 & 0.20 & 2.694 & 2.644 & 2.633 & 2.510 & 2.474 & 2.479 & 2.711 & 2.656 & 2.661 \\ 
	100 & 25 & 4.00 & 3.332 & 2.799 & 2.813 & 2.703 & 2.574 & 2.566 & 3.233 & 2.774 & 2.837 \\ 
	100 & 50 & 2.00 & 2.853 & 2.682 & 2.678 & 2.570 & 2.477 & 2.480 & 2.960 & 2.703 & 2.740 \\ 
	100 & 100 & 1.00 & 2.800 & 2.712 & 2.674 & 2.494 & 2.434 & 2.440 & 2.749 & 2.652 & 2.628 \\ 
	100 & 250 & 0.40 & 2.706 & 2.678 & 2.660 & 2.375 & 2.358 & 2.360 & 2.655 & 2.601 & 2.599 \\ 
	250 & 25 & 10.00 & 3.302 & 2.833 & 2.860 & 2.675 & 2.504 & 2.533 & 3.398 & 2.809 & 2.847 \\ 
	250 & 50 & 5.00 & 2.948 & 2.714 & 2.736 & 2.678 & 2.580 & 2.552 & 2.926 & 2.736 & 2.770 \\ 
	250 & 100 & 2.50 & 2.788 & 2.709 & 2.706 & 2.421 & 2.342 & 2.370 & 2.728 & 2.644 & 2.641 \\ 
	250 & 250 & 1.00 & 2.695 & 2.652 & 2.638 & 2.460 & 2.441 & 2.434 & 2.676 & 2.650 & 2.627 \\ 
		\hline\hline
	\end{longtable}
\end{scriptsize} 

\newpage

\begin{scriptsize} 
	\begin{longtable} {ccc|ccc||ccc||ccc}
		\caption{$T$ $\times$ bias for different estimators with $t_3$ errors. Location-scale shift ($\lambda=1$)}\\
		\label{table_t_1_bias}\\
		\hline
		&          &  &\multicolumn{3}{c}{$\tau = 0.25$} & \multicolumn{3}{c}{$\tau = 0.5$} & \multicolumn{3}{c}{$\tau = 0.75$}\\
		\hline
		$n$	&	$T$	&	$n/T$	&	$MD$	&	$FE$	&	$MDT$	&	$MD$	&	$FE$	&	$MDT$	&	$MD$	&	$FE$	&	$MDT$	\\
		\hline
	25 & 25 & 1.00 & 2.924 & 1.424 & 0.781 & -0.032 & -0.026 & -0.016 & -3.020 & -1.435 & -0.849 \\ 
	25 & 50 & 0.50 & 2.898 & 1.416 & 0.795 & 0.112 & 0.044 & 0.109 & -2.824 & -1.368 & -0.774 \\ 
	25 & 100 & 0.25 & 2.882 & 1.460 & 0.837 & 0.006 & -0.037 & 0.014 & -2.912 & -1.571 & -1.002 \\ 
	25 & 250 & 0.10 & 3.374 & 1.826 & 1.127 & -0.086 & -0.069 & -0.007 & -3.161 & -1.788 & -1.032 \\ 
	50 & 25 & 2.00 & 2.991 & 1.380 & 0.746 & 0.024 & 0.033 & 0.057 & -3.043 & -1.396 & -0.737 \\ 
	50 & 50 & 1.00 & 2.934 & 1.462 & 0.906 & -0.004 & -0.022 & 0.005 & -2.943 & -1.490 & -0.890 \\ 
	50 & 100 & 0.50 & 3.019 & 1.544 & 0.909 & -0.025 & 0.056 & -0.017 & -3.015 & -1.519 & -0.899 \\ 
	50 & 250 & 0.20 & 2.999 & 1.404 & 0.810 & 0.107 & 0.211 & 0.137 & -3.006 & -1.593 & -0.891 \\ 
	100 & 25 & 4.00 & 2.926 & 1.419 & 0.746 & -0.016 & -0.003 & -0.004 & -3.053 & -1.434 & -0.803 \\ 
	100 & 50 & 2.00 & 3.015 & 1.494 & 0.856 & 0.054 & 0.029 & 0.023 & -2.974 & -1.511 & -0.855 \\ 
	100 & 100 & 1.00 & 3.021 & 1.517 & 0.855 & 0.013 & -0.013 & -0.021 & -2.912 & -1.510 & -0.897 \\ 
	100 & 250 & 0.40 & 3.070 & 1.500 & 0.885 & 0.083 & 0.072 & 0.058 & -3.046 & -1.563 & -0.899 \\ 
	250 & 25 & 10.00 & 3.039 & 1.439 & 0.793 & -0.002 & -0.007 & 0.000 & -3.057 & -1.453 & -0.824 \\ 
	250 & 50 & 5.00 & 2.945 & 1.472 & 0.830 & 0.002 & 0.035 & 0.007 & -2.944 & -1.493 & -0.833 \\ 
	250 & 100 & 2.50 & 2.989 & 1.466 & 0.869 & 0.045 & 0.027 & 0.024 & -3.041 & -1.522 & -0.901 \\ 
	250 & 250 & 1.00 & 3.071 & 1.516 & 0.918 & 0.051 & 0.036 & 0.055 & -2.986 & -1.389 & -0.803 \\ 
		\hline\hline
	\end{longtable}
\end{scriptsize}

\begin{scriptsize} 
	\begin{longtable} {ccc|ccc||ccc||ccc}
		\caption{$\sqrt{nT}$ $\times$ standard error for different estimators with $t_3$ errors. Location-scale shift  ($\lambda=1$)}\\
		\label{table_t_1_se}\\
		\hline
		&          &  &\multicolumn{3}{c}{$\tau = 0.25$} & \multicolumn{3}{c}{$\tau = 0.5$} & \multicolumn{3}{c}{$\tau = 0.75$}\\
		\hline
		$n$	&	$T$	&	$n/T$	&	$MD$	&	$FE$	&	$MDT$	&	$MD$	&	$FE$	&	$MDT$	&	$MD$	&	$FE$	&	$MDT$	\\
		\hline
	25 & 25 & 1.00 & 4.274 & 3.377 & 3.666 & 3.080 & 2.876 & 2.953 & 4.452 & 3.515 & 3.757 \\ 
	25 & 50 & 0.50 & 3.715 & 3.373 & 3.436 & 2.860 & 2.737 & 2.798 & 3.611 & 3.300 & 3.406 \\ 
	25 & 100 & 0.25 & 3.442 & 3.238 & 3.325 & 2.686 & 2.666 & 2.679 & 3.464 & 3.275 & 3.280 \\ 
	25 & 250 & 0.10 & 3.403 & 3.321 & 3.318 & 2.668 & 2.632 & 2.652 & 3.310 & 3.216 & 3.279 \\ 
	50 & 25 & 2.00 & 4.371 & 3.433 & 3.620 & 3.025 & 2.792 & 2.899 & 4.243 & 3.326 & 3.574 \\ 
	50 & 50 & 1.00 & 3.681 & 3.329 & 3.364 & 2.774 & 2.690 & 2.750 & 3.728 & 3.336 & 3.397 \\ 
	50 & 100 & 0.50 & 3.475 & 3.291 & 3.317 & 2.663 & 2.659 & 2.656 & 3.452 & 3.321 & 3.327 \\ 
	50 & 250 & 0.20 & 3.367 & 3.294 & 3.356 & 2.706 & 2.670 & 2.684 & 3.366 & 3.290 & 3.266 \\ 
	100 & 25 & 4.00 & 4.342 & 3.412 & 3.585 & 3.004 & 2.777 & 2.878 & 4.446 & 3.368 & 3.616 \\ 
	100 & 50 & 2.00 & 3.717 & 3.337 & 3.428 & 2.804 & 2.668 & 2.747 & 3.728 & 3.300 & 3.444 \\ 
	100 & 100 & 1.00 & 3.505 & 3.296 & 3.347 & 2.766 & 2.677 & 2.717 & 3.561 & 3.397 & 3.394 \\ 
	100 & 250 & 0.40 & 3.346 & 3.208 & 3.277 & 2.695 & 2.652 & 2.679 & 3.389 & 3.313 & 3.323 \\ 
	250 & 25 & 10.00 & 4.418 & 3.419 & 3.686 & 3.064 & 2.790 & 2.869 & 4.439 & 3.335 & 3.574 \\ 
	250 & 50 & 5.00 & 3.782 & 3.363 & 3.422 & 2.793 & 2.762 & 2.744 & 3.724 & 3.342 & 3.431 \\ 
	250 & 100 & 2.50 & 3.482 & 3.272 & 3.346 & 2.772 & 2.716 & 2.746 & 3.501 & 3.364 & 3.397 \\ 
	250 & 250 & 1.00 & 3.319 & 3.249 & 3.273 & 2.571 & 2.551 & 2.546 & 3.354 & 3.290 & 3.300 \\ 
		\hline\hline
	\end{longtable}
\end{scriptsize}

\newpage

\begin{scriptsize} 
	\begin{longtable} {ccc|ccc||ccc||ccc}
		\caption{$T$ $\times$ bias for different estimators with $\chi_3^2$ errors. Location-scale shift  ($\lambda=1$)}\\
		\label{table_chi_1_bias}\\
		\hline
		&          &  &\multicolumn{3}{c}{$\tau = 0.25$} & \multicolumn{3}{c}{$\tau = 0.5$} & \multicolumn{3}{c}{$\tau = 0.75$}\\
		\hline
		$n$	&	$T$	&	$n/T$	&	$MD$	&	$FE$	&	$MDT$	&	$MD$	&	$FE$	&	$MDT$	&	$MD$	&	$FE$	&	$MDT$	\\
		\hline
25 & 25 & 1.00 & -1.428 & 0.389 & 2.106 & -2.447 & -1.072 & 0.686 & -6.874 & -3.783 & -1.770 \\ 
25 & 50 & 0.50 & -1.533 & 0.098 & 1.720 & -2.682 & -1.259 & 0.371 & -7.136 & -3.987 & -2.302 \\ 
25 & 100 & 0.25 & -1.207 & 0.021 & 1.641 & -2.614 & -1.136 & 0.404 & -7.732 & -4.135 & -2.344 \\ 
25 & 250 & 0.10 & -1.259 & -0.021 & 1.369 & -2.988 & -1.304 & 0.189 & -8.085 & -4.153 & -2.283 \\ 
50 & 25 & 2.00 & -1.625 & 0.227 & 2.026 & -2.640 & -1.158 & 0.558 & -7.454 & -3.987 & -2.022 \\ 
50 & 50 & 1.00 & -1.630 & 0.136 & 1.719 & -2.703 & -1.185 & 0.388 & -7.373 & -4.137 & -2.220 \\ 
50 & 100 & 0.50 & -1.507 & -0.069 & 1.525 & -2.719 & -1.098 & 0.403 & -7.607 & -3.890 & -2.112 \\ 
50 & 250 & 0.20 & -1.307 & -0.009 & 1.346 & -3.107 & -1.237 & 0.164 & -7.682 & -3.618 & -1.732 \\ 
100 & 25 & 4.00 & -1.690 & 0.244 & 2.050 & -2.564 & -1.081 & 0.655 & -7.461 & -4.050 & -2.005 \\ 
100 & 50 & 2.00 & -1.655 & 0.041 & 1.734 & -2.772 & -1.175 & 0.404 & -7.314 & -4.144 & -2.068 \\ 
100 & 100 & 1.00 & -1.501 & -0.006 & 1.556 & -2.839 & -1.221 & 0.383 & -7.632 & -4.104 & -2.161 \\ 
100 & 250 & 0.40 & -1.294 & -0.117 & 1.388 & -2.964 & -1.202 & 0.316 & -8.206 & -4.126 & -2.003 \\ 
250 & 25 & 10.00 & -1.672 & 0.216 & 2.084 & -2.667 & -1.128 & 0.648 & -7.533 & -4.131 & -1.930 \\ 
250 & 50 & 5.00 & -1.658 & 0.055 & 1.759 & -2.825 & -1.192 & 0.417 & -7.388 & -4.127 & -2.053 \\ 
250 & 100 & 2.50 & -1.519 & -0.091 & 1.510 & -2.830 & -1.186 & 0.363 & -7.705 & -4.088 & -2.132 \\ 
250 & 250 & 1.00 & -1.357 & -0.129 & 1.351 & -3.220 & -1.393 & 0.117 & -8.269 & -4.215 & -2.104 \\ 
		\hline\hline
	\end{longtable}
\end{scriptsize} 

\begin{scriptsize} 
	\begin{longtable} {ccc|ccc||ccc||ccc}
		\caption{$\sqrt{nT}$ $\times$ standard error for different estimators with $\chi_3^2$ errors. Location-shift ($\lambda=1$)}\\
		\label{table_chi_1_se}\\
		\hline
		&          &  &\multicolumn{3}{c}{$\tau = 0.25$} & \multicolumn{3}{c}{$\tau = 0.5$} & \multicolumn{3}{c}{$\tau = 0.75$}\\
		\hline
		$n$	&	$T$	&	$n/T$	&	$MD$	&	$FE$	&	$MDT$	&	$MD$	&	$FE$	&	$MDT$	&	$MD$	&	$FE$	&	$MDT$	\\
		\hline
	25 & 25 & 1.00 & 3.897 & 3.725 & 3.785 & 5.721 & 5.315 & 5.442 & 9.980 & 8.193 & 8.384 \\ 
	25 & 50 & 0.50 & 3.772 & 3.605 & 3.634 & 5.495 & 5.297 & 5.376 & 8.905 & 8.195 & 8.209 \\ 
	25 & 100 & 0.25 & 3.533 & 3.502 & 3.496 & 5.306 & 5.272 & 5.233 & 8.666 & 8.289 & 8.338 \\ 
	25 & 250 & 0.10 & 3.554 & 3.489 & 3.494 & 5.163 & 5.094 & 5.103 & 8.342 & 8.083 & 7.990 \\ 
	50 & 25 & 2.00 & 4.025 & 3.720 & 3.774 & 5.774 & 5.341 & 5.422 & 10.229 & 8.447 & 8.542 \\ 
	50 & 50 & 1.00 & 3.614 & 3.525 & 3.580 & 5.545 & 5.310 & 5.420 & 9.110 & 8.261 & 8.378 \\ 
	50 & 100 & 0.50 & 3.593 & 3.470 & 3.484 & 5.264 & 5.202 & 5.122 & 8.590 & 8.190 & 8.143 \\ 
	50 & 250 & 0.20 & 3.691 & 3.563 & 3.632 & 5.379 & 5.341 & 5.322 & 8.461 & 8.224 & 8.247 \\ 
	100 & 25 & 4.00 & 3.999 & 3.691 & 3.842 & 5.807 & 5.402 & 5.577 & 10.543 & 8.356 & 8.641 \\ 
	100 & 50 & 2.00 & 3.744 & 3.515 & 3.591 & 5.431 & 5.239 & 5.224 & 8.965 & 8.142 & 8.233 \\ 
	100 & 100 & 1.00 & 3.597 & 3.474 & 3.507 & 5.302 & 5.148 & 5.155 & 8.565 & 8.183 & 8.200 \\ 
	100 & 250 & 0.40 & 3.633 & 3.568 & 3.601 & 5.384 & 5.201 & 5.298 & 8.389 & 8.179 & 8.200 \\ 
	250 & 25 & 10.00 & 4.141 & 3.744 & 3.878 & 5.698 & 5.331 & 5.396 & 10.486 & 8.288 & 8.432 \\ 
	250 & 50 & 5.00 & 3.788 & 3.677 & 3.699 & 5.471 & 5.417 & 5.409 & 9.204 & 8.187 & 8.347 \\ 
	250 & 100 & 2.50 & 3.687 & 3.602 & 3.671 & 5.370 & 5.250 & 5.273 & 8.828 & 8.198 & 8.377 \\ 
	250 & 250 & 1.00 & 3.578 & 3.503 & 3.532 & 5.120 & 5.070 & 5.023 & 8.065 & 7.949 & 7.917 \\ 
		\hline\hline
	\end{longtable}
\end{scriptsize}

\newpage

\section{Conclusion}\label{sec:conc}

Asymptotic theory for panel data quantile regression (QR) with fixed effects poses many challenges as it involves models with an increasing number of parameters and a non-smooth objective function. Owing to those difficulties, unbiased asymptotic normality of estimators for common parameters in panel data QR has so far only been known to hold under stringent conditions on the length of panels relative to the number of individuals. Specifically, \cite{KatoGalvaoMontes-Rojas12} proved $\sqrt{nT}$-consistency of a panel data fixed effect QR (FE-QR) estimator under the stringent condition that $n^{2}(\log n)^{3}/T\rightarrow0$, and since then, it has been an open question whether the rates on the sample size requirement could be improved to a condition which is closer to the assumption $n/T = o(1)$ which is known to be sufficient in nonlinear panel data models with smooth objective function. 

The major contribution of this paper was to show that, for both the FE-QR and MD-QR estimators, such an improvement is indeed possible in a wide range of scenarios including observations with temporal dependence.

Our results are important to practitioners and theorists. The main practical implication is that, despite a lack of smoothness of the objective function, panel data QR is applicable for the same type of panel dimensions as other popular nonlinear models. This validates the use of QR in many practical scenarios in which its validity previously lacked theoretical justification.  Our theory also provides grounds for subsequent methodological research on panel data QR which now can rely on an improved and more realistic growth rate on the sample size. We believe that the proof techniques provided here will also be useful to future researchers. 

There are ample directions for future research. For instance, we have not considered censored observations or scenarios with endogeneity. Another topic that merits further exploration is related to the bootstrap and its refinements.

\newpage

\bibliographystyle{econometrica}
\bibliography{pqrb}

\newpage

\linespread{1.35}

\small

\newpage

\appendix

\section{Appendix}

\subsection{Technical results used in proof for MD-QR and FE-QR}

\subsubsection{Results in the i.i.d. case}

\begin{lemma}
\label{VCClemma}
Under Assumptions (I) and (A0)--(A3), 
\begin{equation*}
\widehat{\bm{\gamma}}_{i}(\tau) - \bm{\gamma}_{i0}(\tau) = -\frac{1}{T}B_i^{-1}\sum_{t=1}^{T} \Z_{it}(\1(Y_{it} \leq q_{i,\tau}(\Z_{it})) - \tau)+R_{iT}^{(1)}(\tau) + R_{iT}^{(2)}(\tau),
\end{equation*}
with 
\begin{align}
\label{vcc1}&\sup_{i}\sup_{\eta\in \T}\|R_{iT}^{(2)}(\eta)\| = O_p\Big(\frac{\log T}{T}\Big),\\
\label{vcc2}&\sup_{i}\sup_{\eta\in \T}\|R_{iT}^{(1)}(\eta)\|  = O_p\Big(\Big(\frac{\log T}{T}\Big)^{3/4}\Big),\\
\label{vcc3}&\sup_{i}\sup_{\eta\in \T}\|\E[R_{iT}^{(1)}(\eta)]\| = O\Big(\frac{\log T}{T}\Big),\\
\label{vcc4}&\sup_{i}\sup_{\eta\in \T}\|R_{iT}^{(1)}(\eta)\| = O(1) \quad a.s.,\\
\label{vcc5}&\sup_{i}\sup_{\eta\in \T}\Big\|\E\Big[\Big(R_{iT}^{(1)}(\eta)-\E[R_{iT}^{(1)}(\eta)]\Big)\Big(R_{iT}^{(1)}(\eta)-\E[R_{iT}^{(1)}(\eta)]\Big)^\top\Big]\Big\| = O\Big(\Big(\frac{\log T}{T}\Big)^{3/2}\Big).
\end{align}
\end{lemma}
\begin{proof}
Observe the decomposition
\begin{equation*}
\widehat{\bm{\gamma}}_{i}(\eta) - \bm{\gamma}_{i0}(\eta) = -\frac{1}{T}B_i^{-1}\sum_{t=1}^{T} \Z_{it}(\1(Y_{it} \leq q_{i,\eta}(\Z_{it})) - \eta) + r_{i,1}(\eta) + r_{i,2}(\eta) + r_{i,3}(\eta),
\end{equation*}
where
\begin{align*}
r_{i,1}(\eta) &:= \frac{1}{T}B_i^{-1}\sum_{t=1}^{T} \Z_{it}(\1(Y_{it} \leq \Z_{it}^\top \widehat{\bm{\gamma}}_{i}(\eta)) - \eta),
\\
r_{i,2}(\eta) &:= - \frac{1}{T}B_i^{-1} \sum_{t=1}^{T} \Big\{\Z_{it}\Big( \1(Y_{it} \leq \Z_{it}^\top \widehat{\bm{\gamma}}_{i}(\eta)) -\1(Y_{it} \leq \Z_{it}^\top \bm{\gamma}_{i0}(\eta)) \Big)
\\
& \quad \quad \quad \quad \quad \quad \quad \quad- \int z[F_{Y|Z}(z^\top \widehat{\bm{\gamma}}_{i}(\eta) \mid z) - F_{Y|Z}(z^\top {\bm{\gamma}}_{i0}(\eta)\mid z)] dP^{\Z_{i1}}(z)\Big\},
\\
r_{i,3}(\eta) &:= - B_i^{-1}\Big[ \int z[F_{Y|Z}(z^\top \widehat{\bm{\gamma}}_{i}(\eta) \mid z) - F_{Y|Z}(z^\top {\bm{\gamma}}_{i0}(\eta)\mid z)] dP^{\Z_{i1}}(z) - B_i (\widehat{\bm{\gamma}}_{i}(\eta) - \bm{\gamma}_{i0}(\eta))\Big].
\end{align*}
Let $R_{iT}^{(1)}(\eta) := r_{i,2}(\eta), R_{iT}^{(2)}(\eta) := r_{i,1}(\eta) + r_{i,3}(\eta)$.

The first four statements \eqref{vcc1}--\eqref{vcc4} follow from Theorem S.6.1 of \cite{VCC}. More precisely, note that in the notation of \cite{VCC} under the assumptions (I) and (A1)--(A3) we have $g_N = 0, c_N = 0$, $\xi_m, m$ are constant. Apply the union bound to handle the $\sup_i$ and choose $\kappa_n = C\log T$ in \cite{VCC} for a suitable constant $C$.

To prove \eqref{vcc5}, denote the $j$-th element of the vector $R_{iT}^{(1)}(\eta)-\E[R_{iT}^{(1)}(\eta)]$ by $r_{j,i,\eta}$. By \eqref{vcc4} $\sup_{i,j}\sup_{\eta\in \T}|r_{j,i,\eta}| \leq \mathcal{C}_7$ and by Theorem S.6.1 in \cite{VCC}  
\begin{equation*}
\P\Big(\sup_{\eta \in \T} |r_{j,i,\eta}|\geq \mathcal{C}_8 \Big(\frac{(\kappa+1)\log T}{T}\Big)^{3/4}\Big) \leq 2 T^{-\kappa},
\end{equation*}
with $\kappa$ and $\mathcal{C}_8$ independent of $j$ and $i$. Hence, take $\kappa = 2$
\begin{equation*}
\E[r_{j,i,\eta}^2] \leq \mathcal{C}_7^2 \P\Big(|r_{j,i,\eta}| > \mathcal{C}_8 \Big(\frac{3\log T}{T}\Big)^{3/4}\Big) + \E\Big[r_{j,i,\eta}^2 \1\Big\{|r_{j,i,\eta}|\leq \mathcal{C}_8 \Big(\frac{3\log T}{T}\Big)^{3/4}\Big\}\Big] \leq  \frac{2\mathcal{C}_7^2}{T^2} + \mathcal{C}_8^2 \Big(\frac{3\log T}{T}\Big)^{3/2}.
\end{equation*}
Since the covariance matrix is of fixed dimension, \eqref{vcc5} follows.
\end{proof}

\subsubsection{Results in the dependent case}

For a class of functions $\mathcal{G}$ we use the following notation
\begin{align*}
\|\mathbb{P}_{i,T} - \mathbb{P}_i\|_{\mathcal{G}} &:= \sup_{g\in \mathcal{G}}\Big |\frac{1}{T}\sum_{t=1}^{T} (g(Y_{it}, \Z_{it}) - \E[g(Y_{i1},\Z_{i1})])\Big |,
\\
\sigma_{q,i}(f) &:= Var \Big(\frac{1}{\sqrt{q}} \sum_{t=1}^q f(Y_{it},\Z_{it})\Big).
\end{align*}

We proceed in the same way as in the previous section and begin by stating and proving several technical lemmas for the dependent case that are useful in the demonstration of Theorem \ref{th2}.

\begin{lemma}\label{varb}
Consider a function $f: \mathbb{R}^{p+2} \to \mathbb{R}$ with $\|f\|_{\infty} \leq \C_8$ and such that for all $i=1,...,n$ $\E[f(Y_{i1},\Z_{i1})] = 0$ and $Var(f(Y_{i1},\Z_{i1})) \leq \delta <1$. Further assume that (A0), (D1) hold. Then 
\begin{equation*}
\sigma_{q,i}(f) \leq \bar{\C}\delta ( 1 + | \log \delta|)
\end{equation*} 
for a constant $\bar{\C}$ depending on $\C_8$, $\C_\beta$ and $b_\beta$ only. 
\end{lemma}
\begin{proof}[Proof of Lemma~\ref{varb}] Throughout the proof use the short-hand notation $\xi_t := (Y_{it},\Z_{it})$ dropping the dependence on $i$ (note that all constants are independent of $i$). First the properties of the function $f$ lead to the following bounds ( holding for any fixed positive integer $j$) 
\begin{align*}
\E[|f(\xi_1)|^2] \E[|f(\xi_{1+j})|^2] &= Var(f(\xi_{1}))^2 \leq \delta^2,
\\
\E[|f(\xi_1) f(\xi_{1+j})|^2] &\leq \|f\|_{\infty}^2 Var(f(\xi_{1})) \leq \C_8^2 \delta. 
\end{align*}
Let $H_\delta = \max\{\delta^2, \C_8^2 \delta\}$. Apply Lemma C.1 in \cite{KatoGalvaoMontes-Rojas12} with their $\delta = 1$ to obtain 
\begin{equation*}
|Cov(f(\xi_1), f(\xi_{1+j}))| \leq 4 H_\delta^{1/2}(\beta(j))^{1/2}.
\end{equation*}
On the other hand, we also have 
\begin{equation*}
|Cov(f(\xi_1), f(\xi_{1+j}))| = |\E[f(\xi_1) f(\xi_{1+j})]| \leq Var(f(\xi_1)) \leq \delta.
\end{equation*}
Hence 
\begin{equation*}
|Cov(f(\xi_1), f(\xi_{1+j}))| \leq \min\{4 H_{\delta}^{1/2}(\beta(j))^{1/2}, \delta\}.
\end{equation*}
Now we know that 
\begin{equation*}
Var\Big(\frac{1}{\sqrt{q}}\sum_{t=1}^{q} f(\xi_t)\Big) = Var(f(\xi_1)) + 2\sum_{j=1}^{q-1} \Big(1-\frac{j}{q}\Big) Cov(f(\xi_1), f(\xi_{1+j})),
\end{equation*}
therefore, for any $j_0 \geq 1$,
\begin{align*}
Var\Big(\frac{1}{\sqrt{q}}\sum_{t=1}^{q} f(\xi_t)\Big)  & \leq \delta + 2 \sum_{j=1}^{q-1} \Big(1-\frac{j}{q}\Big) \min \{4 H_\delta^{1/2}(\beta(j))^{1/2}, \delta\}\\
& \leq \delta \Big (1 +2 \sum_{j=1}^{q-1}  \min \Big\{1, \frac{4H_\delta^{1/2}}{\delta} (\beta(j))^{1/2}\Big\}\Big)\\
& \leq \delta \Big ( 1+ 2\sum_{j=1}^{q-1} \min \Big \{ 1, 4\Big(\frac{\C_\beta (\C_8+1)b_\beta^{j_0}}{\delta}\Big)^{1/2} (b_\beta^{1/2})^{j-j_0}\Big\}\Big)\\
& \leq \bar{\C} \delta(1 + |\log \delta|),
\end{align*}
where $\bar{\C}$ only depends on $\C_8$, $\C_9$ and $b_\beta$. The second last inequality is obtained since 
\begin{equation*}
\frac{4H_{\delta}^{1/2}}{\delta}(\beta(j))^{1/2} \leq 4 c_{\beta}^{1/2} (b_{\beta}^{1/2})^{j} \Big(\max \{1, \frac{C_8}{\delta}\}\Big)^{1/2} \leq 4\Big(\frac{\C_\beta (\C_8+1)b_\beta^{j_0}}{\delta}\Big)^{1/2} (b_\beta^{1/2})^{j-j_0}.
\end{equation*}
The last inequality is obtained by choosing $j_0$ to be the smallest integer such that 
\begin{equation*}
4\Big(\frac{\C_\beta (\C_8+1)b_\beta^{j_0}}{\delta}\Big)^{1/2} \leq 1,
\end{equation*}
and thus for $C\leq 1$, 
\begin{equation*}
2\sum_{j=1}^{q-1} \min \Big \{ 1, C (b_\beta^{1/2})^{j-j_0}\Big\} \leq 2 \Big \{\Big(\sum_{j=1}^{j_0}1\Big) + C \sum_{j=j_0+1}^{\infty}(b_{\beta}^{1/2})^{j-j_0}\Big\}\lesssim j_0 + 1,
\end{equation*}
and since the chosen $j_0$ is of order $O(|\log \delta|)$, the last inequality holds. 
\end{proof}

\bigskip

\begin{lemma}
\label{lem:G1G2}
Consider the classes of functions 
\begin{align*}
\mathcal{G}_1 & := \Big\{ (y,z) \mapsto a^\top z (\1\{y\leq b^\top z\} - \eta)\1\{\|z\|\leq M\} \Big| b \in \mathbb{R}^{p+1},\eta \in \mathcal{T}, a\in \mathcal{S}^{p+1} \Big\},
\\
\mathcal{G}_2(\delta) &:= \Big\{ (y,z) \mapsto a^\top z (\1\{y\leq b_1^\top z\} - \1\{y\leq b_2^\top z\})\1\{\|z\|\leq M\} \Big| \|b_1-b_2\| \leq \delta, a\in \mathcal{S}^{p+1} \Big\},
\\
\mathcal{G}_{3,k,\ell}(\delta) &:= \Big\{ (y_1,z_1,y_2,z_2) \mapsto g_{b_1,k,\ell}(y_1,z_1,y_2,z_2)- g_{b_2, k, \ell}(y_1,z_1,y_2,z_2)\Big| \|b_1 - b_2 \|\leq \delta, b_1,b_2 \in \mathbb{R}^{p+1}\Big\},
\end{align*}
where
\begin{equation*}
g_{b,k,\ell}(y_1, z_1, y_2, z_2) := z_{1,k}z_{2,\ell} (\tau - \1\{y_1 \leq z_1^\top b \})(\tau - \1\{y_2 \leq z_2^\top b\}).
\end{equation*}
Under assumptions (A0)--(A3) and (D1) there exists a constant $c_0$ which is independent of $n, T, i$ such that for any $\kappa > 1$
\begin{align}
P\Big( \|\mathbb{P}_{i,T} - \mathbb{P}_i\|_{\mathcal{G}_1} \geq c_0 \kappa^{1/2} \Big( \frac{\log T}{T}\Big)^{1/2} \Big) &\leq T^{-\kappa}, \label{G1_beta}
\\
P\Big(\sup_{0\leq \delta\leq 1}\frac{\|\mathbb{P}_{i,T} - \mathbb{P}_i\|_{\mathcal{G}_2(\delta)}}{\chi_T(\delta)} \geq c_0\kappa^2 \Big) &\leq T^{-\kappa}, \label{G2_beta} 
\end{align}
where
\begin{equation}
\chi_T(\delta) := T^{-1/2}\delta^{1/2}\log T + T^{-1}(\log T)^{2}. \label{def:chi}
\end{equation}
If moreover assumption (D2) holds, defining
\begin{equation} \label{STjkl}
S_{T,j,k,\ell}(\delta) := \sup_i \sup_{g \in \mathcal{G}_{3,k,\ell}(\delta)}\Big|\frac{1}{|T_j|}\sum_{t\in T_j} g(Y_{it},\Z_{it},Y_{ij+t},\Z_{ij+t}) - \E[ g(Y_{it},\Z_{it},Y_{ij+t},\Z_{ij+t})] \Big|,
\end{equation}
with $T_j:=\{t | 1\leq t \leq T, 1\leq t+j \leq T\}$, then for any $1\leq j \leq m_T, 1\leq k,\ell \leq p+1$,
\begin{equation}
\P\Big(S_{T,j,k,\ell}(\delta) \geq c_0\kappa^2 \Big( T^{-1/2}\delta^{1/2}(\log T)^{1/2}(m_T + \log T)^{1/2} + T^{-1}(\log T)(m_T + \log T) \Big) \Big) \leq T^{-\kappa}. \label{G3_beta}
\end{equation}
\end{lemma}

\begin{proof}[Proof of Lemma~\ref{lem:G1G2}] 
First, for functions $g_1$ belonging to the class $\mathcal{G}_1$, it is easy to show that $\|g_1\|_{\infty} \leq U_1$ and $\sup_i \sup_{g_1\in \mathcal{G}_1} Var(g_1(Y_{i1}, \Z_{i1})) \leq c_1$ for some constants $U_1,c_1 <\infty$. Applying Lemma~\ref{varb} to $g_1/(2c_1)$, we have for any integer $q \geq 1$ that $\sigma_q^2(g_1) \leq u_1$ for some constant $u_1 < \infty$. Finally, note that for any probability measure $Q$ and any $0<\epsilon<1$ we have
\begin{equation*}
N(\mathcal{G}_1,L_1(Q),\epsilon) \leq N(\mathcal{G}_1,L_2(Q),\epsilon) \leq (A/\epsilon)^v, 
\end{equation*}
for some constants $A,v < \infty$; here the first inequality follows by an application of the Cauchy-Schwarz inequality and the second inequality follows by similar arguments as given in the proof of Lemma C.3 in \cite{CVC}. Now apply Proposition C.2 of \cite{KatoGalvaoMontes-Rojas12} to find that for a constant $C$ independent of $T,n,i,q,\kappa$ we have for any $s_T,q_T$ such that $q_T^2\log q_T = o(T)$
\begin{equation*}
\P\Big(T \|\mathbb{P}_{i,T} - \mathbb{P}_{i}\|_{\mathcal{G}_1} \geq C \Big(\sqrt{T\log T}+ \sqrt{s_TT} + s_T q_T\Big) \Big) \leq 2e^{-s_T}+2T\beta(q_T).
\end{equation*}
Now let $q_T = c_{q_1} \kappa\log T$, $s_T = c_{s_1} \kappa \log T$ where the constants $c_{q_1}, c_{s_1}$ are chosen such that $2e^{-s_T}+2T\beta(q_T) \leq T^{-\kappa}$. This shows~\eqref{G1_beta}. 

At the cost of changing constants we will prove~\eqref{G2_beta} for $g_2/(4 M^3  f_{\max})^{1/2}$ for functions $g_2 \in \mathcal{G}_2(\delta)$. We have $\|g_2/(4 M^3  f_{\max})^{1/2}\|_{\infty} \leq U_2$ under assumption (A1), and under Assumptions (A1), (A2), 
\begin{align*}
Var\Big(g_2(Y_{i1}, \Z_{i1})/(4 M^3  f_{\max})^{1/2}\Big) &\leq 4 M^2 \E[|F_{Y\mid \Z}(b_1^\top \Z_{i1}) - F_{Y\mid \Z}(b_2^\top \Z_{i1})|]/(4 M^3  f_{\max}) 
\leq \delta.
\end{align*}
We begin by assuming that $1 \geq \delta \geq 1/T$. Finally, note that for any probability measure $Q$ and any $0<\epsilon<1, 1 \geq \delta > 0$ we have
\begin{equation*}
N(\mathcal{G}_2(\delta),L_1(Q),\epsilon) \leq N(\mathcal{G}_2(1),L_1(Q),\epsilon) \leq N(\mathcal{G}_2(1),L_2(Q),\epsilon) \leq (A/\epsilon)^v, 
\end{equation*}
for some constants $A,v < \infty$; here the second inequality follows by an application of the Cauchy-Schwarz inequality and the third inequality follows by similar arguments as given in the proof of Lemma C.3 in \cite{CVC}. Apply Lemma~\ref{varb} to obtain that we have $\sigma_{q}^2(g_2)\leq \bar{C} \delta \log T$. Next apply Proposition C.2 of \cite{KatoGalvaoMontes-Rojas12} to the class $\mathcal{G}_2(\delta)$ to find that for a constant $C_1$ independent of $T,n,i,q,\kappa$ we have for any $s_T,q_T,\sigma_T^2(\delta)$ such that $\sigma_{T}^2(\delta) \geq \sup_{g_2 \in {\mathcal{G}}_2(\delta)} \sigma_q^2(g_2)$ and $q_T^2 \log T \leq \widetilde{c}T\sigma_T^2(\delta)$ for a fixed constant $\widetilde{c}$ which is independent of $n,T,i,\delta$ 
\begin{equation*}
\P\Big(T \|\mathbb{P}_{i,T}-\mathbb{P}_{i}\|_{\mathcal{G}_2(\delta)} \geq C_1 \Big( \sqrt{T} \sigma_{T}(\delta) \sqrt{\log T} + \sigma_{T}(\delta) \sqrt{s_T T} + s_T q_T \Big)\Big) \leq 2 e^{-s_T} + 2T\beta(q_T).
\end{equation*}
Pick $q_T = c_q \kappa \log T$ and $s_T = c_s \kappa \log T$ with $c_q, c_s$ such that $2 e^{-s_T} + 2T\beta(q_T) \leq T^{-(\kappa+1)}$ and let $\sigma_{T}^2(\delta) = \max \{\delta \log T, c_q^2\kappa^2(\log T)^3/(\widetilde{c}T)\}$. This shows the existence of a constant $\widetilde{c}_0$ such that 
\begin{equation}\label{G2_beta:prel}
\P\Big(T\|\mathbb{P}_{i,T} - \mathbb{P}_i\|_{\mathcal{G}_2(\delta)}\geq \widetilde{c}_0\kappa^2  \Big(T^{1/2}\delta^{1/2}\log T + (\log T)^{2} \Big)\Big) \leq T^{-(\kappa+1)}.
\end{equation}
Now to prove~\eqref{G2_beta} note that $\chi_T(u)$ is decreasing in $u$ and $\chi_T(u/2) \geq \chi_T(u)/2$, $\chi_T(0) \geq \chi_T(1/T)/2$ and note that
\begin{align*}
& \sup_{0\leq \delta\leq 1} \frac{\|\mathbb{P}_{i,T} - \mathbb{P}_i\|_{\mathcal{G}_2(\delta)}}{\chi_T(\delta)}
\\
& \leq \Big( \sup_{0\leq \delta\leq 1/T} \frac{\|\mathbb{P}_{i,T} - \mathbb{P}_i\|_{\mathcal{G}_2(\delta)}}{\chi_T(\delta)}\Big) \vee \Big(\max_{k: T^{-1} \leq 2^{-k} \leq 1} \sup_{2^{-k-1}\leq \delta\leq 2^{-k}} \frac{\|\mathbb{P}_{i,T} - \mathbb{P}_i\|_{\mathcal{G}_2(\delta)}}{\chi_T(\delta)}\Big)
\\
& \leq \Big(\frac{2\|\mathbb{P}_{i,T} - \mathbb{P}_i\|_{\mathcal{G}_2(1/T)}}{\chi_T(1/T)}\Big) \vee \Big(\max_{k: T^{-1} \leq 2^{-k} \leq 1} \frac{2\|\mathbb{P}_{i,T} - \mathbb{P}_i\|_{\mathcal{G}_2(2^{-k})}}{\chi_T(2^{-k})}\Big).
\end{align*}
Now the last line contains the maximum of $O(\log T)$ elements, and so by adjusting the constant $\widetilde{c}_0$ in~\eqref{G2_beta:prel} and applying the union bound we obtain~\eqref{G2_beta}.

Next we prove~\eqref{G3_beta}. Simple calculation shows that under (D2) for any $j \geq 1$
\begin{equation*}
\E\Big[\Big(g_{b_1,k,\ell}(Y_{i1}, \Z_{i1}, Y_{i1+j}, \Z_{i1+j}) - g_{b_2,k,\ell}(Y_{i1}, \Z_{i1}, Y_{i1+j}, \Z_{i1+j})\Big)^2\Big] \leq C\|b_1 - b_2\|,
\end{equation*}
for a constant $C$ independent of $i,j, k, \ell$. Next, note that there exist constants $A, v$ such that for all $k,\ell$
\begin{equation*}
N\Big(\mathcal{G}_{3,k,\ell}(\delta), L_1(Q), \epsilon\Big) \leq (A/\epsilon)^v.
\end{equation*}
To see this, observe that any function in $\mathcal{G}_{3,k,\ell}(\delta)$ can be expressed as through sums and products of functions from the classes $\mathcal{H}_1 := \{(y_1,z_1,y_2,z_2) \mapsto z_{1,k}z_{2,\ell}| 1 \leq k,\ell \leq p+1\}, \mathcal{H}_2 := \{(y_1,z_1,y_2,z_2) \mapsto \tau - \1\{y_1 \leq z_1^\top b \} | b \in \mathbb{R}^{p+1}\}$, $\mathcal{H}_3 := \{(y_1,z_1,y_2,z_2) \mapsto \tau - \1\{y_2 \leq z_2^\top b \} | b \in \mathbb{R}^{p+1}\}$ and that each of the three classes satisfies
\begin{equation*}
N(\mathcal{H}_j,L_2(Q),\epsilon) \leq (\widetilde{A}/ \epsilon)^{\widetilde{v}}, 
\end{equation*}
for all $0 < \epsilon\leq 1$ and some constants $\widetilde{A}, \widetilde{v} < \infty$. Hence, by the Cauchy-Schwarz inequality and Lemma 24 in \cite{BCCF} (note that the proof of this Lemma continues to hold for arbitrary probability measures, discreteness is not required), we find that
\begin{equation*}
N(\mathcal{G}_{3,k,\ell}(\delta),L_1(Q),\epsilon) \leq N(\mathcal{G}_{3,k,\ell}(1),L_1(Q),\epsilon)  \leq (A/\epsilon)^v, 
\end{equation*}
for some $A,v < \infty$. Finally, note that under (D1) the series of random vectors $\{(Y_{it}, \Z_{it}, Y_{it+j}, \Z_{it+j})\}_{t\in\mathbb{Z}}$ is $\beta$-mixing with mixing coefficients $\widetilde{\beta}(t)$ satisfying $\widetilde{\beta}(t) \leq \beta(0\vee (t-j))$. Now similar arguments as those used to prove Lemma \ref{varb} show that for $g\in \mathcal{G}_{3,k,\ell}(\delta)$
\begin{equation*}
\sigma_{q,i,j}^2(g) := Var \Big(\frac{1}{\sqrt{q}} \sum_{t=1}^q f(Y_{it},\Z_{it},Y_{it+j},\Z_{it+j})\Big),
\end{equation*}
we have for $1 \geq \delta \geq 1/T$ and a constant $C$ independent of $n,T,\delta$ 
\begin{equation*}
\max_{1\leq j \leq m_T} \sup_i \sup_{g\in \mathcal{G}_{3,k,\ell}(\delta)} \sup_{q \geq 1}\sigma_{q,i,j}^2(g) \leq C \delta (m_T + \log T).
\end{equation*}
Hence we can apply Proposition C.2 of \cite{KatoGalvaoMontes-Rojas12} by picking $q_T = c_q \kappa(m_T +  \log T)$ and by similar arguments as for deriving \eqref{G2_beta:prel} we obtain~\eqref{G3_beta}. 
\end{proof}

\begin{lemma}
\label{VCClemma_beta}
Under Assumptions (A0)--(A3) and (D1) 
\begin{equation*}
\widehat{\bm{\gamma}}_{i}(\tau) - \bm{\gamma}_{i0}(\tau) = -\frac{1}{T}B_i^{-1}\sum_{t=1}^{T} \Z_{it}(\1(Y_{it} \leq \Z_{it}^\top \bm{\gamma}_{i}(\tau)) - \tau)+R_{iT}^{(1)}(\tau) + R_{iT}^{(2)}(\tau),
\end{equation*}
with 
\begin{align}
\label{vcc1_beta}&\sup_{i}\sup_{\eta\in \T}\|R_{iT}^{(2)}(\eta)\| = O_p\Big(\frac{\log T}{T}\Big),\\
\label{vcc2_beta}&\sup_{i}\sup_{\eta\in \T}\|R_{iT}^{(1)}(\eta)\|  = O_p\Big(\Big(\frac{(\log T)^{5/4}}{T^{3/4}}\Big)\Big),\\
\label{vcc3_beta}&\sup_{i}\sup_{\eta\in \T}\|\E[R_{iT}^{(1)}(\eta)]\| = O\Big(\frac{(\log T)^2}{T}\Big),\\
\label{vcc4_beta}&\sup_{i}\sup_{\eta\in \T}\|R_{iT}^{(1)}(\eta)\| = O(1),\\
\label{vcc5_beta}&\sup_{i}\sup_{\eta\in \T}\Big\|\E\Big[\Big(R_{iT}^{(1)}(\eta)-\E[R_{iT}^{(1)}(\eta)]\Big)\Big(R_{iT}^{(1)}(\eta)-\E[R_{iT}^{(1)}(\eta)]\Big)^\top\Big]\Big\| = O\Big(\Big(\frac{(\log T)^{5/2}}{T^{3/2}}\Big)\Big).
\end{align}
\end{lemma}

\begin{proof}

Observe the decomposition
\begin{equation*}
\widehat{\bm{\gamma}}_{i}(\eta) - \bm{\gamma}_{i0}(\eta) = -\frac{1}{T}B_i^{-1}\sum_{t=1}^{T} \Z_{it}(\1(Y_{it} \leq q_{i,\eta}(\Z_{it})) - \eta) + r_{i,1}(\eta) + r_{i,2}(\eta) + r_{i,3}(\eta),
\end{equation*}
where
\begin{align*}
r_{i,1}(\eta) &:= \frac{1}{T}B_i^{-1}\sum_{t=1}^{T} \Z_{it}(\1(Y_{it} \leq \Z_{it}^\top \widehat{\bm{\gamma}}_{i}(\eta)) - \eta),
\\
r_{i,2}(\eta) &:= - \frac{1}{T}B_i^{-1} \sum_{t=1}^{T} \Big\{\Z_{it}\Big( \1(Y_{it} \leq \Z_{it}^\top \widehat{\bm{\gamma}}_{i}(\eta)) -\1(Y_{it} \leq \Z_{it}^\top \bm{\gamma}_{i0}(\eta)) \Big)
\\
& \quad \quad \quad \quad \quad \quad \quad \quad- \int z[F_{Y|Z}(z^\top \widehat{\bm{\gamma}}_{i}(\eta) \mid z) - F_{Y|Z}(z^\top {\bm{\gamma}}_{i0}(\eta)\mid z)] dP^{\Z_{i1}}(z)\Big\},
\\
r_{i,3}(\eta) &:= - B_i^{-1}\Big[ \int z[F_{Y|Z}(z^\top \widehat{\bm{\gamma}}_{i}(\eta) \mid z) - F_{Y|Z}(z^\top {\bm{\gamma}}_{i0}(\eta)\mid z)] dP^{\Z_{i1}}(z) - B_i (\widehat{\bm{\gamma}}_{i}(\eta) - \bm{\gamma}_{i0}(\eta))\Big].
\end{align*}
Let $R_{iT}^{(1)}(\eta) := r_{i,2}(\eta), R_{iT}^{(2)}(\eta) := r_{i,1}(\eta) + r_{i,3}(\eta)$.

Now similar arguments as in the proof of Theorem 5.1 in \cite{CVC} (noting that, in the notation of \cite{CVC} $g_n = 0, c_n(\gamma_n) = 0$, $\xi_m = M$ and $m = p+1$) using Lemma~\ref{lem:G1G2} in the present paper instead of Lemma C.3 in \cite{CVC} leads to~\eqref{vcc1_beta} and the same arguments show that with probability tending to one,
\begin{align*}
\sup_i\sup_{\eta \in \T} \|R_{iT}^{(1)}\| \leq \chi_T((\log T)^{1/2}/T^{1/2}) = O\Big(\frac{(\log T)^{5/4}}{T^{3/4}}\Big),
\end{align*} 
so \eqref{vcc2_beta} holds. 


Moreover, \eqref{vcc4_beta} follows from the definition of $r_{i,2}(\eta)$ and assumption (A1), while \eqref{vcc5_beta} can be proved by the same arguments as \eqref{vcc5} (again using Lemma~\ref{lem:G1G2} in the present paper instead of Lemma C.3 in \cite{CVC}). 


It remains to prove~\eqref{vcc3_beta}. For $t = 1,...,T$ define $J_t := \{1\leq s\leq T: |s-t|\leq L \log T\}$ with a constant $L$ to be determined later. Let
\begin{equation*}
\widehat{\bm{\gamma}}_{i}^{(-J_s)}(\eta) := \argmin_{\bm{\gamma} \in \mathbb{R}^{p+1}} \sum_{t=1, t \notin J_s}^{T} \rho_{\eta}( Y_{it}- \Z_{it}\tr \bm{\gamma}).
\end{equation*}
Similarly as before, we have the decomposition
\begin{align*}
\widehat{\bm{\gamma}}_{i}^{(-J_s)}(\eta) - \bm{\gamma}_{i0}(\eta) =  -\frac{1}{T - |J_s|}B_i^{-1}\sum_{t \notin J_s} \Z_{it}(\1(Y_{it} \leq q_{i,\eta}(\Z_{it})) - \eta) + r_{i,1}^{(-J_s)}(\eta) + r_{i,2}^{(-J_s)}(\eta) + r_{i,3}^{(-J_s)}(\eta),
\end{align*}
where
\begin{align*}
r_{i,1}^{(-J_s)}(\eta) & := \frac{1}{T-|J_s|}B_i^{-1}\sum_{t \notin J_s} \Z_{it}(\1(Y_{it} \leq \Z_{it}^\top \widehat{\bm{\gamma}}^{(-J_s)}_i(\eta)) - \eta),
\\
r_{i,3}^{(-J_s)}(\eta) & :=  - B_i^{-1}\Big[ \int z[F_{Y|Z}(z^\top \widehat{\bm{\gamma}}^{(-J_s)}_{i}(\eta) \mid z) - F_{Y|Z}(z^\top {\bm{\gamma}}_{i0}(\eta)\mid z)] dP^{\Z_{i1}}(z) - B_i (\widehat{\bm{\gamma}}^{(-J_s)}_{i}(\eta) - \bm{\gamma}_{i0}(\eta))\Big],
\end{align*}
and
\begin{align*}
r_{i,2}^{(-J_s)}(\eta) := &  - \frac{1}{T-|J_s|}B_i^{-1} \sum_{t \notin J_s} \Big\{\Z_{it}\Big( \1(Y_{it} \leq \Z_{it}^\top \widehat{\bm{\gamma}}_{i}^{(-J_s)}(\eta)) -\1(Y_{it} \leq \Z_{it}^\top \bm{\gamma}_{i0}(\eta)) \Big) 
\\
& - \int z[F_{Y|Z}(z^\top \widehat{\bm{\gamma}}_{i}^{(-J_s)}(\eta) \mid z) - F_{Y|Z}(z^\top {\bm{\gamma}}_{i0}(\eta)\mid z)] dP^{\Z_{i1}}(z)\Big\}.
\end{align*}
By an application of Lemma C.2 in \cite{CVC} (note that the $g_n$ appearing in that Lemma is zero in our setting and set $t=2$ in that Lemma) we find that for constants $c_1 > 0,c_2 > 0$ independent of $i,T,n,J_t$
\begin{equation*}
\Big\{\sup_{\eta \in \mathcal{T}}\|\widehat{\bm{\gamma}}_{i}^{(-J_t)}(\eta) - \bm{\gamma}_{i0}(\eta)\| \leq c_1 s_{n,i}(t) \Big\} \supset \Big\{ s_{n,i}(t) \leq c_2 \Big\},
\end{equation*}
where
\begin{align*}
s_{n,i}(s) & := \sup_{\eta \in \mathcal{T}}\Big\| \frac{1}{T-|J_s|}\sum_{t \notin J_s} \Z_{it}(\1(Y_{it} \leq \Z_{it}^\top \bm{\gamma}_{i0}(\eta)) - \eta) \Big\|
\\
& \leq \sup_{\eta \in \mathcal{T}}\Big\| \frac{1}{T}\sum_{t =1}^T \Z_{it}(\1(Y_{it} \leq \Z_{it}^\top \bm{\gamma}_{i0}(\eta)) - \eta) \Big\| + \frac{C L \log T}{T},
\end{align*}
for a constant $C$ independent of $s,T,i$ and the inequality holds almost surely. Define the events 
\begin{equation*}
\Omega^{(1)}_{i,T}(\kappa) := \Big\{ \sup_{\eta \in \mathcal{T}}\Big\| \frac{1}{T}\sum_{t =1}^T \Z_{it}(\1(Y_{it} \leq \Z_{it}^\top \bm{\gamma}_{i0}(\eta)) - \eta) \Big\| \leq 2 c_0 \kappa^{1/2} \Big( \frac{\log T}{T}\Big)^{1/2} \Big\}.
\end{equation*}
An application of Lemma~\ref{lem:G1G2} shows that $\sup_{i} \P( \Omega^{(1)}_{i,T}(\kappa)) \geq 1 - T^{-\kappa}$. Now by properties of linear optimization (see for instance Lemma 34 on page 106 in \cite{BCCF}) we have under (A1)--(A3) 
\begin{equation*}
\sup_{s,\eta\in \mathcal{T}} \|r_{i,1}^{(-J_s)}(\eta)\| \leq \frac{\|B_i^{-1}\|M(p+1)}{T - |J_s|} \leq \frac{C_1}{T}.
\end{equation*}
Moreover, by Lemma C.1 of \cite{CVC} Assumptions (A1)--(A3) imply the existence of a constant $C_2$ independent of $i,s,T$ such that
\begin{equation*}
\sup_{\eta\in \mathcal{T}} \|r_{i,3}^{(-J_s)}(\eta)\| \leq C_2 \sup_{\eta\in \mathcal{T}} \|\widehat{\bm{\gamma}}_{i}^{(-J_s)}(\eta) - \bm{\gamma}_{i0}(\eta)\|^2.
\end{equation*}
Taken together, this shows that on $\Omega^{(1)}_{i,T}(\kappa)$ we have for $T$ large enough
\begin{equation}\label{eq:help1}
\sup_{\eta\in \mathcal{T}} \sup_t \Big\| r_{i,1}^{(-J_t)}(\eta) + r_{i,3}^{(-J_t)}(\eta) \Big\| \leq \kappa C_3 \frac{\log T}{T}, \quad \sup_{\eta\in \mathcal{T}} \Big\| r_{i,1}(\eta) + r_{i,3}(\eta) \Big\| \leq \kappa C_3 \frac{\log T}{T},
\end{equation}
where the second bound follows by similar arguments as the first one. Next consider the events
\begin{equation*}
\Omega^{(2)}_{i,T}(\kappa) := \Big\{ \sup_{0\leq \delta\leq 1}  \sup_{\|b_1-b_2\|\leq \delta}\sup_s \frac{\|\frac{1}{T-|J_s|}\sum_{t \notin J_s} W_{it}(b_1,b_2) - \E[W_{i1}(b_1,b_2)]\|}{\chi_T(\delta)} \leq 2 c_0 \kappa^2\Big\},
\end{equation*}
where
\begin{equation*}
W_{it}(b_1,b_2) := \Z_{it} (\1\{Y_{it}\leq b_1^\top \Z_{it}\} - \1\{Y_{it}\leq b_2^\top \Z_{it}\}).
\end{equation*}
Under (A1)--(A3) there exists a constant $C_5$ independent of $i,s,T,\delta$ such that for $T$ large enough 
\begin{align*}
& \sup_{\|b_1-b_2\|\leq \delta}\sup_s\Big\|\frac{1}{T-|J_s|}\sum_{t \notin J_s} W_{it}(b_1,b_2) - \E[W_{i1}(b_1,b_2)]\Big\|
\\
& \leq \sup_{\|b_1-b_2\|\leq \delta}\Big\|\frac{1}{T}\sum_{t=1}^T W_{it}(b_1,b_2) - \E[W_{i1}(b_1,b_2)]\Big\| + \frac{C_5 L \log T}{T},
\end{align*}
almost surely, and thus Lemma~\ref{lem:G1G2} implies that for $T$ large enough 
\begin{equation*}
\inf_{i} \P( \Omega^{(2)}_{i,T}(\kappa)) \geq \inf_{i} \P\Big(\sup_{0\leq \delta\leq 1}\frac{\|\mathbb{P}_{i,T} - \mathbb{P}_i\|_{\mathcal{G}_2(\delta)}}{\chi_T(\delta)} \leq c_0 \kappa^2 \Big) \geq 1-T^{-\kappa}.
\end{equation*}
Let $\delta_{i,T} := \sup_{\eta \in \mathcal{T}} \sup_t \|\widehat{\bm{\gamma}}_{i}^{(-J_t)}(\eta) - \widehat{\bm{\gamma}}_{i}(\eta)\|$. Now observe that by \eqref{eq:help1} and direct computations using (A1)--(A3)
\begin{align*}
&\sup_{\eta \in \mathcal{T}} \sup_s \|\widehat{\bm{\gamma}}_{i}^{(-J_s)}(\eta) - \widehat{\bm{\gamma}}_{i}(\eta)\|  = \sup_{\eta \in \mathcal{T}} \sup_s \|\widehat{\bm{\gamma}}_{i}^{(-J_s)}(\eta)-\bm{\gamma}_{i0}(\eta) -( \widehat{\bm{\gamma}}_{i}(\eta)-\bm{\gamma}_{i0}(\eta))\|
\\
& \leq \sup_{\eta \in \mathcal{T}} \sup_s \Big \| \frac{1}{T} B_i^{-1} \sum_{t=1}^{T} \Z_{it}(\1(Y_{it} \leq q_{i,\eta}(\Z_{it}))-\eta) - \frac{1}{T - |J_s|}B_i^{-1}\sum_{t\notin J_s}\Z_{it}(\1 (Y_{it} \leq q_{i,\eta}(\Z_{it}))-\eta)\Big\| 
\\
&\quad\quad  +\sup_{\eta \in \mathcal{T}} \sup_s  \|r_{i,1}^{(-J_s)}(\eta)+r_{i,3}^{(-J_s)}(\eta)-r_{i,1}(\eta)-r_{i,3}(\eta)\| +\sup_{\eta \in \mathcal{T}} \sup_s \|r_{i,2}^{(-J_s)}(\eta)-r_{i,2}(\eta)\|
\\
& \leq \Big(\frac{C_{6}L\log T}{T}\Big) +\sup_{\eta \in \mathcal{T}} \sup_s \|r_{i,2}^{(-J_s)}(\eta)-r_{i,2}(\eta)\|.
\end{align*}
Direct computations show that on $\Omega^{(1)}_{i,T}(\kappa) \cap \Omega^{(2)}_{i,T}(\kappa)$ (note that on $\Omega^{(1)}_{i,T}(\kappa)$ we have $\delta_{i,T} \leq 1$ for $T$ sufficiently large), 
\begin{align*}
 \sup_{t} \sup_{\eta \in \mathcal{T}} \|r_{i,2}^{(-J_t)}(\eta) - r_{i,2}(\eta)\| & \leq \frac{C_7L\log T}{T } +  \|\mathbb{P}_{i,T} - \mathbb{P}_{i}\|_{\mathcal{G}_2(\delta_{i,T})}
\\ 
& \leq \frac{C_7L\log T}{T } + c_0 \kappa^2(T^{-1/2}\delta_{i,T}^{1/2} \log T + T^{-1}(\log T)^2),
\end{align*}
where the second inequality follows by Lemma~\ref{lem:G1G2}. Combining the results obtained so far we find that on $\Omega^{(1)}_{i,T}(\kappa) \cap \Omega^{(2)}_{i,T}(\kappa)$ we have for $T$ sufficiently large
\begin{equation*}
\delta_{i,T} \leq C_8\kappa^2 \Big(T^{-1}(\log T)^2 + T^{-1/2}\delta_{i,T}^{1/2} \log T\Big).
\end{equation*}
A simple calculation shows that any non-negative $\delta_{i,T}$ satisfying this inequality automatically satisfies $\delta_{i,T} \leq 4 C_8^2 \kappa^4 T^{-1}(\log T)^2$. Applying the union bound over $i$ this shows the existence of a constant $C_9$ such that
\begin{equation}\label{eq:bminusJ}
\P\Big( \sup_{i} \sup_{t=1,...,T} \sup_{\eta \in \mathcal{T}} \|\widehat{\bm{\gamma}}_{i}^{(-J_t)}(\eta) - \bm{\gamma}_{i0}(\eta)\| \geq C_9 \kappa^4 \frac{(\log T)^2}{T}\Big) \leq 2 n T^{-\kappa}.
\end{equation}
To finalize the proof, consider the following decomposition
\begin{equation*}
r_{i,2}(\eta) = \frac{1}{T}B_i^{-1} \sum_{t=1}^{T} \Z_{it}(\1(Y_{it} \leq \Z_{it}^\top {\bm{\gamma}}_{i0}(\eta)) -\eta) + r_{i,2}^{(1)}(\eta) + r_{i,2}^{(2)}(\eta),
\end{equation*}
where 
\begin{align*}
r_{i,2}^{(1)}(\eta) &:= - \frac{1}{T}B_i^{-1} \sum_{t=1}^{T} \Big\{\Z_{it}(\1(Y_{it} \leq \Z_{it}^\top \widehat{\bm{\gamma}}_{i}^{(-J_t)}(\eta)) -\eta) - \int z[F_{Y|Z}(z^\top \widehat{\bm{\gamma}}_{i}^{(-J_t)}(\eta) \mid z) - \eta] dP^{\Z_{i1}}(z) \Big\},
\\
r_{i,2}^{(2)}(\eta) &:= \frac{1}{T}B_i^{-1} \sum_{t=1}^{T} \Big\{\Z_{it}\Big( \1(Y_{it} \leq \Z_{it}^\top \widehat{\bm{\gamma}}_{i}^{(-J_t)}(\eta)) -\1(Y_{it} \leq \Z_{it}^\top \widehat{\bm{\gamma}}_{i}(\eta)) \Big) 
\\
& \quad \quad \quad \quad\quad \quad\quad \quad- \int z[F_{Y|Z}(z^\top \widehat{\bm{\gamma}}_{i}^{(-J_t)}(\eta) \mid z) - F_{Y|Z}(z^\top \widehat{\bm{\gamma}}_{i}(\eta)\mid z)] dP^{\Z_{i1}}(z)\Big\}.
\end{align*} 
Clearly $\E[r_{i,2}(\eta)] = \E[r_{i,2}^{(1)}(\eta)+r_{i,2}^{(2)}(\eta)]$. Now given~\eqref{eq:bminusJ}, exactly the same arguments as in the proof of Theorem S.6.1 of \cite{VCC} (see page 43 of the supplementary material of \cite{VCC}) yield (note that the $m, \xi_m$ of \cite{VCC} are fixed constants in our setting and use~\eqref{eq:bminusJ} instead of Lemma S.6.3 in \cite{VCC}) 
\begin{equation*}
\sup_i \|\E[r_{i,2}^{(2)}(\eta)]\| = O\Big(\frac{(\log T)^2}{T}\Big).
\end{equation*}
To bound $\E[r_{i,2}^{(1)}(\eta)]$, note that for any $t = 1,...,T$ the quantity $\widehat{\bm{\gamma}}_i^{(-J_t)}(\eta))$ does not contain elements of $\{(Y_{it},\Z_{it}): t\in J_t\}$. Hence, by a property of $\beta$-mixing (see for instance Lemma 2.6 in \cite{Yu1994}) we find that 
\begin{equation*}
\sup_{i,t} \Big\|\E\Big[ \Z_{it}(\1(Y_{it} \leq \Z_{it}^\top \widehat{\bm{\gamma}}_i^{(-J_t)}(\eta)) -\eta) - \int z[F_{Y|Z}(z^\top \widehat{\bm{\gamma}}_{i}^{(-J_t)}(\eta) \mid z) - \eta] dP^{\Z_{i1}}(z) \Big]\Big\| \leq M\beta(L\log T).
\end{equation*}
Under assumption (D1) we can choose the constant $L$ large enough to ensure that $M\beta(L\log T) \leq (\log T)^2/T$, and thus the proof of~\eqref{vcc3_beta} is complete. 

\end{proof}


\subsection{Proofs of Theorem~\ref{th1FE} and Theorem~\ref{th2FE}}

Since the proofs of both results are very similar, we will present both proofs together. Before proceeding, we note the following fact that will be used on several occasions:
\[
\mbox{\textbf{fact 1:}} \quad \quad 0 < \delta < a + b\delta^{1/2} \Rightarrow \delta \leq 4\max(a,b^2).
\]
More details will be given in the dependent case and we will point out necessary adjustments in the i.i.d. case along the way. Recall the definition of $g_i, \Gamma_n$ in~\eqref{eq:def-gi},~\eqref{eq:def-gamman} and define
\begin{align*}
\H_{ni}^{(1)}(\alpha,\zb) &:= \frac{1}{T}\sum_{t=1}^T (\tau - \1\{Y_{it} \leq \X_{it}^\top \zb+ \alpha\}), \quad H_{ni}^{(1)}(\alpha,\zb) := \E[\H_{ni}^{(1)}(\alpha,\zb)]
\\
\H_{ni}^{(2)}(\alpha,\zb) &:= \frac{1}{T}\sum_{t=1}^T \X_{it}(\tau - \1\{Y_{it} \leq \X_{it}^\top \zb+ \alpha\}), \quad H_{ni}^{(2)}(\alpha,\zb) := \E[\H_{ni}^{(2)}(\alpha,\zb)]
\\
\H_{ni}^{(3)}(\alpha,\zb) &:= -g_i\H_{ni}^{(1)}(\alpha,\zb) + \H_{ni}^{(2)}(\alpha,\zb), \quad H_{ni}^{(3)}(\alpha,\zb) := \E[\H_{ni}^{(3)}(\alpha,\zb)].
\end{align*}

Throughout this section we will drop the index $FE$ from the notation $\widehat\zb_{FE}$ since there is no MD-QR estimator in this section. Define
\begin{equation}\label{eq:alphatilde}
\widetilde\alpha_i := \argmin_{a \in R} \sum_{t=1}^T \rho_\tau(Y_{it} - \X_{it}^\top \zb_0 - a).
\end{equation}
It can further be shown by standard arguments that there exists a constant $c$ such that both the dependent and i.i.d. setting we have for all $\kappa>1$
\begin{equation}\label{eq:alphatilde-tail}
\sup_i P\Big( |\widetilde\alpha_i-\alpha_{i0}| > c\kappa^{1/2}T^{-1/2}\log T\Big) = O(T^{-\kappa}).
\end{equation}

Next we state a useful technical Lemma which will be used in the following proof, a proof of this Lemma is provided at the end of this section.

\begin{lemma} \label{lem:checkalpha}
Let assumptions (A0)-(A3) hold. We have
\[
\max_{i=1,...,n} |\widetilde\alpha_i - \widehat\alpha_i| = O_P\Big(r_{n,T}(\|\widehat\zb-\zb_0\|)\Big)
\]
where 
\[
r_{n,T}(\delta) = \delta + T^{-1}\log T
\] 
under (I) and 
\[
r_{n,T}(\delta) = \delta + T^{-1}(\log T)^2 
\] 
under (D1) and (D2).
\end{lemma}

We are ready to prove Theorem~\ref{th1FE} and Theorem~\ref{th2FE}. Begin by observing that the results in~\cite{KatoGalvaoMontes-Rojas12} imply that under the assumptions of both theorems we have
\begin{equation}\label{eq:ratealpha}
\sup_{i} |\widehat\alpha_i - \alpha_{i0}| = O_P(T^{-1/2}(\log T)^{1/2})
\end{equation}
and $\|\widehat\zb - \zb_0\| = O_P(T^{-1/2})$.

Following the arguments that lead to the expansion (A.5) in~\cite{KatoGalvaoMontes-Rojas12}, we obtain the following representation (recall the definition of $\Gamma_n$ in~\eqref{eq:def-gamman}; note that the following representation holds in the i.i.d. as well as in the dependent case since all $O_P, o_P$ terms result from Taylor expansions and not from empirical process arguments, note that the following representation makes use of~\eqref{eq:ratealpha})
\begin{align*}
&\widehat\zb - \zb_0 + o_P(\|\widehat\zb - \zb_0\|)  
\\
=~& \Gamma_n^{-1}\frac{1}{n}\sum_{i=1}^n \H_{ni}^{(3)}(\alpha_{i0},\zb_0)
+ \Gamma_n^{-1} \frac{1}{n} \sum_{i=1}^n \Big[\H_{ni}^{(3)}(\widehat\alpha_i,\widehat\zb) - \H_{ni}^{(3)}(\alpha_{i0},\zb_0) - \Big\{H_{ni}^{(3)}(\widehat\alpha_i,\widehat\zb) - H_{ni}^{(3)}(\alpha_{i0},\zb_0)\Big\}\Big]
\\
&~+ O_P(T^{-1}\log T).
\end{align*}
Next we further expand the middle term. In what follows we will discuss the proof under (D1) in detail and point out important differences in the i.i.d. case where necessary. 


Observe the decomposition
\begin{align*}
&\frac{1}{n}\sum_{i=1}^n \H_{ni}^{(3)}(\widehat\alpha_i,\widehat\zb) - \H_{ni}^{(3)}(\alpha_{i0},\zb_0) - \Big\{H_{ni}^{(3)}(\widehat\alpha_i,\widehat\zb) - H_{ni}^{(3)}(\alpha_{i0},\zb_0)\Big\}
\\
= & \frac{1}{n}\sum_{i=1}^n \H_{ni}^{(3)}(\widehat\alpha_i,\widehat\zb) - \H_{ni}^{(3)}(\widetilde\alpha_i,\zb_0) - \Big\{H_{ni}^{(3)}(\widehat\alpha_i,\widehat\zb) - H_{ni}^{(3)}(\widetilde\alpha_i,\zb_0)\Big\}
\\
& + \frac{1}{n}\sum_{i=1}^n \H_{ni}^{(3)}(\widetilde\alpha_i,\zb_0) - \H_{ni}^{(3)}(\alpha_{i0},\zb_0) - \Big\{H_{ni}^{(3)}(\widetilde\alpha_i,\zb_0) - H_{ni}^{(3)}(\alpha_{i0},\zb_0)\Big\}.
\end{align*}

By similar arguments as is the proof of~\eqref{vcc3_beta} (see Lemma~\ref{VCClemma_beta}) under (D1) we obtain 
\[
\Big\|\E\Big[\frac{1}{n}\sum_{i=1}^n \H_{ni}^{(3)}(\widetilde\alpha_i,\zb_0) - \H_{ni}^{(3)}(\alpha_{i0},\zb_0) - \Big\{H_{ni}^{(3)}(\widetilde\alpha_i,\zb_0) - H_{ni}^{(3)}(\alpha_{i0},\zb_0)\Big\} \Big]\Big\| = O((\log T)^2/T)
\]
while in the i.i.d. case we have by similar arguments as in the proof of Theorem S.6.1 of~\cite{VCC} (see also the discussion in the beginning of the proof of Lemma~\ref{VCClemma})
\[
\Big\|\E\Big[\frac{1}{n}\sum_{i=1}^n \H_{ni}^{(3)}(\widetilde\alpha_i,\zb_0) - \H_{ni}^{(3)}(\alpha_{i0},\zb_0) - \Big\{H_{ni}^{(3)}(\widetilde\alpha_i,\zb_0) - H_{ni}^{(3)}(\alpha_{i0},\zb_0)\Big\} \Big]\Big\| = O((\log T)/T).
\]
Moreover we note that the vectors 
\[
\xi_{i,n,T} := \H_{ni}^{(3)}(\widetilde\alpha_i,\zb_0) - \H_{ni}^{(3)}(\alpha_{i0},\zb_0) - \Big\{H_{ni}^{(3)}(\widetilde\alpha_i,\zb_0) - H_{ni}^{(3)}(\alpha_{i0},\zb_0)\Big\}
\] 
are independent across $i$, and that 
\begin{align*}
\|\xi_{i,n,T}\| \leq (\|g_i\|+1) \|\mathbb{P}_{i,T} - \mathbb{P}_i\|_{\mathcal{G}_2(|\widetilde\alpha_i-\alpha_{i0}|)}
\end{align*}
where $\mathcal{G}_2(\delta)$ is defined in Lemma~\ref{lem:G1G2}. Recalling~\eqref{eq:alphatilde-tail} and applying~\eqref{G2_beta} we find that for the constant $c$ from~\eqref{eq:alphatilde-tail}
\begin{align*}
&P\Big( \sup_i\|\xi_{i,n,T}\| > T^{-2/3} \Big) 
\\
\leq &P\Big( \sup_i |\widetilde\alpha_i-\alpha_{i0}| > 2cT^{-1/2}\log T \Big) + P\Big( (\|g_i\|+1) \|\mathbb{P}_{i,T} - \mathbb{P}_i\|_{\mathcal{G}_2(2cT^{-1/2}\log T)} > T^{-2/3} \Big)
\\
= & O(T^{-2}) 
\end{align*}
where we used~\eqref{eq:alphatilde-tail} to bound the first piece in the sum and~\eqref{G2_beta} in combination with $\chi_T(2cT^{-1/2}\log T) = o(T^{-2/3})$ for the second piece.
%
Since also almost surely
\[
\sup_i \Big\| \H_{ni}^{(3)}(\widetilde\alpha_i,\zb_0) - \H_{ni}^{(3)}(\alpha_{i0},\zb_0) - \Big\{H_{ni}^{(3)}(\widetilde\alpha_i,\zb_0) - H_{ni}^{(3)}(\alpha_{i0},\zb_0)\Big\}\Big\| \leq K 
\]
for a $K < \infty$ this yields, denoting by $\xi_{i,n,T}^{(j)}$ the j'th entry of the vector $\xi_{i,n,T}$
\begin{align*}
\sup_{i,j} Var(\xi_{i,n,T}^{(j)}) &\leq \sup_i \E[(\xi_{i,n,T}^{(j)})^2] = \sup_i \E[(\xi_{i,n,T}^{(j)})^2(\1\{|\xi_{i,n,T}^{(j)}|> T^{-2/3}\} + \1\{|\xi_{i,n,T}^{(j)}|\leq T^{-2/3}\} )]
\\
&\leq K^2 \sup_i P(|\xi_{i,n,T}^{(j)}|> T^{-2/3}) + T^{-4/3} = O(T^{-4/3}).
\end{align*}
Hence by independence across $i$
\[
\frac{1}{n}\sum_{i=1}^n \xi_{i,n,T} 
= O_P\Big(\sup_{i,j} \Big(|\E[\xi_{i,n,T}^{(j)}]| + \sqrt{Var(\xi_{i,n,T}^{(j)})}/\sqrt{n}\Big)\Big) = O_P(n^{-1/2}T^{-2/3} + (\log T)^2/T)
\]
under (D1), (C2). In the i.i.d. case similar arguments show a similar bound with $O_P(n^{-1/2}T^{-2/3} +T^{-1}\log T)$. Combining all results so far we have
\begin{align*}
\widehat\zb-\zb_0 =& \Gamma_n^{-1}\frac{1}{n}\sum_{i=1}^n \H_{ni}^{(3)}(\alpha_{i0},\zb_0)
\\
&+ \Gamma_n^{-1}\Big[\frac{1}{n}\sum_{i=1}^n \H_{ni}^{(3)}(\widehat\alpha_i,\widehat\zb) - \H_{ni}^{(3)}(\widetilde\alpha_i,\zb_0) - \Big\{H_{ni}^{(3)}(\widehat\alpha_i,\widehat\zb) - H_{ni}^{(3)}(\widetilde\alpha_i,\zb_0)\Big\}\Big]
\\
&+ O_P(T^{-1}(\log T)^2) + o_P(\|\widehat\zb-\zb_0\|) + o_P((nT)^{-1/2})
\end{align*}
under (D1) and the remainder is $O_P(T^{-1}\log T) + o_P(\|\widehat\zb-\zb_0\|) + o_P((nT)^{-1/2})$ in the i.i.d. case.

Next, note that for $\mathcal{G}_2(\delta)$ from Lemma~\ref{lem:G1G2}
\begin{align*}
&\Big| \H_{ni}^{(3)}(\widehat\alpha_i,\widehat\zb) - \H_{ni}^{(3)}(\widetilde\alpha_i,\zb_0) - \Big\{H_{ni}^{(3)}(\widehat\alpha_i,\widehat\zb) - H_{ni}^{(3)}(\widetilde\alpha_i,\zb_0)\Big\}\Big|
\\
\leq& (\|g_i\|+1) \|\mathbb{P}_{i,T} - \mathbb{P}_i\|_{\mathcal{G}_2(\|\widehat \zb-\zb_0\|+\sup_i |\widetilde\alpha_i-\widehat\alpha_i|)}
\end{align*}
Combining this with~\eqref{G2_beta} under (D1) 
it follows that
\begin{multline*}
\sup_i \Big|  \H_{ni}^{(3)}(\widehat\alpha_i,\widehat\zb) - \H_{ni}^{(3)}(\widetilde\alpha_i,\zb_0) - \Big\{H_{ni}^{(3)}(\widehat\alpha_i,\widehat\zb) - H_{ni}^{(3)}(\widetilde\alpha_i,\zb_0)\Big\}\Big|
\\ 
= O_P\Big(\{\|\widehat\zb-\zb_0\|+\sup_i |\widetilde\alpha_i-\widehat\alpha_i|\}^{1/2}T^{-1/2}\log T + T^{-1}(\log T)^2\Big),
\end{multline*}
while an application of Lemma~\ref{lem:checkalpha} further shows 
\begin{multline*}
O_P\Big(\{\|\widehat\zb-\zb_0\|+\sup_i |\widetilde\alpha_i-\widehat\alpha_i|\}^{1/2}T^{-1/2}\log T + T^{-1}(\log T)^2\Big)
\\ 
= O_P\Big(\{\|\widehat\zb-\zb_0\|+T^{-1}(\log T)^2\}^{1/2}T^{-1/2}\log T + T^{-1}(\log T)^2\Big)
\\
= O_P\Big(\|\widehat\zb-\zb_0\|^{1/2}T^{-1/2}\log T + T^{-1}(\log T)^2\Big).
\end{multline*}
Finally note that $\sum_{i=1}^n H_{ni}^{(3)}(\alpha_{i0},\zb_0) = 0$, and by the results in \cite{KatoGalvaoMontes-Rojas12} we have 
\[
\frac{1}{n}\sum_{i=1}^n \H_{ni}^{(3)}(\alpha_{i0},\zb_0) = \frac{1}{n}\sum_{i=1}^n \Big\{\H_{ni}^{(3)}(\alpha_{i0},\zb_0) - H_{ni}^{(3)}(\alpha_{i0},\zb_0) \Big\} = O_P((nT)^{-1/2}).
\]
Combining all bounds we obtained so far thus yields
\begin{align}
\widehat\zb-\zb_0 =& \Gamma_n^{-1}\frac{1}{n}\sum_{i=1}^n \H_{ni}^{(3)}(\alpha_{i0},\zb_0) + O_P\Big(T^{-1}(\log T)^2\Big) + \|\widehat\zb-\zb_0\|^{1/2}O_P(T^{-1/2}\log T) \nonumber
\\
& \quad \quad + o_P(\|\widehat\zb-\zb_0\|+(nT)^{-1/2}) \label{eq:expbeta1}
\\
=& O_P\Big((nT)^{-1/2} + T^{-1}(\log T)^2\Big) + \|\widehat\zb-\zb_0\|^{1/2}O_P(T^{-1/2}\log T) + o_P(\|\widehat\zb-\zb_0\|). \nonumber
\end{align}
Recalling \textbf{fact 1} this implies
\[
\|\widehat\zb-\zb_0 \| =  O_P\Big((nT)^{-1/2} + T^{-1}(\log T)^2\Big).
\]
Plugging this into~\eqref{eq:expbeta1} yields
\begin{align*}
\widehat\zb-\zb_0 =& \Gamma_n^{-1}\frac{1}{n}\sum_{i=1}^n \H_{ni}^{(3)}(\alpha_{i0},\zb_0) + O_P\Big(\{T^{-1/2}\log T + (nT)^{-1/4}\}T^{-1/2}\log T\Big).
\end{align*}
Similar arguments in the i.i.d. case show that
\begin{align*}
\widehat\zb-\zb_0 =& \Gamma_n^{-1}\frac{1}{n}\sum_{i=1}^n \H_{ni}^{(3)}(\alpha_{i0},\zb_0)  + O_P\Big(\{T^{-1/2}(\log T)^{1/2} + (nT)^{-1/4}\}T^{-1/2}(\log T)^{1/2}\Big).
\end{align*}
The proof is complete. \hfill $\Box$

\bigskip


\noindent \textbf{Proof of Lemma \ref{lem:checkalpha}}
Define 
\[
\widehat \F_{n,T}^{(i)}(y;\zb) := \frac{1}{T}\sum_{t=1}^T \1\{Y_{it} - \X_{it}^\top \zb \leq y\}, \quad F_{n,T}^{(i)}(y;\zb) := \E[\widehat \F_{n,T}^{(i)}(y;\zb)].
\]
Under (D1) we have by Lemma~\ref{lem:G1G2} (in the definition of the function class $\mathcal{G}_2(\delta)$, set $a = (1,0,...,0)^\top \in \mathcal{S}^{p}$, $b_1 = (y,\zb^\top)^{\top}, b_2 = (y,\zb^\top_0)^{\top}$)
\begin{align*}
\sup_{i=1,...,n}\Big|\widehat \F_{n,T}^{(i)}(y;\zb) - \widehat \F_{n,T}^{(i)}(y;\zb_0)\Big| &\leq \sup_{i=1,...,n} \|\mathbb{P}_{i,T} - \mathbb{P}_i\|_{\mathcal{G}_2(\|\zb-\zb_0\|)} + \sup_{i,y}\Big| F_{n,T}^{(i)}(y;\zb) - F_{n,T}^{(i)}(y;\zb_0)\Big|
\\
&\leq \sup_{i=1,...,n} \|\mathbb{P}_{i,T} - \mathbb{P}_i\|_{\mathcal{G}_2(\|\zb-\zb_0\|)} + Mf_{max}\|\zb-\zb_0\|.
\end{align*}
Similarly we have (in the definition of the function class $\mathcal{G}_2(\delta)$, set $a = (1,0,...,0)^\top \in \mathcal{S}^{p}$, $b_1 = (y,\zb^\top_0)^{\top}, b_2 = (y',\zb^\top_0)^{\top}$) 
\[
\sup_i \Big|\widehat \F_{n,T}^{(i)}(y;\zb_0) - \widehat \F_{n,T}^{(i)}(y';\zb_0) - \{F_{n,T}^{(i)}(y;\zb_0) - F_{n,T}^{(i)}(y';\zb_0)\}\Big| \leq \sup_i \|\mathbb{P}_{i,T} - \mathbb{P}_i\|_{\mathcal{G}_2(|y-y'|)}.
\]
Finally, note that $\widetilde\alpha_i$ is the empirical $\tau-$quantile of $\widehat \F_{n,T}^{(i)}(\cdot;\zb_0)$, the empirical cdf of $\{Y_{it}-X_{it}^\top\zb_0: t=1,..,T\}$, while $\widehat \alpha_i$ is the empirical $\tau-$quantile of $\widehat \F_{n,T}^{(i)}(\cdot;\widehat\zb)$, the empirical cdf of $\{Y_{it}-X_{it}^\top\widehat\zb: t=1,...,T\}$. Since the data are in general position (in the dependent case this follows from the existence of conditional densities imposed in (D2)) it further follows that both samples $\{Y_{it}-X_{it}^\top\zb_0: t=1,..,T\}, \{Y_{it}-X_{it}^\top\widehat\zb: t=1,...,T\}$ have no two observations with exactly the same value, so that
\begin{align*}
2/T 
&\geq \Big|\widehat \F_{n,T}^{(i)}(\widehat\alpha_i;\widehat\zb) - \widehat \F_{n,T}^{(i)}(\widetilde\alpha_i;\zb_0) \Big|
\\
&\geq \Big|\widehat \F_{n,T}^{(i)}(\widehat\alpha_i;\zb_0) - \widehat \F_{n,T}^{(i)}(\widetilde\alpha_i;\zb_0) \Big| - \|\mathbb{P}_{i,T} - \mathbb{P}_i\|_{\mathcal{G}_2(\|\zb-\zb_0\|)} - Mf_{max}\|\zb-\zb_0\|
\\
&\geq \Big|F_{n,T}^{(i)}(\widehat\alpha_i;\zb_0) - F_{n,T}^{(i)}(\widetilde\alpha_i;\zb_0) \Big| - \|\mathbb{P}_{i,T} - \mathbb{P}_i\|_{\mathcal{G}_2(\|\zb-\zb_0\|)} - Mf_{max}\|\zb-\zb_0\|
\\
& \quad - \|\mathbb{P}_{i,T} - \mathbb{P}_i\|_{\mathcal{G}_2(\sup_i|\widetilde\alpha_i-\widehat\alpha_i|)} 
\end{align*}
Next, observe that under (A2) and (A3) there exists $\varepsilon>0$ such that for all $|y-\alpha_{i0}|\vee |y'-\alpha_{i0}| \leq \varepsilon$ we have 
\[
\Big|F_{n,T}^{(i)}(y;\zb_0)-F_{n,T}^{(i)}(y';\zb_0)\Big| \geq f_{min}|y-y'|/2.
\]
Note that by~\eqref{eq:alphatilde-tail} and~\eqref{eq:ratealpha} 
\begin{equation}\label{eq:helpprob}
P(\Omega_{n,T}) := P\Big(\sup_i |\widehat\alpha_i-\alpha_{i0}|\vee |\widetilde\alpha_i-\alpha_{i0}| \leq \varepsilon \Big) \to 1.
\end{equation}
We now first discuss the dependent case. Recall the definition of $\chi_T$ in~\eqref{def:chi} and let 
\[
S_{n,1} := \sup_{0\leq \delta\leq 1}\frac{\|\mathbb{P}_{i,T} - \mathbb{P}_i\|_{\mathcal{G}_2(\delta)}}{\chi_T(\delta)}.
\]
Combining all bounds so far we obtain that on the event $\Omega_{n,T}$ from~\eqref{eq:helpprob}
\begin{align*}
\sup_i |\widehat \alpha_i - \widetilde\alpha_i| 
&\leq \frac{2}{f_{min}} \sup_i \Big|F_{n,T}^{(i)}(\widehat\alpha_i;\zb_0) - F_{n,T}^{(i)}(\widetilde\alpha_i;\zb_0) \Big|
\\
&\leq \frac{2}{f_{min}}\Big( 2/T + Mf_{max}\|\widehat\zb-\zb_0\| + S_{n,1}\{\chi_T(\|\widehat\zb-\zb_0\|) + \chi_T(\sup_i|\widehat\alpha_i-\widetilde\alpha_i|)\}\Big)
\\
&\lesssim S_{n,1}T^{-1/2}(\log T)\{\sup_i|\widehat\alpha_i-\widetilde\alpha_i|\}^{1/2} + (S_{n,1}+1)(\log T)^2/T + (S_{n,1}+1)\|\widehat\zb-\zb_0\|. 
\end{align*}
Applying \textbf{fact 1} with $\delta = \sup_i |\widehat \alpha_i - \widetilde\alpha_i|$ this shows that on $\Omega_{n,T}$
\[
\sup_i |\widehat \alpha_i - \widetilde\alpha_i| \lesssim (S_{n,1}+1)^2\{T^{-1}(\log T)^2+\|\widehat\zb-\zb_0\|\}. 
\]
Since $S_{n,1} = O_P(1)$ by~\eqref{G2_beta} this completes the proof in the dependent case. The proof in the i.i.d. case is exactly the same after replacing $\chi_T$ in $S_{n,1}$ with 
\[
\widetilde\chi_T(\delta) := T^{-1/2}\delta^{1/2}(\log T)^{1/2} + T^{-1}\log T
\] 
and using Lemma S.6.4 in~\cite{VCC} instead of~\eqref{G2_beta}. \hfill $\Box$

\subsection{Proof of Theorem~\ref{th1} and Lemma~\ref{lemma:CovHK}}

We will make use of the following notation: 
\begin{align*}
\psi_{i, \tau} (\Z, Y) &= \Z (\1(Y \leq q_{i,\tau}(\Z))-\tau),
\\
f_{it} &:= \frac{2d_T}{q_{i,\tau+d_T}(\Z_{it})-q_{i,\tau-d_T}(\Z_{it})}, \quad e_{it} := 1/f_{it},
\\
B_{iT} &= \frac{1}{T}\sum_{t=1}^{T} f_{it} \Z_{it} \Z_{it}^\top,
\\
V_{iT} &= \E[B_{iT}]^{-1} A_i\E[B_{iT}]^{-1},
\\
a_T &= \frac{(\log T)^{1/2}}{(Td_T)^{1/2}}
\end{align*}
where $q_{i,\eta}(\mathbf{z})$ is the conditional $\eta$ quantile of $Y_{i1}$ given $\Z_{i1} = \mathbf{z}$. Moreover, denote by $W_i$ is the lower $p\times p$ matrix of $V_i$, by $W_{iT}$ is the lower $p\times p$ matrix of $V_{iT}$, and by $\widehat{W}_{iT}$ the  lower $p\times p$ matrix of $\widehat{V}_{iT}$. We begin by stating and proving several technical lemmas that are useful in the proof of Theorem \ref{th1}.

\begin{lemma}
\label{lemma1}
Under Assumptions (I), (A0)--(A2) we have  $\E[B_{iT}] - B_i = o(1)$ and $\|W_{iT}-W_i\| = o(1)$ uniformly in $i$. 
\end{lemma}

\bigskip


\bigskip

\begin{lemma}
\label{lemma2}
Let Assumptions (I), (A0)--(A3), (MI) hold. Let $\widehat{e}_{it} := \widehat{f}_{it}^{-1}$, then $\sup_{i,t}\left |\widehat{e}_{it} - e_{it}\right | = O_p(a_T)$ and 
\begin{equation*}
\sup_{i,t}\left |\widehat{f}_{it} - f_{it} -\frac{e_{it} - \widehat{e}_{it}}{e_{it}^2}\right| = O_p(a_T^2).
\end{equation*}
Moreover,
\begin{equation}\label{psiit}
\sup_i \Big\| \frac{1}{T}\sum_{t=1}^{T} \frac{\psi_{i,\tau+d_T}(\Z_{it}, Y_{it}) - \psi_{i,\tau - d_T}(\Z_{it}, Y_{it})}{2d_T} \Big\| = O_p(a_T). 
\end{equation}
\end{lemma}

\bigskip

\begin{lemma}\label{lem:aitbit}
Under Assumptions (I), (A0)--(A3), (MI)
\begin{equation*}
\widehat{B}_{iT} - B_{iT} = \frac{1}{T}\sum_{t=1}^{T} \sum_{k=1}^{p+1} \frac{\psi_{i,\tau+d_T,k}(\Z_{it}, Y_{it}) - \psi_{i,\tau - d_T,k}(\Z_{it}, Y_{it})}{2d_T} \E\Big[\frac{(B_i^{-1}\Z_{i1})_k \Z_{i1}\Z_{i1}^\top}{e_{i1}^2} \Big] + R_{2i},
\end{equation*}
with $\psi_{i,\tau+d_T,k}(\Z_{it}, Y_{it})$ denoting the k-th element of the vector $\psi_{i,\tau+d_T}(\Z_{it}, Y_{it})$ and
\begin{equation} \label{Bit2}
\sup_{i} \Big\|\frac{1}{T}\sum_{t=1}^{T} \sum_{k=1}^{p+1} \frac{\psi_{i,\tau+d_T,k}(\Z_{it}, Y_{it}) - \psi_{i,\tau - d_T,k}(\Z_{it}, Y_{it})}{2d_T} \E\Big[\frac{(B_i^{-1}\Z_{i1})_k \Z_{i1}\Z_{i1}^\top}{e_{i1}^2} \Big]\Big\| = O_p(a_T),
\end{equation} 
and
\begin{equation} \label{Bit3}
\sup_{i} \|R_{2i}\| = O_p\Big(a_T^2 + a_T \Big(\frac{\log T}{T}\Big)^{1/2} + \frac{(\log T)^{3/4}}{d_T T^{3/4}}\Big).
\end{equation}
Moreover,
\begin{equation} \label{Bit1}
\sup_i \|B_{iT}^{-1}-\E[B_{iT}]^{-1}\| = O_p\Big(\Big(\frac{\log T}{T}\Big)^{1/2}\Big),
\end{equation}
and 
\begin{equation} \label{Ait1}
\sup_i \|\widehat{A}_{iT}-A_i\| = O_p\Big(\Big(\frac{\log T}{T}\Big)^{1/2}\Big).
\end{equation}
\end{lemma}

%

\bigskip

\noindent \textbf{Proof of Theorem \ref{th1}.}
We are ready to prove the main result. Observe that $\sup_i \|A_i\| = O(1), \sup_{i} \|\E[B_{iT}]^{-1}\| = O(1)$ and thus 
\begin{align}
\widehat{V}_{iT} - V_{iT}
& = (\widehat{B}_{iT}^{-1}- \E[B_{iT}]^{-1}) A_i \E[B_{iT}]^{-1} + \E[B_{iT}]^{-1} (\widehat{A}_{iT}-A_i)\E[B_{iT}]^{-1}  \nonumber
\\
& \quad + \E[B_{iT}]^{-1}A_i (\widehat{B}_{iT}^{-1} - \E[B_{iT}]^{-1}) + O\Big(\|\widehat{B}_{iT}^{-1} - \E[B_{iT}]^{-1}\|^2+\|\widehat{A}_{iT}-A_{i}\|^2\Big) \nonumber
\\
& = \E[B_{iT}]^{-1} A_i \E[B_{iT}]^{-1}(\E[B_{iT}]-\widehat{B}_{iT})\E[B_{iT}]^{-1} + \E[B_{iT}]^{-1}( \widehat{A}_{iT} - A_i) \E[B_{iT}]^{-1}
\nonumber
\\
& \quad + \E[B_{iT}]^{-1} (\E[B_{iT}]-\widehat{B}_{iT})\E[B_{iT}]^{-1}A_i \E[B_{iT}]^{-1} + R_{3i} \nonumber\\
& =:  \frac{1}{T} \sum_{t} \eta_{iT}(\Z_{it}, Y_{it}) + R_{3i}, \label{Vit1}
\end{align}
where by Lemma \ref{lem:aitbit} we obtain after some simplifications
\begin{equation*}
\sup_i \|R_{3i}\| = O_p(a_T^2)
\end{equation*}
and  
\begin{align*}
& \eta_{iT}(\Z_{it}, Y_{it})  
\\
:=& - \E[B_{iT}]^{-1}A_i\E[B_{iT}]^{-1} \Big\{\sum_{k=1}^{p+1}  \frac{\psi_{i,\tau+d_T,k}(\Z_{it}, Y_{it}) - \psi_{i,\tau - d_T,k}(\Z_{it}, Y_{it})}{2d_T} \E\Big[\frac{(B_i^{-1}\Z_{i1})_k \Z_{i1}\Z_{i1}^\top}{e_{i1}^2} \Big] \Big\} \E[B_{iT}]^{-1}
\\
&- \E[B_{iT}]^{-1} \Big\{\sum_{k=1}^{p+1}  \frac{\psi_{i,\tau+d_T,k}(\Z_{it}, Y_{it}) - \psi_{i,\tau - d_T,k}(\Z_{it}, Y_{it})}{2d_T} \E\Big[\frac{(B_i^{-1}\Z_{i1})_k \Z_{i1}\Z_{i1}^\top}{e_{i1}^2} \Big] \Big\} \E[B_{iT}]^{-1}A_i\E[B_{iT}]^{-1}
\\
& + \E[B_{iT}]^{-1} \frac{1}{T}\sum_{t=1}^T \tau(1-\tau) (\Z_{it}\Z_{it}^\top - \E[\Z_{it}\Z_{it}^\top]) \E[B_{iT}]^{-1}.
\end{align*}
Note that by definition, $\eta_{iT}(\Z_{it}, Y_{it})$ are independent across $i$ with $\E[\eta_{iT}(\Z_{it}, Y_{it})] = 0$,  
{Since $\widehat{W}_{iT}$ is the lower $p \times p$ sub-matrix of $\widehat{V}_{iT}$ and since $\sup_{i=1,,...,n}\|W_{iT}^{-1}\| = O_P(1)$ under the assumptions made, we have after some simple algebra
\begin{align}
	\widehat{W}_{iT}^{-1} - W_{iT}^{-1} &= - W_{iT}^{-1}(\widehat{W}_{iT} - W_{iT}) W_{iT}^{-1} + O\Big( \|\widehat{W}_{iT} - W_{iT}\|^2\Big) \notag \\
	& := \frac{1}{T} \sum_{t=1}^T \widetilde{\eta}_{iT}(\Z_{it}, Y_{it}) + \widetilde{R}_{3i} \label{Vit1-new}
\end{align}
where 
\begin{equation} \label{eq:iid-tildab}
\sup_i \|\widetilde{R}_{3i}\| = O_p(a_T^2), \quad \sup_i \Big\|\frac{1}{T} \sum_{t=1}^T \widetilde{\eta}_{iT}(\Z_{it}, Y_{it})\Big\| = O_p(a_T)
\end{equation}
by Lemma~\ref{lemma2} and Lemma~\ref{lem:aitbit}. Moreover, the $\widetilde{\eta}_{iT}(\Z_{it}, Y_{it})$ are $p \times p$ matrices with entries $\theta_{k,\ell}(i,t)$ satisfying for some constant $C$
\begin{equation} \label{eq:thetakl-iid}
\E[\theta_{k,\ell}(i,t)] = 0, \quad |\theta_{k,\ell}(i,t)| \leq C/d_T, \quad \E|\theta_{k,\ell}(i,t)| \leq C.
\end{equation}
Let $\phi_{i,\tau}(\Z_{it}, Y_{it})$ denote the last $p$ entries of $B_i^{-1} \psi_{i,\tau}(\Z_{it}, Y_{it})$. Then by Lemma \ref{VCClemma}
\begin{align*}
\widehat{\bm{\beta}}_{i}(\tau) - \bm{\beta}_{0}(\tau) & = \frac{1}{T}\sum_{t=1}^{T} \phi_{i,\tau} (\Z_{it}, Y_{it}) + \widetilde{R}_{iT}^{(1)}(\tau) + \widetilde{R}_{iT}^{(2)}(\tau),
\end{align*}
where $\widetilde{R}_{iT}^{(k)}(\tau), k=1,2$ denote the vectors that contain the last $p$ entries of $R_{iT}^{(k)}(\tau), k=1,2$ from Lemma \ref{VCClemma}. Next note that
\begin{align*}
\frac{1}{n} \sum_{i} \widehat{W}_{iT}^{-1} (\widehat{\bm{\beta}}_{i}(\tau) - \bm{\beta}_{0}(\tau)) & = \frac{1}{n}\sum_{i=1}^{n} \Big(\frac{1}{T} \sum_{t=1}^{T} \widetilde{\eta}_{iT}(\Z_{it}, Y_{it}) \Big) \Big ( \frac{1}{T} \sum_{t=1}^{T} \phi_{i,\tau}(\Z_{it}, Y_{it}) \Big)\\
& + \frac{1}{n} \sum_{i=1}^{n} W_{iT}^{-1} \Big( \frac{1}{T}\sum_{t=1}^{T} \phi_{i,\tau}(\Z_{it}, Y_{it}) + \widetilde{R}_{iT}^{(1)}(\tau)\Big) \\
& + \frac{1}{n} \sum_{i=1}^{n} W_{iT}^{-1} \widetilde{R}_{iT}^{(2)} + \frac{1}{n} \sum_{i=1}^{n} \Big( \frac{1}{T}\sum_{t=1}^{T} \widetilde{\eta}_{iT}(\Z_{it}, Y_{it}) +\widetilde{R}_{3i}\Big) \Big(\widetilde{R}_{iT}^{(1)}(\tau) + \widetilde{R}_{iT}^{(2)}(\tau)\Big)\\
&  + \frac{1}{n} \sum_{i=1}^{n} \widetilde{R}_{3i} \Big (\frac{1}{T}\sum_{t=1}^{T} \phi_{i,\tau}(\Z_{it}, Y_{it})\Big),
\end{align*}
where 
\begin{align*}
&\frac{1}{n} \sum_{i=1}^{n} W_{iT}^{-1} \widetilde{R}_{iT}^{(2)}(\tau)= O_p\Big(\frac{\log T}{T}\Big),\\
& \frac{1}{n} \sum_{i=1}^{n} \widetilde{R}_{3i} \Big (\frac{1}{T}\sum_{t=1}^{T} \phi_{i,\tau}(\Z_{it}, Y_{it})\Big) = O_p\Big(a_T^2\frac{(\log T)^{1/2}}{T^{1/2}}\Big),\\
&  \frac{1}{n} \sum_{i=1}^{n} \Big( \frac{1}{T}\sum_{t=1}^{T} \widetilde{\eta}_{iT}(\Z_{it}, Y_{it}) + \widetilde{R}_{3i}\Big) \Big(\widetilde{R}_{iT}^{(1)}(\tau) + \widetilde{R}_{iT}^{(2)}(\tau)\Big) = O_p\Big(a_T \frac{(\log T)^{3/4}}{T^{3/4}}\Big),
\end{align*}
where the first line is a direct consequence of Lemma \ref{VCClemma} and the second line follows by~\eqref{Vit1-new},~\eqref{eq:iid-tildab} and the bound $\sup_i \Big|\frac{1}{T}\sum_{t=1}^{T} \phi_{i,\tau}(\Z_{it}, Y_{it})\Big| = O_p\Big(\frac{(\log T)^{1/2}}{T^{1/2}}\Big)$. The last line is a combination of Lemma~\ref{VCClemma}, Lemma~\ref{lem:aitbit} and~\eqref{Vit1-new},~\eqref{eq:iid-tildab}.  

Thus by (MI) 
\begin{align}
\frac{1}{n} \sum_{i} \widehat{W}_{iT}^{-1} (\widehat{\bm{\beta}}_{i}(\tau) - \bm{\beta}_{0}(\tau)) \nonumber
& = \frac{1}{n}\sum_{i=1}^{n} \Big(\frac{1}{T} \sum_{t=1}^{T} \widetilde{\eta}_{iT}(\Z_{it}, Y_{it}) \Big) \Big ( \frac{1}{T} \sum_{t=1}^{T} \phi_{i,\tau}(\Z_{it}, Y_{it}) \Big)\\
& + \frac{1}{n} \sum_{i=1}^{n} W_{iT}^{-1} \Big( \frac{1}{T}\sum_{t=1}^{T} \phi_{i,\tau}(\Z_{it}, Y_{it}) + \widetilde{R}_{iT}^{(1)}(\tau)\Big) + o_p\Big(\frac{1}{\sqrt{nT}}\Big). \label{dec1-iid}
\end{align}
For the first term, note that $\E\Big[\widetilde{\eta}_{iT}(\Z_{it}, Y_{it}) \Big] = 0$ and $\E\Big[\phi_{i,\tau}(\Z_{it}, Y_{it}) \Big]=0$ and $\widetilde{\eta}_{iT}(\Z_{it}, Y_{it})$ and $\phi_{i,\tau}(\Z_{it}, Y_{it})$ are independent across $t$ individuals. We further note that the entries of the vectors $\phi_{i,\tau}(\Z_{it}, Y_{it})$ are uniformly bounded. Combining this with~\eqref{eq:thetakl-iid} we obtain uniformly in $i$ 
\begin{equation*}
\E\Big[\Big(\frac{1}{T} \sum_{t=1}^{T} \widetilde{\eta}_{iT}(\Z_{it}, Y_{it}) \Big)\Big ( \frac{1}{T} \sum_{t=1}^{T} \phi_{i,\tau}(\Z_{it}, Y_{it})\Big) \Big] = \frac{1}{T} \E\Big[\widetilde{\eta}_{iT}(\Z_{i1}, Y_{i1}) \phi_{i, \tau}(\Z_{i1}, Y_{i1})\Big] = O\Big(\frac{1}{T}\Big),
\end{equation*}
and we have uniformly in $i$ 
\begin{align*}
& Var\Big[\Big(\frac{1}{T} \sum_{t=1}^{T} \widetilde{\eta}_{iT} (\Z_{it}, Y_{it}) \Big) \Big ( \frac{1}{T} \sum_{t=1}^{T} \phi_{i,\tau}(\Z_{it}, Y_{it})\Big) \Big]
\\
&= \E\Big[\Big(\frac{1}{T^2}\sum_{t=1}^{T} \widetilde{\eta}_{iT} (\Z_{it}, Y_{it}) \phi_{i, \tau}(\Z_{it}, Y_{it}) \Big)\Big(\frac{1}{T^2}\sum_{t=1}^{T} \widetilde{\eta}_{iT} (\Z_{it}, Y_{it}) \phi_{i, \tau}(\Z_{it}, Y_{it}) \Big)^\top\Big] + O\Big(\frac{1}{T^2}\Big)\\
& = \E\Big[\frac{1}{T^4} \sum_{t_1, t_2}\sum_{s_1, s_2} \widetilde{\eta}_{iT} (\Z_{it_1}, Y_{it_1})\phi_{i, \tau}(\Z_{is_1}, Y_{is_1}) \Big(\widetilde{\eta}_{iT} (\Z_{it_2}, Y_{it_2})\phi_{i, \tau}(\Z_{is_1}, Y_{is_1}) \phi_{i, \tau}(\Z_{is_2}, Y_{is_2})\Big)^\top\Big] + O\Big(\frac{1}{T^2}\Big)\\
& = O\Big(\frac{1}{d_T T^3} + \frac{1}{d_T T^2} + \frac{1}{T^2}\Big) = O\Big(\frac{1}{T^2 d_T}\Big)
\end{align*}
where we used~\eqref{eq:thetakl-iid} in the last step. Since $\sum_t\widetilde{\eta}_{iT}(\Z_{it}, Y_{it})$ and $\sum_t\phi_{i,\tau}(\Z_{it}, Y_{it})$ are independent across $i$, we have by assumption (MI)
\begin{align*}
& \E\Big[\frac{1}{n}\sum_{i=1}^{n} \Big(\frac{1}{T} \sum_{t=1}^{T} \widetilde{\eta}_{iT}(\Z_{it}, Y_{it}) \Big) \Big ( \frac{1}{T} \sum_{t=1}^{T} \phi_{i,\tau}(\Z_{it}, Y_{it}) \Big)\Big] = O\Big(\frac{1}{T}\Big)= o\Big(\frac{1}{\sqrt{nT}}\Big),\\
& Var\Big[\frac{1}{n}\sum_{i=1}^{n} \Big(\frac{1}{T} \sum_{t=1}^{T} \widetilde{\eta}_{iT}(\Z_{it}, Y_{it}) \Big) \Big ( \frac{1}{T} \sum_{t=1}^{T} \phi_{i,\tau}(\Z_{it}, Y_{it}) \Big)\Big] = O\Big(\frac{1}{n}\frac{1}{T^2 d_T}\Big) = o\Big(\frac{1}{nT}\Big).
\end{align*}

For the second term in~\eqref{dec1-iid} note first that since $Var(\widetilde{R}_{iT}^{(1)}(\tau)) = o\Big(\frac{1}{T}\Big)$ by \eqref{vcc5} in Lemma~\ref{VCClemma} and since $W_{iT}^{-1}$ is deterministic and by independence across individuals, we have $Var\Big(\frac{1}{n}\sum_{i=1}^{n} W_{iT}^{-1} \widetilde{R}_{iT}^{(1)}(\tau)\Big) = o\Big(\frac{1}{nT}\Big)$. Additionally, by Lemma~\ref{VCClemma}, $\sup_{i}\sup_{\eta \in \T}\|\E[\widetilde{R}_{iT}^{(1)}(\eta)]\| = O\Big(\frac{\log T}{T}\Big)$, hence $\E\Big(\frac{1}{n}\sum_{i=1}^{n} W_{iT}^{-1} \widetilde{R}_{iT}^{(1)}(\tau)\Big) = O\Big(\frac{\log T}{T}\Big)$. Putting together, by the Chebyshev inequality and under Assumption (MI),
\begin{equation*}
\frac{1}{n}\sum_{i=1}^{n} W_{iT}^{-1} \widetilde{R}_{iT}^{(1)}(\tau) = O\Big(\frac{\log T}{T}\Big) + o_p\Big(\frac{1}{\sqrt{nT}}\Big) = o_p\Big(\frac{1}{\sqrt{nT}}\Big).
\end{equation*}
Next, 
\begin{align*}
 \frac{1}{n}\sum_{i=1}^{n} W_{iT}^{-1} \frac{1}{T}\sum_{t=1}^{T} \phi_{i,\tau}(\Z_{it}, Y_{it})= \frac{1}{n} \sum_{i=1}^{n} (W_{iT}^{-1} - W_i^{-1}) \frac{1}{T}\sum_{t=1}^{T} \phi_{i,\tau}(\Z_{it}, Y_{it})  + \frac{1}{n} \sum_{i=1}^{n} W_i^{-1} \frac{1}{T}\sum_{t=1}^{T} \phi_{i,\tau}(\Z_{it}, Y_{it})
\end{align*}
and by Lemma~\ref{lemma1}, $\|W_{iT}^{-1} - W_i^{-1}\|= o(1)$ uniformly in $i$ and $Var\Big(\frac{1}{T}\sum_{t=1}^{T} \phi_{i,\tau}(\Z_{it}, Y_{it})\Big) = O\Big(\frac{1}{T}\Big)$, hence 
\begin{equation*}
Var\Big(\frac{1}{n} \sum_{i=1}^{n} (W_{iT}^{-1} - W_i^{-1}) \frac{1}{T}\sum_{t=1}^{T} \phi_{i,\tau}(\Z_{it}, Y_{it})\Big) = o\Big(\frac{1}{nT}\Big).
\end{equation*}
Since $W_{iT}^{-1}$ and $W_i^{-1}$ are deterministic, we have 
\begin{equation*}
\E\Big(\frac{1}{n} \sum_{i=1}^{n} (W_{iT}^{-1} - W_i^{-1}) \frac{1}{T}\sum_{t=1}^{T} \phi_{i,\tau}(\Z_{it}, Y_{it})\Big) =0.
\end{equation*}
Therefore, by Chebyshev inequality,
\begin{equation*}
\frac{1}{n} \sum_{i=1}^{n} (W_{iT}^{-1} - W_i^{-1}) \frac{1}{T}\sum_{t=1}^{T} \phi_{i,\tau}(\Z_{it}, Y_{it})= o_p\Big(\frac{1}{\sqrt{nT}}\Big).
\end{equation*}
Putting together, we have 
\begin{align}\label{eqalmost}
\frac{1}{n} \sum_{i} \widehat{W}_{iT}^{-1} (\widehat{\bm{\beta}}_{i}(\tau) -\bm{\beta}_{0}(\tau)) = \frac{1}{n} \sum_{i=1}^{n} W_i^{-1} \frac{1}{T}\sum_{t=1}^{T} \phi_{i,\tau}(\Z_{it}, Y_{it}) + o_p\Big(\frac{1}{\sqrt{nT}}\Big).
\end{align}
Next, note that by Lemma~\ref{lemma1}, \eqref{Vit1-new} and the definition of $\widehat{W}_{iT}, W_{i}$ 
\begin{equation*}
\frac{1}{n} \sum_{i=1}^n\widehat{W}_{iT}^{-1} = \frac{1}{n} \sum_{i=1}^{n} W_i^{-1} + \frac{1}{nT} \sum_{i=1}^n \sum_{t=1}^{T} \widetilde{\eta}_{iT}(\Z_{it},Y_{it})
+ o_P(1) = \frac{1}{n} \sum_{i=1}^{n} W_i^{-1} + o_P(1),
\end{equation*}
where the last equation follows from the properties of $\widetilde{\eta}_{iT}(\Z_{it},Y_{it})$. Plugging this into the definition of $\widehat{\bm{\beta}}_{MD}(\tau)$ in combination with \eqref{eqalmost}, the result in \eqref{eq:main} follows by an application of the Lindeberg CLT after some straightforward computations. \hfill $\Box$

\bigskip

\noindent \textbf{Proof of Lemma \ref{lemma:CovHK}.}
The proof follows directly from Lemma \ref{lem:aitbit}.
$\Box$

%
%

\subsection{Proof of Theorem~\ref{th2} and Lemma~\ref{lemma:CovD}}

\bigskip

\begin{lemma} \label{lem:pisiit_beta}
Under assumptions (A0)--(A3), (D1), and if $\log n \leq \C_4 \log T$ for some constant $\C_4$ we have 
\begin{equation}\label{psiit2}
\sup_i \Big\| \frac{1}{T}\sum_{t=1}^{T} \frac{\psi_{i,\tau+d_T}(\Z_{it}, Y_{it}) - \psi_{i,\tau - d_T}(\Z_{it}, Y_{it})}{2d_T} \Big\| = O_p\Big(\frac{\log T}{(Td_T)^{1/2} }\Big). 
\end{equation}
\end{lemma}

\bigskip

For the next result we use the following additional notation. Let 
\begin{align*}
w_{it} &:= \Z_{it} (\tau - \1(Y_{it} \leq \bm{\gamma}_{i0}(\tau) ^\top \Z_{it})),
\\
\mu_{i,j}(b) &:= \E[(\Z_{i1}\Z_{i1+j}^\top + \Z_{i1+j}\Z_{i1}^\top)(\tau - \1(Y_{i1} \leq \Z_{i1}^\top b))(\tau - \1(Y_{i1+j}\leq \Z_{i1+j}^\top b))].
\end{align*} 
Under Assumptions (A1), (D2) there exist matrices $D_{i,j,k}(\tau) \in \R^{(p+1)\times(p+1)}$ with bounded elements such that 
\begin{equation}\label{muij}
\sup_{i,j} \Big\|\mu_{i,j}(\widehat{\bm{\gamma}}_{i}(\tau)) - \mu_{i,j}(\bm{\gamma}_{i0}(\tau)) - \sum_{k=1}^{p+1}D_{i,j,k}(\tau) (\widehat{\bm{\gamma}}_{i,k}(\tau) - \bm{\gamma}_{i0,k}(\tau))\Big\| = O\Big(\sup_i \|\widehat{\bm{\gamma}}_{i}(\tau) - \bm{\gamma}_{i0}(\tau)\|^2\Big).
\end{equation}
Define
\begin{align*}
A_{iT} & := \frac{\tau(1-\tau)}{T}\sum_{t=1}^{T} \Z_{it} \Z_{it}^\top  + \sum_{1\leq j \leq m_T}\Big(1-\frac{j}{T}\Big) \Big(\frac{1}{T}\sum_{t\in T_j}(w_{it}w_{it+j}^\top + w_{it+j}w_{it}^\top) \Big).
\end{align*}

\begin{lemma}\label{lem:ait_beta}
Under Assumptions (A0)--(A3) and (D1)--(D2), (MD) 
\begin{equation*}
\widetilde{A}_{iT}- A_{iT} = \sum_{1\leq j \leq m_T}\Big(1-\frac{j}{T}\Big)\frac{|T_j|}{T}\sum_{k=1}^{p+1} D_{i,j,k}(\tau) \frac{1}{T}\sum_{t=1}^{T} (B_{i}^{-1}\psi_{i,\tau}(\Z_{it}, Y_{it}))_k + R_{4i},
\end{equation*}
with
\begin{equation}\label{Ait3}
\sup_i \|R_{4i}\| = O_P\Big(m_T (\log T + m_T)^{1/2}\Big(\frac{\log T}{T}\Big)^{3/4}\Big),
\end{equation}
and 
\begin{equation}\label{Ait4}
\sup_i \Big\|A_{iT} - \E[A_{iT}] \Big\|= O_p\Big( m_T\Big(\frac{m_T \log T}{T}\Big)^{1/2} \Big),
\end{equation}
and 
\begin{equation}\label{Ait4.1}
\sup_i \Big \|\widetilde{A}_{iT} - \E[A_{iT}]\Big\|=O_p\Big(m_T \Big(\frac{m_T\log T}{T}\Big)^{1/2}\Big),
\end{equation}
and provided that $m_T\to \infty$, 
\begin{equation} \label{Ait5}
\sup_i \Big\| \E[A_{iT}] - \widetilde{A}_i \Big\| = o(1).
\end{equation}
\end{lemma}

\bigskip

Now we state a lemma that is parallel to Lemma \ref{lemma2} under $\beta$-mixing sequences. 
\begin{lemma}
\label{lemma2_new}
Under Assumptions (A0)--(A3) and (D1)--(D2), (MD)  
\begin{equation*}
\sup_{i,t}|\widehat{e}_{it}-e_{it}|= O_p\Big(\frac{\log T}{(Td_T)^{1/2}}\Big),
\end{equation*}
and 
\begin{equation*}
\sup_{i,t}\Big|\widehat{f}_{it}-f_{it}-\frac{e_{it}-\widehat{e}_{it}}{e_{it}^2}\Big| = O_p\Big(\frac{(\log T)^2}{Td_T}\Big).
\end{equation*}
\end{lemma}

\bigskip

Now we present a lemma parallel to Lemma \ref{lem:aitbit} for the $\beta$-mixing sequences. 
\begin{lemma}\label{lem:aitbit2}
Under assumptions (A0)--(A3) and (D1)--(D2), (MD), 
\begin{equation*}
\widehat{B}_{iT} - B_{iT} = \frac{1}{T}\sum_{t=1}^{T} \sum_{k=1}^{p+1} \frac{\psi_{i,\tau+d_T,k}(\Z_{it}, Y_{it}) - \psi_{i,\tau - d_T,k}(\Z_{it}, Y_{it})}{2d_T} \E\Big[\frac{(B_i^{-1}\Z_{i1})_k \Z_{i1}\Z_{i1}^\top}{e_{i1}^2} \Big] + R_{2i},
\end{equation*}
with
\begin{equation} \label{Bit2_beta}
\sup_{i} \Big\|\frac{1}{T}\sum_{t=1}^{T} \sum_{k=1}^{p+1} \frac{\psi_{i,\tau+d_T,k}(\Z_{it}, Y_{it}) - \psi_{i,\tau - d_T,k}(\Z_{it}, Y_{it})}{2d_T} \E\Big[\frac{(B_i^{-1}\Z_{i1})_k \Z_{i1}\Z_{i1}^\top}{e_{i1}^2} \Big]\Big\| = O_p\Big(\frac{\log T}{(Td_T)^{1/2}}\Big),
\end{equation} 
and 
\begin{equation} \label{Bit3_beta}
\sup_{i} \|R_{2i}\| = O_p\Big(\frac{(\log T)^2}{Td_T} + \Big(\frac{\log T}{(Td_T)^{1/2}}\Big)\Big(\frac{\log T}{T}\Big)^{1/2}+ \frac{(\log T)^{5/4}}{d_T T^{3/4}}\Big).
\end{equation}
Moreover,
\begin{equation} \label{Bit1_beta}
\sup_i \|B_{iT}-\E[B_{iT}]\| = O_p\Big(\frac{(\log T)^{1/2}}{T^{1/2}}\Big).
\end{equation}
\end{lemma}


\bigskip 
\bigskip

\noindent
\textbf{Proof of Theorem \ref{th2}.}
We use the following notations for various matrix objects: 
\begin{align*}
 \widetilde{V}_{iT} & = \widehat{B}_{iT}^{-1}\widetilde{A}_{iT}\widehat{B}_{iT}^{-1}\\
V_{iT} & = \E[B_{iT}]^{-1}\E[A_{iT}] \E[B_{iT}]^{-1}\\
V_i & = B_i^{-1} \widetilde{A}_iB_{i}^{-1},
\end{align*}
and let $\widetilde{W}_{iT}$, $W_{iT}$ and $W_{i}$ denote the lower $p\times p$ submatrix of $\widetilde{V}_{iT}$, $V_{iT}$ and $V_i$ respectively. Observe that 
\begin{align}\label{Vit1_beta}
\widetilde{V}_{iT} - V_{iT}  & = -\E[B_{iT}]^{-1}\E[A_{iT}]\E[B_{iT}]^{-1} (\widehat{B}_{iT} - \E[B_{iT}])\E[B_{iT}]^{-1} \nonumber
\\
&\quad -\E[B_{iT}]^{-1} (\widehat{B}_{iT} - \E[B_{iT}]) \E[B_{iT}]^{-1}\E[A_{iT}]\E[B_{iT}]^{-1} \nonumber
\\
& \quad  + O\Big(\|\widehat{B}_{iT} - \E[B_{iT}]\|^2+\|\widetilde{A}_{iT}-(\E[A_{iT}])\|\Big) \nonumber
\\
& =:  \frac{1}{T} \sum_{t=1}^T \eta_{iT}(\Z_{it}, Y_{it}) + R_{3i}, 
\end{align}
where  
\begin{align*}
& \eta_{iT}(\Z_{it}, Y_{it})  
\\
:=& - \E[B_{iT}]^{-1}\E[A_{iT}]\E[B_{iT}]^{-1} \Big\{\sum_{k=1}^{p+1}  \frac{\psi_{i,\tau+d_T,k}(\Z_{it}, Y_{it}) - \psi_{i,\tau - d_T,k}(\Z_{it}, Y_{it})}{2d_T} \E\Big[\frac{(B_i^{-1}\Z_{i1})_k \Z_{i1}\Z_{i1}^\top}{e_{i1}^2} \Big] \Big\} \E[B_{iT}]^{-1}
\\
&- \E[B_{iT}]^{-1} \Big\{\sum_{k=1}^{p+1}  \frac{\psi_{i,\tau+d_T,k}(\Z_{it}, Y_{it}) - \psi_{i,\tau - d_T,k}(\Z_{it}, Y_{it})}{2d_T} \E\Big[\frac{(B_i^{-1}\Z_{i1})_k \Z_{i1}\Z_{i1}^\top}{e_{i1}^2} \Big] \Big\} \E[B_{iT}]^{-1}\E[A_{iT}]\E[B_{iT}]^{-1}. 
\end{align*}
and by Lemma~\ref{lem:aitbit2} and Lemma~\ref{lem:ait_beta} we obtain that under (MD) 
\begin{equation*}
\sup_i \|R_{3i}\|= O_p\Big(\frac{(\log T)^2}{Td_T} + \Big(\frac{m_T^3 \log T}{T} \Big)^{1/2}\Big).
\end{equation*}
Since $\widetilde{W}_{iT}$ are the lower $p\times p$ sub-matrix of $\widetilde{V}_{iT}$, we have 
\begin{align}
\widetilde{W}_{iT}^{-1} - W_{iT}^{-1} &= - W_{iT}^{-1} (\widetilde{W}_{iT} - W_{iT}) W_{iT}^{-1} + O\Big( \|\widetilde{W}_{iT}- W_{iT}\|^2\Big)\\
& =: \frac{1}{T}\sum_{t=1}^T \widetilde{\eta}_{iT}(\Z_{it}, Y_{it}) + \widetilde{R}_{3i}
\end{align}
and the terms $\widetilde{\eta}_{iT}(\Z_{it}, Y_{it}) $ are independent across $i$ and centered and 
\[
\sup_i\|\widetilde{R}_{3i}\| =O_p\Big(\frac{(\log T)^2}{Td_T} + \Big(\frac{m_T^3 \log T}{T} \Big)^{1/2}\Big).
\]
Let $\phi_{i,\tau}(\Z_{it}, Y_{it})$ denote the last $p$ entries of $-B_i^{-1} \psi_{i,\tau}(\Z_{it}, Y_{it})$. Then by Lemma \ref{VCClemma_beta}
\begin{align*}
\widehat{\bm{\beta}}_{i}(\tau) - \bm{\beta}_{0}(\tau) & = \frac{1}{T}\sum_{t=1}^{T} \phi_{i,\tau} (\Z_{it}, Y_{it}) + \widetilde{R}_{iT}^{(1)}(\tau) + \widetilde{R}_{iT}^{(2)}(\tau),
\end{align*}
where $\widetilde{R}_{iT}^{(k)}(\tau), k=1,2$ denote the vectors that contain the last $p$ entries of $R_{iT}^{(k)}(\tau), k=1,2$ from Lemma \ref{VCClemma_beta}. Note that
\begin{align*}
&\frac{1}{n}\sum_{i=1}^{n} \widetilde{W}_{iT}^{-1} \Big( \widehat{\bm{\beta}}_{i}(\tau) - \bm{\beta}_{0}(\tau) \Big) 
\\
=~& \frac{1}{n}\sum_{i=1}^{n} \Big(\frac{1}{T} \sum_{t=1}^{T} \widetilde{\eta}_{iT}(\Z_{it}, Y_{it}) \Big) \Big ( \frac{1}{T} \sum_{t=1}^{T} \phi_{i,\tau}(\Z_{it}, Y_{it}) \Big)
 + \frac{1}{n} \sum_{i=1}^{n} W_{iT}^{-1} \Big( \frac{1}{T}\sum_{t=1}^{T} \phi_{i,\tau}(\Z_{it}, Y_{it}) + \widetilde{R}_{iT}^{(1)}(\tau)\Big) 
\\ 
& + \frac{1}{n} \sum_{i=1}^{n} W_{iT}^{-1} \widetilde{R}_{iT}^{(2)} + \frac{1}{n} \sum_{i=1}^{n} \Big( \frac{1}{T}\sum_{t=1}^{T} \widetilde{\eta}_{iT}(\Z_{it}, Y_{it}) +\widetilde{R}_{3i}\Big) \Big(\widetilde{R}_{iT}^{(1)}(\tau) + \widetilde{R}_{iT}^{(2)}(\tau)\Big)
\\
&  + \frac{1}{n} \sum_{i=1}^{n} \widetilde{R}_{3i} \Big (\frac{1}{T}\sum_{t=1}^{T} \phi_{i,\tau}(\Z_{it}, Y_{it})\Big),
\end{align*}
and 
\begin{align*}
&\frac{1}{n} \sum_{i=1}^{n} W_{iT}^{-1} \widetilde{R}_{iT}^{(2)}(\tau)= O_p\Big(\frac{\log T}{T}\Big),
\\
& \frac{1}{n} \sum_{i=1}^{n} \widetilde{R}_{3i} \Big(\frac{1}{T}\sum_{t=1}^{T} \phi_{i,\tau}(\Z_{it}, Y_{it})\Big) = O_p\Big(\Big(\frac{(\log T)^2}{Td_T} + \Big(\frac{m_T^3 \log T}{T} \Big)^{1/2}\Big)\frac{(\log T)^{1/2}}{T^{1/2}}\Big),
\\
&  \frac{1}{n} \sum_{i=1}^{n} \Big( \frac{1}{T}\sum_{t=1}^{T} \widetilde{\eta}_{iT}(\Z_{it}, Y_{it}) + \widetilde{R}_{3i}\Big) \Big(\widetilde{R}_{iT}^{(1)}(\tau) + \widetilde{R}_{iT}^{(2)}(\tau)\Big) = O_p\Big(\Big(\frac{\log T}{(Td_T)^{1/2}}+\Big(\frac{m_T^3 \log T}{T} \Big)^{1/2}\Big) \frac{(\log T)^{5/4}}{T^{3/4}}\Big),
\end{align*}
where the first line is a direct consequence of Lemma \ref{VCClemma_beta} and the second line follows by \eqref{Vit1_beta} and the bound $\sup_i \Big|\frac{1}{T}\sum_{t=1}^{T} \phi_{i,\tau}(\Z_{it}, Y_{it})\Big| = O_p\Big(\frac{(\log T)^{1/2}}{T^{1/2}}\Big)$. The last line is a combination of Lemma~\ref{VCClemma_beta}, Lemma~\ref{lem:aitbit2}, Lemma~\ref{lem:ait_beta} and~\eqref{Vit1_beta}.  

Thus under (MD) 
\begin{align*}
\frac{1}{n} \sum_{i} \widehat{W}_{iT}^{-1} (\widehat{\bm{\beta}}_{i}(\tau) - \bm{\beta}_{0}(\tau)) 
& = \frac{1}{n}\sum_{i=1}^{n} \Big(\frac{1}{T} \sum_{t=1}^{T} \widetilde{\eta}_{iT}(\Z_{it}, Y_{it}) \Big) \Big( \frac{1}{T} \sum_{t=1}^{T} \phi_{i,\tau}(\Z_{it}, Y_{it}) \Big)\\
&\quad + \frac{1}{n} \sum_{i=1}^{n} W_{iT}^{-1} \Big( \frac{1}{T}\sum_{t=1}^{T} \phi_{i,\tau}(\Z_{it}, Y_{it}) + \widetilde{R}_{iT}^{(1)}(\tau)\Big) + o_p\Big(\frac{1}{\sqrt{nT}}\Big).
\end{align*}
To bound the first term, denote by $\theta_t$ the $(k, \ell)$'s entry of the matrix $\widetilde{\eta}_{iT}(\Z_{it}, Y_{it})$ and let $\xi_{t}$ denote the $j$-th entry of the vector $\phi_{i,\tau}(\Z_{it}, Y_{it})$ (the dependence on $i$ is dropped for notational convenience). Note that for arbitrary $j,k,\ell$ we have $\E[\theta_t] = 0, \E[\xi_t] = 0$. 

Also note that there exists a constant $C$ independent of $j,k,\ell$ such that for all sufficiently large $T$ we have $|\xi_t| \leq C$ almost surely, $|\theta_t |\leq Cd_T^{-1}$, $\E|\theta_t|\leq C$ and thus $\E|\xi_t\theta_t|\leq C \E|\theta_t| \leq C^2$, $\E[\theta_{t}^2] \leq Cd_T^{-1} \E|\theta_t| \leq C^2 d_T^{-1}$ and $Var(\xi_t) \leq C^2$, which implies for any $t_1,t_2$
\begin{equation*}
\E[\xi_{t_1}^2\theta_{t_2}^2] \leq C^2\E[\theta_{t_2}^2] \leq C^3 d_T^{-1}.
\end{equation*}
Applying Lemma C.1 of \cite{KatoGalvaoMontes-Rojas12} with $\delta = 1$ we find that $|\E[\theta_{t_1} \xi_{t_2}]| \leq \widetilde{C} d_T^{-1/2} \beta(|t_1-t_2|)^{1/2}$ and thus we have uniformly in $i,k,\ell,j$
\begin{align*}
\Big| \E\Big[\Big(\frac{1}{T} \sum_{t=1}^{T} \theta_t \Big)\Big ( \frac{1}{T} \sum_{t=1}^{T} \xi_t\Big) \Big]\Big| 
& \leq \frac{1}{T^2}\sum_{t_1}\sum_{t_2} |\E[\theta_{t_1} \xi_{t_2}]|
\\
& \leq \frac{1}{T^2}\sum_{t_1} \sum_{t_2} \min \Big\{C^2, \widetilde{C}d_T^{-1/2}\beta(|t_1-t_2|)^{1/2}\Big\}
\\
&  \leq \frac{2}{T}\sum_{t=1}^\infty \min \Big\{C^2, \widetilde{C}d_T^{-1/2}\beta(t)^{1/2}\Big\} = O\Big(\frac{\log T}{T}\Big),
\end{align*}
where the last inequality follows by similar arguments as the proof of Lemma~\ref{varb} since $|\log d_T| = O(\log T)$.

Next note that, since the variables $\theta_t,\xi_t$ are centered, we have by elementary properties of cumulants
\begin{align*}
\E[ \theta_{t_1}\theta_{t_2}\xi_{s_1}\xi_{s_2}] =& cum(\theta_{t_1},\theta_{t_2},\xi_{s_1},\xi_{s_2}) + \E[\theta_{t_1}\theta_{t_2}]\E[\xi_{s_1}\xi_{s_2}] + \E[\theta_{t_1}\xi_{s_2}]\E[\xi_{s_1}\theta_{t_2}]
+ \E[\theta_{t_1}\xi_{s_1}]\E[\theta_{t_2}\xi_{s_2}].
\end{align*}
Let $m(t_1,t_2,s_1,s_2)$ denote the maximal distance between any two of the indices $s_1,s_2,t_1,t_2$.

Then by similar arguments as is the proof of Proposition 3.1 in \cite{KVDH2016} one can show that for some constant $C$ independent of $i,T$ (note that $d_T \theta_t$ is uniformly bounded and that cumulants are linear in each of their arguments and apply Lemma 2.6 of \cite{Yu1994} instead of $\alpha$-mixing to bound terms of the form $\E[\psi_1\psi_2] - \E[\psi_1]\E[\psi_2]$ where $\psi_1,\psi_2$ denote random variables that can be products of the $\theta_{t_1},\xi_{s_i}$) to obtain
\begin{equation*}
|cum(\theta_{t_1},\theta_{t_2},\xi_{s_1},\xi_{s_2})| \leq C d_T^{-2} \beta(\lfloor m(t_1,t_2,s_1,s_2)/4\rfloor).
\end{equation*}
Moreover, by applying Lemma 2.6 of \cite{Yu1994} we find that 
\begin{align*}
|\E[\theta_{t_1}\theta_{t_2}]| &\leq C d_T^{-2} \beta(|t_1-t_2|),
\\
|\E[\xi_{t_1}\xi_{t_2}]| &\leq C\beta(|t_1-t_2|).
\end{align*} 
Thus
\begin{multline*}
Var\Big[\Big(\frac{1}{T} \sum_{t=1}^{T} \theta_t \Big) \Big ( \frac{1}{T} \sum_{t=1}^{T} \xi_t\Big) \Big] 
\leq  \E\Big[\Big(\frac{1}{T^2}\sum_{t_1}\sum_{t_2} \theta_{t_1} \xi_{t_2} \Big)^2\Big] 
= \frac{1}{T^4} \sum_{t_1}\sum_{ t_2}\sum_{s_1}\sum_{s_2}\E[ \theta_{t_1}\theta_{t_2}\xi_{s_1}\xi_{s_2}]
\\
\leq  \frac{1}{T^4} \sum_{t_1}\sum_{ t_2}\sum_{s_1}\sum_{s_2} C (d_T^{-2} \beta(\lfloor m(t_1,t_2,s_1,s_2)/4\rfloor) + O\Big(\frac{1}{T^2d_T^2}\Big)
= O\Big(\frac{1}{T^2d_T^2}\Big),
\end{multline*}
since
\begin{align*}
\frac{1}{T^4} \sum_{t_1,t_2,s_1,s_2} C \beta(\lfloor m(t_1,t_2,s_1,s_2)/4\rfloor) &\leq \frac{C}{T^4} \sum_{m \in \mathbb{Z}} |\{(s_1,s_2,t_1,t_2): m(t_1,t_2,s_1,s_2) = m\}|\beta(\lfloor m/4 \rfloor)
\\
&\leq \frac{C}{T^4} \sum_{m \in \mathbb{Z}} T (2m+1)^3\beta(\lfloor m/4 \rfloor) = O\Big(\frac{1}{T^3}\Big).
\end{align*}
Now note that by assumptions the collections of random variables $\{\widetilde{\eta}_{iT}(\Z_{it}, Y_{it}), \phi_{i,\tau}(\Z_{it}, Y_{it}): t=1,...,T\}$ are independent for different values of $i$ and thus we have under (MD)
\begin{align*}
& \E\Big[\frac{1}{n}\sum_{i=1}^{n} \Big(\frac{1}{T} \sum_{t=1}^{T} \widetilde{\eta}_{iT}(\Z_{it}, Y_{it}) \Big) \Big ( \frac{1}{T} \sum_{t=1}^{T} \phi_{i,\tau}(\Z_{it}, Y_{it}) \Big)\Big] = O\Big(\frac{\log T}{T}\Big)= o\Big(\frac{1}{\sqrt{nT}}\Big),
\\
& Var\Big[\frac{1}{n}\sum_{i=1}^{n} \Big(\frac{1}{T} \sum_{t=1}^{T} \widetilde{\eta}_{iT}(\Z_{it}, Y_{it}) \Big) \Big ( \frac{1}{T} \sum_{t=1}^{T} \phi_{i,\tau}(\Z_{it}, Y_{it}) \Big)\Big] = O\Big(\frac{1}{n}\Big) O\Big(\frac{1}{T^2d_T^2}\Big) = o\Big(\frac{1}{nT}\Big).
\end{align*}
This shows
\begin{equation*}
\frac{1}{n}\sum_{i=1}^{n} \Big(\frac{1}{T} \sum_{t=1}^{T} \widetilde{\eta}_{iT}(\Z_{it}, Y_{it}) \Big) \Big ( \frac{1}{T} \sum_{t=1}^{T} \phi_{i,\tau}(\Z_{it}, Y_{it}) \Big) = o_p\Big(\frac{1}{\sqrt{nT}}\Big).
\end{equation*}
The fact that
\begin{equation*}
\frac{1}{n} \sum_{i=1}^{n} W_{iT}^{-1}\widetilde{R}_{iT}^{(1)}(\tau) = o_p\Big(\frac{1}{\sqrt{nT}}\Big)
\end{equation*}
follows by exactly the same arguments as in the independent case (see the proof of Theorem~\ref{th1}) once we replace the use of Lemma~\ref{VCClemma} there by Lemma~\ref{VCClemma_beta} and note that $\sup_i \|W_{iT}\| = O(1)$ with $W_{iT}$ being deterministic.
\begin{equation*}
\frac{1}{n} \sum_{i} \widetilde{W}_{iT}^{-1} (\widehat{\bm{\beta}}_{i}(\tau) - \bm{\beta}_{0}(\tau)) = \frac{1}{n} \sum_{i=1}^{n} W_{iT}^{-1} \frac{1}{T}\sum_{t=1}^{T} \phi_{i,\tau}(\Z_{it}, Y_{it}) + o_p\Big(\frac{1}{\sqrt{nT}}\Big).
\end{equation*}
The weak convergence in \eqref{eq:main_dep} follows by similar arguments as in the independent case (see the proof of Theorem \ref{th1}) once we observe that $\sup_i \| W_{iT}^{-1} - W_i^{-1}\| = o(1)$, details are omitted for the sake of brevity. \hfill $\Box$

\bigskip

\noindent \textbf{Proof of Lemma \ref{lemma:CovD}.}
The proof follows directly from Lemma \ref{lem:ait_beta}.
$\Box$



\section{Proofs of technical results}

\subsection{Proofs of technical results related to section~\ref{sec:indep}}

\begin{proof}[Proof of Lemma~\ref{lemma1}]
Under Assumptions (A1), (A2), we have
\begin{align*}
\|\E[B_{iT}] - B_i \|& =\Big \|\E\Big[\Big(f_{i1} - f_{Y\mid \Z}(q_{i, \tau}(\Z_{i1})\mid\Z_{i1})\Big)\Z_{i1}\Z_{i1}^\top\Big]\Big\|\\
& = \Big\|\E\Big[\Big(\frac{2d_T}{q_{i,\tau+d_T}(\Z_{i1})-q_{i,\tau-d_T}(\Z_{i1})} - f_{Y\mid \Z}(q_{i, \tau}(\Z_{i1})\mid\Z_{i1})\Big)\Z_{i1}\Z_{i1}^\top\Big]\Big\|\\
&\leq M^2 \E\Big|\frac{2d_T}{q_{i,\tau+d_T}(\Z_{i1})-q_{i,\tau-d_T}(\Z_{i1})} - f_{Y\mid \Z}(q_{i, \tau}(\Z_{i1})\mid\Z_{i1})\Big|= o(1).
\end{align*}
The statement for $W_{iT}$ follows from its definition.
\end{proof}

\medskip

\begin{proof}[Proof of Lemma~\ref{lemma2}]
By definition, 
\begin{equation*}
\widehat{e}_{is} - e_{is} = \Z_{is}^\top \Big(\big(\widehat{\bm{\gamma}}_{i}(\tau + d_T) - \bm{\gamma}_{i0}(\tau + d_T) \big)- \big(\widehat{\bm{\gamma}}_{i}(\tau - d_T) - \bm{\gamma}_{i0}(\tau - d_T) \big)\Big)/2d_T.
\end{equation*}
We know from Lemma \ref{VCClemma} that 
\begin{equation*}
\Z_{is}^\top (\widehat{\bm{\gamma}}_{i}(\tau \pm d_T) - \bm{\gamma}_{i0}(\tau \pm d_T)) = -\frac{1}{T} \Z_{is}^\top B_i^{-1}\sum_{t=1}^{T}  \Z_{it} \Big(\1 \{Y_{it} \leq q_{i, \tau \pm d_T}(\Z_{it})\} - (\tau \pm d_T)\Big) + O_p\Big(\frac{(\log T)^{3/4}}{T^{3/4}}\Big).
\end{equation*}
Hence with $U_{it} := F_{Y|\Z}(Y_{it}|\Z_{it}) \sim U[0,1]$ independent of $\Z_{it}$, 
\begin{equation}\label{eq:eit1}
\widehat{e}_{is} - e_{is} = -\frac{1}{2Td_T} \Z_{is}^\top B_i^{-1}\sum_{t=1}^{T} \Z_{it} \Big( \1 \{U_{it} \leq \tau + d_T\} - \1 \{U_{it} \leq \tau - d_T\} - 2d_T\Big)+ O_p\Big( \frac{(\log T)^{3/4}}{d_T T^{3/4}}\Big).
\end{equation}
Define the vectors 
\begin{equation*}
M_{it} :=  \Z_{it} \Big( \1 \{U_{it} \leq \tau + d_T\} - \1 \{U_{it} \leq \tau - d_T\} - 2d_T\Big)/2d_T.
\end{equation*}
Fix an arbitrary $k \in \{1,...,p+1\}$ and let $M_{it,k}$ denote the $k$-th entry of the vector $M_{it}$. It is clear that $\E[M_{it,k}] = 0$ and $Var[M_{it,k}] = \frac{\mathcal{C}_1}{d_T}$ for some constant $\mathcal{C}_1$  under Assumption (A1). Under Assumption (A1), we also have $\sup_{i,t,k} |M_{it,k}| \leq \mathcal{C}_2/d_T$ for some constant $\mathcal{C}_2>0$. Apply the Bernstein inequality to obtain 
\begin{align*}
\P\Big(\Big|\sum_{t=1}^{T} M_{it,k}\Big| > T\epsilon\Big) 
&\leq 2\exp \Big(-\frac{\frac{1}{2}T^2\epsilon^2}{\sum_{t=1}^{T} \E[M_{it,k}^2] + \frac{1}{3}\mathcal{C}_2 (d_T)^{-1} T\epsilon}\Big)
\\
& = 2\exp\Big(-\frac{\frac{1}{2}T^2\epsilon^2}{\mathcal{C}_1T(d_T)^{-1}+ \frac{1}{3}\mathcal{C}_2 (d_T)^{-1} T\epsilon}\Big).
\end{align*}
Take $\epsilon = \mathcal{C}_3 T^{-1/2}d_T^{-1/2}(\log T)^{1/2}$ for a constant $\mathcal{C}_3$ which will be determined later. Under assumption (MI), $\frac{\log T}{Td_T} \to 0$, so that $\epsilon\to 0$ and the right hand side of the inequality becomes
\begin{equation*}
2\exp\Big(-\frac{1}{2}\frac{\mathcal{C}_3^2d_T^{-1}\log T}{\mathcal{C}_1d_T^{-1} + \frac{1}{3}\mathcal{C}_2 d_T^{-1}T^{-1/2}d_T^{-1/2}(\log T)^{1/2}}\Big)\leq 2\exp\Big(-\frac{1}{4} \mathcal{C}_3^2 \log T/\mathcal{C}_1\Big)
\end{equation*}
where the last inequality holds for $T$ sufficiently large.

Now under Assumption (I) we have $\log n \leq \mathcal{C}_4 \log T$ for some constant $\mathcal{C}_4>0$, so that
\begin{align*}
\P\Big(\sup_k\sup_i \Big|\frac{1}{T}\sum_{t=1}^{T} M_{it,k}\Big|>\epsilon \Big) &\leq \sum_{i} \P\Big( \Big|\frac{1}{T}\sum_{t=1}^{T} M_{it,k}\Big|>\epsilon \Big)
\\
& \leq 2n\exp\Big( -\frac{\mathcal{C}_3^2}{4\mathcal{C}_1}  \log T\Big) \to 0,
\end{align*}
by taking $\mathcal{C}_3^2 > 4\mathcal{C}_1\mathcal{C}_4$. Hence 
\begin{equation*}
\sup_i \Big\|\frac{1}{T}\sum_{t=1}^{T} M_{it}\Big\| = O_p\Big( T^{-1/2}d_T^{-1/2}(\log T)^{1/2}\Big),
\end{equation*}
and combined with~\eqref{eq:eit1} and the fact that $\sup_i \|B_i^{-1}\| = O(1)$ (here $\|\cdot\|_{op}$ denotes the operator norm) we obtain 
\begin{equation*}
\sup_{i,s}|\widehat{e}_{is} - e_{is}| = O_p\Big(\frac{(\log T)^{1/2}}{(Td_T)^{1/2}} + \frac{(\log T)^{3/4}}{d_T T^{3/4}}\Big) = O_p\Big(\frac{(\log T)^{1/2}}{(Td_T)^{1/2}}\Big),
\end{equation*}
where the last result follows from assumption (MD) and we recall the definition $a_T = \frac{(\log T)^{1/2}}{(Td_T)^{1/2}}$. 

The second statement follows from the Taylor expansion
\begin{equation*}
\widehat{f}_{it} - f_{it} = \widehat{e}_{it}^{-1} -e_{it}^{-1} = \frac{e_{it} - \widehat{e}_{it}}{e_{it}^2} + O\Big(|\widehat{e}_{it} - e_{it}|^2\Big),
\end{equation*}
where the remainder term is uniform in $i,t$ since under (A2)
\begin{align}
\inf_{i,t} e_{i,t} &=  \nonumber
\inf_{i,t} \frac{q_{i,\tau + d_T}(\Z_{it}) - q_{i,\tau - d_T}(\Z_{it})}{2d_T} 
\geq \inf_{i,t}\inf_{|\eta - \tau| 
\leq d_T} \frac{1}{f_{Y\mid \Z}(q_{i,\eta}(\Z_{it})\mid \Z_{it})}
\\
&= \frac{1}{\sup_{i,t}\sup_{\eta, \Z}f_{Y\mid\Z}(q_{i,\eta}(\Z_{it})\mid \Z_{it})} \geq 1/f_{max}, \label{eitlowb}
\end{align}
almost surely. Finally, the bound in~\eqref{psiit} follows by similar arguments as the bound for $\widehat{e}_{it} - e_{it}$ and we omit the details.
\end{proof}

\begin{proof}[Proof of Lemma~\ref{lem:aitbit}]
First, note that 
\begin{align*}
&\widehat{B}_{iT} - B_{iT} = \frac{1}{T}\sum_{t=1}^{T} (\widehat{f}_{it} -f_{it})\Z_{it}\Z_{it}^\top  
\\
&=  - \frac{1}{T}\sum_{t=1}^{T} \frac{\Z_{it}^\top[\widehat{\bm{\gamma}}_{i}(\tau + d_T)- \bm{\gamma}_{i0}(\tau + d_T) - \{ \widehat{\bm{\gamma}}_{i}(\tau-d_T) - \bm{\gamma}_{i0}(\tau-d_T)\}]}{2e_{it}^2d_T}  (\Z_{it}\Z_{it}^\top)
+  R_{1i}
\\
& = \frac{1}{T}\sum_{t=1}^{T} \sum_{k=1}^{p+1} \frac{\psi_{i,\tau+d_T,k}(\Z_{it}, Y_{it}) - \psi_{i,\tau - d_T,k}(\Z_{it}, Y_{it})}{2d_T} \E\Big[\frac{(B_i^{-1}\Z_{i1})_k }{e_{i1}^2} \Z_{i1}\Z_{i1}^\top\Big]
+ R_{2i},
\end{align*}
where by Lemma~\ref{lemma2} and Assumption (A1)
\begin{equation*}
\sup_{i} \|R_{1i} \| = O_p(a_T^2).
\end{equation*}
Next, observe that \eqref{Bit2} follows from the uniform boundedness of $B_i^{-1}$ (by (A1), (A3)), the uniform boundedness of $\E\Big[\frac{(B_i^{-1}\Z_{i1})_k \Z_{i1}\Z_{i1}^\top}{e_{i1}^2} \Big]$ which follows by Assumption (A1) and~\eqref{eitlowb}, and~\eqref{psiit}. Now we will handle $R_{2i}$. First define the matrix
\begin{equation*}
N_{it,k} :=  \frac{(B_i^{-1}\Z_{i1})_k \Z_{i1}\Z_{i1}^\top}{e_{i1}^2}  - \E\Big[\frac{(B_i^{-1}\Z_{i1})_k \Z_{i1}\Z_{i1}^\top}{e_{i1}^2} \Big],
\end{equation*}
then $\E[N_{it,k}] = 0$ and, denoting by $N_{it,k,j,\ell}$ the $(j,\ell)$-th entry of $N_{it,k}$,  
by \eqref{eitlowb} and Assumption (A1) $\sup_{i,t,k,j,\ell}|N_{it,k,j,\ell}| \leq \mathcal{C}_5$ and $\sup_{i,t,k,j,\ell}Var[N_{it,k,j,\ell}] \leq \mathcal{C}_6$ for some constants $\mathcal{C}_5$, $\mathcal{C}_6 > 0$. Apply the Bernstein inequality to obtain 
\begin{align*}
\P\Big(\Big|\sum_{t=1}^{T} N_{it,k,j,\ell}\Big| > T \epsilon_2\Big) & \leq 2\exp \Big(-\frac{\frac{1}{2}T^2 \epsilon_2^2}{\sum_{t=1}^{T} \E[N_{it,k,j,\ell}^2] + \frac{1}{3}\mathcal{C}_5 T \epsilon_2}\Big)\\
& \leq 2\exp\Big(-\frac{\frac{1}{2}T^2 \epsilon_2^2}{T\mathcal{C}_6 + \frac{1}{3}\mathcal{C}_5 T\epsilon_2}\Big).
\end{align*}
Take $\epsilon_2 = \mathcal{C}_7 T^{-1/2}(\log T)^{1/2}$ for some constant $\mathcal{C}_7>0$ to be determined later and the right hand side of the inequality becomes, for $T$ sufficiently large,
\begin{equation*}
2\exp\Big(-\frac{1}{2} \frac{\mathcal{C}_7^2\log T}{\mathcal{C}_6 +\frac{1}{3}\mathcal{C}_5 T^{-1/2}(\log T)^{1/2}}\Big)\leq 2\exp\Big(-\frac{\mathcal{C}_7^2\log T}{4\mathcal{C}_6} \Big).
\end{equation*}
Now under Assumption (I) $\log n \leq \mathcal{C}_4 \log T$. Choose $\mathcal{C}_7^2 > 4\mathcal{C}_6\mathcal{C}_4$ to obtain for every $j,k,\ell$, 
\begin{equation*}
\P\Big(\sup_i \Big|\frac{1}{T}\sum_{t=1}^{T} N_{it,k,j,\ell}\Big|> \epsilon_2) \leq \sum_{i=1}^{n} \P\Big(\Big|\frac{1}{T}\sum_{t=1}^{T} N_{it,k,j,\ell}\Big|>\epsilon_2\Big) = 2n \exp\Big(-\frac{\mathcal{C}_7^2\log T}{4\mathcal{C}_6} \Big)\to 0,
\end{equation*}
and thus $\sup_{i,k} \Big\| \frac{1}{T}\sum_{t=1}^{T} N_{it,k}\Big\| = O_p\Big((\log T)^{1/2}/T^{1/2}\Big)$. 

Now using the notations used in Lemma \ref{VCClemma} we have 
\begin{align*}
R_{2i}  =&~~ R_{1i} + \sum_{k=1}^{p+1} \Big( \frac{1}{T}\sum_{s=1}^T \frac{\psi_{i,\tau+d_T,k}(\Z_{is}, Y_{is}) - \psi_{i,\tau - d_T,k}(\Z_{is}, Y_{is})}{2d_T}\Big)  \Big(\frac{1}{T}\sum_{t=1}^{T} N_{it,k}\Big)\\
&~~ - \sum_{k=1}^{p+1} \Big (\frac{R_{iT,k}^{(1)}(\tau + d_T) + R_{iT,k}^{(2)}(\tau + d_T)}{2d_T}\Big) \Big ( \frac{1}{T} \sum_{t=1}^T Z_{it,k} \Z_{it}\Z_{it}^\top\Big)\\
&~~ + \sum_{k=1}^{p+1} \Big (\frac{R_{iT,k}^{(1)}(\tau - d_T) + R_{iT,k}^{(2)}(\tau - d_T)}{2d_T}\Big) \Big( \frac{1}{T} \sum_{t=1}^T Z_{it,k} \Z_{it}\Z_{it}^\top\Big). 
\end{align*}
Since $\sup_{i}\sup_{ \eta\in \T} (\|R_{iT}^{(1)}(\eta)\| + \|R_{iT}^{(2)}(\eta)\|)= O_p\Big(\Big(\frac{\log T}{T}\Big)^{3/4}\Big)$ by Lemma \ref{VCClemma}, $\sup_{i,k} \Big\| \frac{1}{T} \sum_{t=1}^T Z_{it,k} \Z_{it}\Z_{it}^\top\Big\| = O_p(1)$, and $\sup_i \|R_{1i}\| = O_p(a_n^2)$ as proved in Lemma \ref{lemma2}, and using~\eqref{psiit} combined with the bound on $\sup_{i,k} \Big\| \frac{1}{T}\sum_{t=1}^{T} N_{it,k}\Big\|$ derived above we obtain 
\begin{equation*}
\sup_i \|R_{2i}\| = O_p\Big(a_T^2 + a_T \Big(\frac{\log T}{T}\Big)^{1/2} + \frac{(\log T)^{3/4}}{d_T T^{3/4}}\Big).
\end{equation*}
The statement in~\eqref{Ait1} follows by an application of the Bernstein inequality which is similar to the one given above. Finally, for~\eqref{Bit1} note that by (A1) and (A2) $\sup_i \|\E[B_{iT}]^{-1}\| < \infty$ and 
	\[
{B}_{iT}^{-1} - \E[B_{iT}]^{-1} = - \E[B_{iT}]^{-1}({B}_{iT} - \E[B_{iT}])\E[B_{iT}]^{-1} + O\Big(\|{B}_{iT} - \E[B_{iT}]\|^2\Big)
	\] 
	and
\begin{align*}
\sup_i \|B_{iT} - \E[B_{iT}]\| = O_p\Big(\frac{(\log T)^{1/2}}{T^{1/2}}\Big)
\end{align*}
by an application of the Bernstein inequality which is similar to the one given above. 
\end{proof}

\subsection{Proofs of technical results related to section~\ref{sec:dep}}

\begin{proof}[Proof of Lemma~\ref{lem:pisiit_beta}]
By definition, we have for $U_{it} := F_{Y|\Z}(Y_{it}|\Z_{it})\sim U[0,1]$ independent of $\Z_{it}$, 
\begin{align*}
\frac{1}{T}\sum_{t=1}^{T} \frac{\psi_{i,\tau+d_T}(\Z_{it}, Y_{it}) - \psi_{i,\tau - d_T}(\Z_{it}, Y_{it})}{2d_T} & = \frac{1}{T}\sum_{t} \frac{\Z_{it} \Big( \1 (U_{it} \leq \tau + d_T) - \1(U_{it} \leq \tau - d_T)-2d_T\Big)}{2d_T}.
\end{align*}
Consider the function $f_{k,T}(u,z) := z_k\{ \1 (u \leq \tau + d_T) - \1(u \leq \tau - d_T)-2d_T\}$. Under Assumptions (A1), (A2) there exist constants $C_1,C_2$ independent of $i,T$ such that $\sup_{i,t} \sup_k |f_k(Y_{it},\Z_{it})| \leq C_1$, $\E[f_k(Y_{it},\Z_{it})] = 0$ and $Var(f_k(Y_{it},\Z_{it})) \leq C_2d_T$. Applying Lemma~\ref{varb} we obtain
\begin{equation*}
\sigma_q^2(f) := Var\Big(\frac{1}{\sqrt{q}}\sum_{t=1}^{q} f_k(Y_{it},\Z_{it})\Big) \leq \C_{10}d_T|\log d_T|.
\end{equation*}
Now applying the Bernstein inequality for $\beta$-mixing sequences (Corollary C.1 in \cite{KatoGalvaoMontes-Rojas12}) we have for some constant $\C$ independent of $i,T,k$ and  $q_T \in [1, T/2]$ and for all $s_T > 0$ 
\begin{equation*}
\P\Big(\Big|\frac{1}{T}\sum_{t=1}^{T} f_k(Y_{it},\Z_{it}) \Big| \geq \C \Big\{ \frac{\sqrt{(s_T \vee 1)}}{\sqrt{T}} \sigma_q(f) +  \frac{s_Tq_T}{T}\Big\}\Big) \leq 2\exp(-s_T) + 2T\beta(q_T).
\end{equation*}
Pick $s_T = \C_{11} \log T$ for some constant $\C_{11}>0$ to be determined later and $q_T = T^{c}$ for $c < \frac{1}{2}$. By assumption $\log n \leq \C_4 \log T$, so choosing $\C_{11}>\C_4$ we obtain
\begin{equation*}
\sup_i \Big\| \frac{1}{T}\sum_{t=1}^{T} \frac{\psi_{i,\tau+d_T}(\Z_{it}, Y_{it}) - \psi_{i,\tau - d_T}(\Z_{it}, Y_{it})}{2d_T} \Big\| = O_p\Big( \frac{\log T}{(Td_T)^{1/2} }\Big).
\end{equation*}


\end{proof}

\bigskip

\begin{proof}[Proof of Lemma~\ref{lemma2_new}]
By definition, 
\begin{equation*}
\widehat{e}_{is}-e_{is} = \Z_{is}^\top \Big( \widehat{\bm{\gamma}}_{i}(\tau+d_T)- \bm{\gamma}_{i0}(\tau+d_T)-(\widehat{\bm{\gamma}}_{i}(\tau -d_T) - \bm{\gamma}_{i0}(\tau - d_T))\Big) /2d_T.
\end{equation*}
By Lemma~\ref{VCClemma_beta}, we obtain 
\begin{equation*}
\Z_{is}^\top (\widehat{\bm{\gamma}}_{i}(\tau \pm d_T)- \bm{\gamma}_{i0}(\tau \pm d_T)) = -\frac{1}{T}\Z_{is}^\top B_i^{-1}\sum_{t=1}^{T}  \Z_{it}\Big(\1(Y_{it} \leq q_{i, \tau \pm d_T}(\Z_{it}))-(\tau \pm d_T)\Big) +O_p\Big( \frac{(\log T)^{5/4}}{T^{3/4}} \Big). 
\end{equation*}
Hence with $U_{it} := F_{Y\mid \Z}(Y_{it}\mid \Z_{it}) \sim U[0,1]$ independent of $\Z_{it}$, 
\begin{equation*}
\widehat{e}_{is} -e_{is} = -\frac{1}{2Td_T}\Z_{is}^\top B_i^{-1}\sum_{t=1}^{T}  \Z_{it}\Big(\1(U_{it}\leq \tau + d_T) -\1(U_{it}\leq \tau - d_T)-2d_T\Big)+O_p\Big( \frac{(\log T)^{5/4}}{d_T T^{3/4}} \Big). 
\end{equation*}
By Lemma~\ref{lem:pisiit_beta} we have 
\begin{equation*}
\sup_i \Big\|\frac{1}{T}\sum_{t=1}^{T} \frac{ \Z_{it}\Big(\1(U_{it}\leq \tau + d_T) -\1(U_{it}\leq \tau - d_T)-2d_T\Big)}{2d_T}\Big\| = O_p\Big(\frac{\log T}{(Td_T)^{1/2}}\Big).
\end{equation*}
Putting together we obtain 
\begin{equation*}
\sup_{i,s}|\widehat{e}_{is}-e_{is}| = O_p\Big(\frac{\log T}{(Td_T)^{1/2}}\Big) + O_p\Big( \frac{(\log T)^{5/4}}{d_TT^{3/4}} \Big) = O_p\Big(\frac{\log T}{(Td_T)^{1/2}}\Big),
\end{equation*}
where the last equality follows under (MD). The rest of the proof follows by the same arguments as in Lemma~\ref{lemma2}. \hfill 
\end{proof}

\bigskip

\begin{proof}[Proof of Lemma~\ref{lem:ait_beta}]
Consider the following decomposition
\begin{align*}
&\widetilde{A}_{iT} - A_{iT}
\\
&=  \sum_{1\leq j \leq m_T}\Big(1-\frac{j}{T}\Big) \Big( \frac{1}{T}\sum_{t\in T_j}(\widehat{w}_{it} \widehat{w}_{it+j}^\top + \widehat{w}_{it+j} \widehat{w}_{it}^\top) - \frac{1}{T}\sum_{t\in T_j}( w_{it}  w_{it+j}^\top +  w_{it+j}  w_{it}^\top)\Big) 
\\
&= \sum_{1\leq j \leq m_T}\Big(1-\frac{j}{T}\Big) \frac{1}{T}\sum_{t\in T_j}\Big((\widehat{w}_{it} \widehat{w}_{it+j}^\top + \widehat{w}_{it+j} \widehat{w}_{it}^\top) - ( w_{it}  w_{it+j}^\top +  w_{it+j}  w_{it}^\top) - \{\mu_{i,j}(\widehat{\bm{\gamma}}_{i}(\tau)) - \mu_{i,j}(\bm{\gamma}_i(\tau)) \}\Big)
\\
& \quad\quad + \sum_{1\leq j \leq m_T} \Big(1-\frac{j}{T}\Big) \frac{|T_j|}{T} \Big(\mu_{i,j}(\widehat{\bm{\gamma}}_{i}(\tau)) - \mu_{i,j}(\bm{\gamma}_{i0}(\tau)) - \sum_{k=1}^{p+1} D_{i,j,k}(\tau) \frac{1}{T}\sum_{t=1}^{T} (B_{i}^{-1}\psi_{i,\tau}(\Z_{it}, Y_{it}))_k\Big)
\\
& \quad\quad + \sum_{1\leq j \leq m_T}\Big(1-\frac{j}{T}\Big)\frac{|T_j|}{T}\sum_{k=1}^{p+1} D_{i,j,k}(\tau) \frac{1}{T}\sum_{t=1}^{T} (B_{i}^{-1}\psi_{i,\tau}(\Z_{it}, Y_{it}))_k
\\
&=:   R_{4i}^{(1)} + R_{4i}^{(2)} + \sum_{1\leq j \leq m_T}\Big(1-\frac{j}{T}\Big)\frac{|T_j|}{T}\sum_{k=1}^{p+1} D_{i,j,k}(\tau) \frac{1}{T}\sum_{t=1}^{T} (B_{i}^{-1}\psi_{i,\tau}(\Z_{it}, Y_{it}))_k.
\end{align*}
Regarding $R_{4i}^{(2)}$, begin by observing that under (A1), (D2) by~\eqref{muij}
\begin{align*}
& \sup_{i,j}\Big\|\mu_{i,j}(\widehat{\bm{\gamma}}_{i}(\tau)) - \mu_{i,j}(\bm{\gamma}_{i0}(\tau)) - \sum_{k=1}^{p+1} D_{i,j,k}(\tau) \frac{1}{T}\sum_{t=1}^{T} (B_{i}^{-1}\psi_{i,\tau}(\Z_{it}, Y_{it}))_k \Big\|
\\
\leq & \sup_{i,j}\Big\| \mu_{i,j}(\widehat{\bm{\gamma}}_{i}(\tau)) - \mu_{i,j}(\bm{\gamma}_{i0}(\tau)) - \sum_{k=1}^{p+1}D_{i,j,k}(\tau) (\widehat{\bm{\gamma}}_{i,k}(\tau) - \bm{\gamma}_{i0,k}(\tau)) \Big\|
\\
&~+ \sup_{i,j}\Big\|\sum_{k=1}^{p+1}D_{i,j,k}(\tau) (\widehat{\bm{\gamma}}_{i,k}(\tau) - \bm{\gamma}_{i0,k}(\tau)) - \sum_{k=1}^{p+1} D_{i,j,k}(\tau) \frac{1}{T}\sum_{t=1}^{T} (B_{i}^{-1}\psi_{i,\tau}(\Z_{it}, Y_{it}))_k \Big\|
\\
\leq & O\Big(\sup_{i}\| \widehat{\bm{\gamma}}_{i}(\tau) - \bm{\gamma}_{i0}(\tau)\|^2\Big) + \sup_{i,j,k} \|D_{i,j,k}(\tau)\| \sup_i \|R_{iT}^{(1)}(\tau) + R_{iT}^{(2)}(\tau)\|,
\end{align*}
where $R_{iT}^{(1)}(\tau),R_{iT}^{(2)}(\tau)$ are defined in Lemma \ref{VCClemma_beta} and hence $\sup_{i}\|R_{4i}^{(2)}\| = O_p\Big( m_T\frac{(\log T)^{5/4}}{T^{3/4}}\Big)$ by \eqref{G1_beta} and Lemma \ref{VCClemma_beta}. 

Next we deal with $R_{4i}^{(1)}$. Pick a constant $C_1$ such that the event 
\begin{equation*}
\Omega_{1,T} := \Big\{\sup_i \|\widehat{\bm{\gamma}}_{i}(\tau) - \bm{\gamma}_{i0}(\tau)\| \leq C_1(\log T)^{1/2}T^{-1/2}\Big\},
\end{equation*}
has probability tending to one, this is possible by Lemma \ref{VCClemma_beta} and \eqref{G1_beta}. Note that on $\Omega_{1,T}$ we have
\begin{equation*}
\sup_i\|R_{4i}^{(1)}\| \lesssim \sum_{1\leq j \leq m_T}\Big(1-\frac{j}{T}\Big) \sup_{k,\ell}S_{T,j,k,\ell}(C_1(\log T)^{1/2}T^{-1/2}) \leq m_T \sup_{k,\ell,j} S_{T,j,k,\ell}(C_1(\log T)^{1/2}T^{-1/2}),
\end{equation*}
where $S_{T,j,k,\ell}$ is defined in~\eqref{STjkl}. Apply~\eqref{G3_beta} and the union bound to find that the right-hand side above is $O_P\Big(m_T (\log T + m_T)^{1/2}\Big(\frac{\log T}{T}\Big)^{3/4}\Big)$. Putting together, we get \eqref{Ait3}. 

To prove~\eqref{Ait4} note that
\begin{align*}
A_{iT}-\E[A_{iT}] & = \frac{\tau(1-\tau)}{T}\sum_{t=1}^{T} (\Z_{it} \Z_{it}^\top - \E[\Z_{i1}\Z_{i1}^\top])\\
& + \sum_{1\leq j \leq m_T} \Big(1-\frac{j}{T}\Big) \Big(\frac{1}{T}\sum_{t\in T_j} \Big((w_{it} w_{it+j}^\top + w_{it+j} w_{it}^\top)- \E[w_{i1}w_{i1+j}^\top + w_{i1+j}w_{i1}^\top]\Big)\Big).
\end{align*}
The first term is $O_p\Big(\frac{(\log T)^{1/2}}{T^{1/2}}\Big)$ by a standard application of the Bernstein inequality for dependent variables. To deal with the second term, note that under (D1) the series of random vectors $\{(Y_{it}, \Z_{it}, Y_{it+j}, \Z_{it+j})\}_{t\in\mathbb{Z}}$ is $\beta$-mixing with mixing coefficients $\widetilde{\beta}(t)$ satisfying $\widetilde{\beta}(t) \leq \beta(0\vee (t-j))$. This implies that the $\beta$-mixing coefficients of the series $\{ v_{it,j}\}_{t \in \mathbb{Z}}$ where $v_{it,j} := (w_{it} w_{it+j}^\top + w_{it+j} w_{it}^\top)_{k,\ell}$
are also bounded by $\beta(0\vee (t-j))$ (here, $M_{k,\ell}$ denotes the $(k,\ell)$ entry of the matrix $M$). Hence by an application of Lemma C.1 in \cite{KatoGalvaoMontes-Rojas12} with $\delta=1$ we have $|Cov(v_{i1,j},v_{it+1,j})| \leq C \beta(0\vee (t-j))^{1/2}$ and direct computations show that 
\begin{equation*}
\sup_{1\leq q \leq T}\sup_{k,\ell} \max_{1\leq j \leq m_T} \sup_i  Var \Big(\frac{1}{\sqrt{q}} \sum_{t=1}^q (w_{it} w_{it+j}^\top + w_{it+j} w_{it}^\top)_{k,\ell} \Big) = O(m_T).
\end{equation*}
Apply Bernstein's inequality for $\beta$-mixing sequences (Corollary C.1 in \cite{KatoGalvaoMontes-Rojas12}) with $q=q_T= m_T + C\log T$, $s=s_T = C\log T$ for suitable constants $C$ and the union bound to obtain
\begin{equation*}
\sup_{i,j} \Big(\frac{1}{|T_j|}\sum_{t\in T_j} \Big((w_{it} w_{it+j}^\top + w_{it+j} w_{it}^\top)- \E[w_{i1}w_{i1+j}^\top + w_{i1+j}w_{i1}^\top]\Big)\Big) = O_P\Big(\sqrt{\frac{m_T \log T}{T}}\Big).
\end{equation*}
Combined with the definition of $A_{iT}$ this yields~\eqref{Ait4}. 

To prove \eqref{Ait4.1}, note that 
\begin{equation*}
\Big \|\widetilde{A}_{iT} - \E[A_{iT}]\Big \| \leq \Big \|\widetilde{A}_{iT}-A_{iT}\Big \| + \Big \|A_{it}-\E[A_{iT}]\Big \|,
\end{equation*}
apply the first result of Lemma~\ref{lem:G1G2} to obtain 
\begin{equation*}
\sup_i \|\widetilde{A}_{iT} - A_{iT}\| =O_p\Big(m_T\Big(\frac{\log T}{T}\Big)^{1/2}\Big),
\end{equation*}
and combined with~\eqref{Ait3}, \eqref{Ait4}, we get \eqref{Ait4.1}. 

To prove \eqref{Ait5}, apply Lemma C.1 in \cite{KatoGalvaoMontes-Rojas12} to obtain 
\begin{equation*}
\|\E[w_{i1}w_{i1+j}^\top + w_{i1+j}w_{i1}^\top]\|\leq C\beta(j)^{1/2},
\end{equation*}
for a constant independent of $i,T,j$. Hence under Assumptions (D1) and (MD), where the latter implies $m_T\to \infty, m_T/T \to 0$, we have
\begin{align*}
\|\E[A_{iT}] - \widetilde{A}_i\| & \leq \sum_{j=1}^{m_T}\Big|\Big(1-\frac{j}{T}\Big)\frac{|T_j|}{T}-1\Big| C\beta(j)^{1/2}+ \sum_{j=m_T+1}^{\infty} C\beta(j)^{1/2}\\
& \leq C_1\Big(\frac{m_T}{T}\sum_{j=1}^{m_T}  \beta(j)^{1/2} + \sum_{j=m_T+1}^{\infty} \beta(j)^{1/2}\Big)=o(1).
\end{align*}
\hfill 
\end{proof}

\bigskip

\begin{proof}[Proof of Lemma~\ref{lem:aitbit2}]
The statement \eqref{Bit2_beta} follows from Lemma~\ref{lem:pisiit_beta} after noting hat the bound $\sup_i e_{i1} > 0$ derived in \eqref{eitlowb} still holds in the present setting since the latter derivation only used (A1).

The proof of~\eqref{Bit3_beta} is very similar to the proof of~\eqref{Bit3} and uses Lemma~\ref{VCClemma_beta}, so we will only give a brief outline. First observe that
\begin{align*}
\widehat{B}_{iT} - B_{iT} &= \frac{1}{T}\sum_{t=1}^{T} (\widehat{f}_{it} -f_{it})\Z_{it}\Z_{it}^\top  
\\
&= -\frac{1}{T}\sum_{t=1}^{T} \frac{\widehat{e}_{it} - e_{it}}{e_{it}^2}\Z_{it}\Z_{it}^\top + O_p\Big(\frac{(\log T)^2}{Td_T}\Big) 
\\
& = \frac{1}{T}\sum_{t=1}^{T} \sum_{k=1}^{p+1}\frac{\psi_{i,\tau+d_T,k}(\Z_{it}, Y_{it}) - \psi_{i,\tau - d_T,k}(\Z_{it}, Y_{it})}{2d_T} \E\Big[\frac{(B_i^{-1}\Z_{i1})_k}{e_{i1}^2}(\Z_{i1}\Z_{i1}^\top) \Big]
+ R_{2i}, 
\end{align*}
where the first equality follows by Lemma~\ref{lemma2_new},
\begin{align*}
R_{2i}  := &~~ \sum_{k=1}^{p+1} \Big( \frac{1}{T}\sum_{s=1}^T \frac{\psi_{i,\tau+d_T,k}(\Z_{is}, Y_{is}) - \psi_{i,\tau - d_T,k}(\Z_{is}, Y_{is})}{2d_T}\Big)  \Big(\frac{1}{T}\sum_{t=1}^{T} N_{it,k}\Big)\\
&~~ - \sum_{k=1}^{p+1} \Big (\frac{R_{iT,k}^{(1)}(\tau + d_T) + R_{iT,k}^{(2)}(\tau + d_T)}{2d_T}\Big) \Big ( \frac{1}{T} \sum_{t=1}^T Z_{it,k} \Z_{it}\Z_{it}^\top\Big)\\
&~~ + \sum_{k=1}^{p+1} \Big (\frac{R_{iT,k}^{(1)}(\tau - d_T) + R_{iT,k}^{(2)}(\tau - d_T)}{2d_T}\Big) \Big( \frac{1}{T} \sum_{t=1}^T Z_{it,k} \Z_{it}\Z_{it}^\top\Big) + O_p\Big(\frac{(\log T)^2}{Td_T}\Big) ,
\\
N_{it,k} := &~~ \frac{(B_i^{-1}\Z_{i1})_k \Z_{i1}\Z_{i1}^\top}{e_{i1}^2}  - \E\Big[\frac{(B_i^{-1}\Z_{i1})_k \Z_{i1}\Z_{i1}^\top}{e_{i1}^2} \Big],
\end{align*}
$R_{iT}^{(1)},R_{iT}^{(2)}$ are defined in Lemma~\ref{VCClemma_beta}, and $R_{iT,k}^{(j)}$ denotes the k-th component of the vector $R_{iT}^{(j)}$. Now an application of the Bernstein inequality for $\beta$-mixing sequences (see Corollary C.1 in \cite{KatoGalvaoMontes-Rojas12}) shows that 
\begin{equation*}
\sup_i \sup_k \Big\| \frac{1}{T}\sum_{t=1}^{T} N_{it,k}\Big\| = O_p\Big((\log T)^{1/2}/T^{1/2}\Big),
\end{equation*}
and combining this with Lemma~\ref{lem:pisiit_beta} and Lemma~\ref{VCClemma_beta} completes the argument.
 
Finally,~\eqref{Bit1_beta} follows by uniform boundedness of $f_{it}$ and an application of the Bernstein inequality for $\beta$-mixing sequences (see Corollary C.1 in \cite{KatoGalvaoMontes-Rojas12}) after noting that by an application of Lemma C.1 in \cite{KatoGalvaoMontes-Rojas12} we have 
\begin{equation*}
\sup_{k,\ell}\sup_i \sup_{1\leq q \leq T} Var\Big(\frac{1}{q^{1/2}} \sum_{t=1}^q f_{it} (\Z_{it} \Z_{it}^\top)_{k,\ell} \Big) = O(1).
\end{equation*}
Details are omitted for the sake of brevity. \hfill
\end{proof}

\end{document}